\documentclass[11pt]{article}
\usepackage{amsmath}
\usepackage{graphicx}
\usepackage{enumerate}
\usepackage{arydshln} 
\usepackage[authoryear,sort,round]{natbib}
\usepackage{url} 
\usepackage{authblk}
\newcommand{\blind}{1}

 \usepackage{geometry}

\usepackage{amssymb}
\usepackage{amsthm}
\usepackage{enumitem}
\usepackage[lined,boxed,]{algorithm2e}
\usepackage{threeparttable,subcaption}
 \usepackage{endfloat}

\usepackage[colorlinks]{hyperref}
 \hypersetup{
 	colorlinks=true,
 	linkcolor=blue,
 	filecolor=blue,
 	urlcolor=cyan,
 }
 \hypersetup{colorlinks,linkcolor={blue},citecolor={blue},urlcolor={red}}

\usepackage{ateAalen-notations}

\usepackage[dvipsnames]{xcolor}
\usepackage[normalem]{ulem}

\begin{document}

\def\spacingset#1{\renewcommand{\baselinestretch}%
{#1}\small\normalsize} \spacingset{1}


\if1\blind
{
  \title{\bf Estimating Treatment Effect under
 Additive Hazards Models with High-dimensional Confounding}
    \author[1]{Jue Hou}
\author[2,3]{Jelena Bradic}
\author[2,3,4]{Ronghui Xu}
\affil[1]{Department of Biostatistics,  Harvard T.H. Chan School of Public Health}
\affil[2]{Department of Mathematics University of California San Diego }
\affil[3]{Halicio{\u{g}}lu Data Science Institute,  University of California San Diego}
\affil[4]{Department of Family Medicine and Public Health,     University of California San Diego }
      \date{ }
  \maketitle

} \fi

\if0\blind
{
  \bigskip
  \bigskip
  \bigskip
  \begin{center}
    {\LARGE\bf Estimating Treatment Effect under
 Additive Hazards Models with High-dimensional Confounding}
\end{center}
  \medskip
} \fi

\bigskip
\begin{abstract}

Estimating treatment effects for survival outcomes in the high-dimensional setting is a critical topic for many biomedical applications and any application with censored observations. In this paper, we establish an `orthogonal' score for learning treatment effects, using observational data with a potentially large number of confounders. The estimator allows for root-$n$, asymptotically valid confidence intervals, despite the bias induced by the regularization. Moreover, we develop a novel hazard difference (HDi) estimator. We establish rate double robustness through the cross-fitting implementation of the proposed estimators.  Finite sample performance is illustrated through numerical experiments, where we observe that the cross-fitted HDi estimator has the best performance of all. We apply the estimators to study the treatment effect of radical prostatectomy versus conservative management for prostate cancer using the SEER-Medicare linked data.
Lastly, we provide a discussion on extensions to general machine learning approaches as well as heterogeneous treatment effects.

\end{abstract}

\noindent%
{\it Published version:} \href{https://doi.org/10.1080/01621459.2021.1930546}{doi:10.1080/01621459.2021.1930546}
\vfill

\newpage
\spacingset{1.45} 
\section{Introduction}

Treatment effect estimation and inference is an essential topic of interest in causal inference. It has drawn immense interest, spanning many different fields. Our work was motivated by the proliferation of `big, observational data' from electronic medical/health records (EMR/EHR), which provides an abundant resource for studying the effect of various treatments and serves as an alternative when randomized trials are implausible or otherwise perceived uneconomical. The challenge in studying causal effects today, among other things, is to handle a large number of potential confounders, ``$p\gg n$". Motived by studies in cancer, using the linked Surveillance, Epidemiology, and End Results (SEER) - Medicare database, our primary focus is the causal effect on a survival outcome.

Research methods for high-dimensional data analysis of survival outcomes have primarily focused on variable selection properties -- inference in practice is often reported on the findings with only selected covariates. Meanwhile, it has been known that the direct estimation following variable selection is biased in high-dimensions, and limited work has been done to provide corrections.
 Here we consider the additive hazards model, where we define the treatment effect as a difference in the hazards between the treated and control groups. We develop a family of orthogonal scores, introduced in as early as \cite{neyman59}. Such scores allow valid statistical inference and tractable asymptotic theory in high-dimensions. We then develop several refinements and extensions, including a new hazard difference (HDi) estimator that utilizes covariate balancing. The HDi estimator is also a particular case of the introduced orthogonal score family, a member of the family linear in the unknown parameter. When coupled with a cross-fitting procedure, we show that the HDi estimator attains rate double robustness in high dimensions, as defined in \cite{SmuclerRotnitzkyRobins19}.

\subsection{Related work}

The first main technical contribution of our work is an asymptotic normality result enabling statistical inference in high-dimensional models for right-censored survival outcomes. Recent results of \cite{BradicEtal11}, \cite{YuEtal2018}, and \cite{HouEtal19} provide asymptotic properties under the high dimensional Cox type proportional hazards models, which include competing risks. To the best of our knowledge, however, we provide the first results on high-dimensional inference, including confidence interval construction, under the additive hazards models.

The orthogonal score has been a familiar concept in the semiparametric literature \citep{Newey90, BKRW}, and is closely related to efficient scores and efficient influence functions.   It is also related to the profile likelihood and the least favorable direction, including the nonparametric maximum likelihood approach, often used for semiparametric models in survival analysis \citep{seve:wong:92,murp:vand:00,vdVaartbook00}.
 However, the nonparametric likelihood does not apply to the additive hazards model \citep{LinYing94}.

The benefits of the orthogonal score have long been known: an estimator obtained from such a score should not be affected by the slower than root-$n$  estimation of the model' nuisance parameters \citep{Newey90,BKRW}. This property was recently utilized for purposes of estimating treatment effects in high-dimensional models \citep{BelloniEtal13,Farrell15,ChernozhukovEtal17,ChernozhukovDML-RR,ChernozhukovEtal-l2RR}. However, models with censored observations present a considerable challenge. Approaches based on uncensored data do not automatically generalize; complex dependencies are induced by the nested, `risk sets,' over time.


There is a connection between an orthogonal score and double robustness. An estimator is doubly robust (DR) if it is consistent (and asymptotically normal), as long as one (but not necessarily both) of the outcome and the treatment assignment models are correct.
However, in the high-dimensional setting, this classical definition does not capture the full complexity of model misspecification.
There are many related terms used in the literature, often with partially overlapping meanings, such as `doubly robust', `locally robust', `small bias', `mixed bias', `model/sparsity double robust'; see \cite{Chernozhukov16LR,ChernozhukovEtal17,Tan17,Tan18,ZhuBradic18,ZhuBradic18ejs,BradicEtal19} for discussion of these issues. In particular,
the `rate double robust' as defined in \cite{SmuclerRotnitzkyRobins19}
is directly relevant to our approach.
Rate double robustness states that a root-$n$ asymptotically normal estimator exists as long as one of the two correctly specified models is sparse enough. We will provide more details in Section 3 below.

For survival outcomes, DR estimators have only been considered in a low-dimensional setting, with a fixed $p < n$. \cite{LuEtal16} considered the DR approach under the additive hazards model
but relied on kernel density estimators which are not suitable for high-dimensional covariates. \cite{Dukes:etal:19}  study the same model through semiparametric efficiency theory and derive a family of DR score functions, albeit under the low-dimensional settings, $p <n$.

\subsection{Notation}

We consider right-censored survival outcomes,  where   $T$  and $C$ denote   the event  and  the censoring time, respectively.  Our observations consist of a sample of independent and identically distributed (i.i.d.) data points $(X_i,\delta_i, D_i, \bZ_i)$, $i=1, ..., n$, where $X_i = \min\{T_i,C_i\}$, $\delta_i = I(T_i \le C_i)$, $D_i \in \{ 0.1\}$ is the treatment assignment, and $\bZ_i \in \mathbb{R}^p$ denotes a $p \times 1$ vector of covariates.
 We also denote ${\tilde Z}_i= (1,Z_i^\top)^\top$ when an intercept term is needed in a regression model.
Denote   the counting process and at-risk process  with $N_i(t)=\delta_iI(X_i\le t)$ and $Y_i(t)= I(X_i\ge t)$, respectively,
with filtration $\Ftn = \sigma\{N_i(u),Y_i(u),D_i,\bZ_i: u\le t, i=1,\dots,n\}$.
  Denote $\haz_i(t;D_i,\bZ_i) = \lim_{\dt \to 0} \P(T_i<t+\dt |T_i\ge t, D_i,\bZ_i)/\dt$ the hazard function for subject $i$.
 To avoid degenerate cases, we consider absolutely continuous random variables.
In the models to be specified below, $ \lambda_0(\cdot)$ is an unknown baseline hazard function and  $\Lambda(t)$ is the corresponding baseline cumulative hazard.
We consider finite study duration and denote the upper limit of follow-up time as $\tau < \infty$,
 and will use  `$ \indep $' to denote statistical independence.
For the model parameters to be estimated, a subscript `0' denotes the true value of a parameter that generates the data; for example, $\theta_0$, $\beta_0$, $\Lambda_0$, $\gamma_0$, etc.

\subsection{Organization}

The organization of the rest of the paper is as follows. In Section \ref{section:inf}, we propose a family of inferential methods based on orthogonal scores for the treatment effect under the additive hazards model for the survival outcome and the logistic regression model for treatment assignment. In Section \ref{section:ext}, we proposed a hazard difference (HDi) estimator and established the rate double robustness. Section \ref{section:simulation} contains extensive simulation studies illustrating favorable finite-sample properties of the newly proposed estimates across several settings in high-dimensions. In Section \ref{section:data}, we apply our methods to the study of the treatment effect of radical prostatectomy versus conservative management of prostate cancer patients aged 65 or older, using the SEER-Medicare linked data. Section \ref{section:new} expands our approach to operate with generic machine-learning estimators as well as that of heterogeneous treatment settings. Lastly, Section \ref{section:discuss} contains  concluding remarks. The Supplementary Materials contain detailed proofs of all the theoretical results.

\section{Inferential methods and guarantees}\label{section:inf}

We are interested in the effect of  treatment $D$, on the survival time $T$ (subject to censoring by $C$), conditional on covariates $Z$. 
For binary treatments, this conditional treatment effect can be seen as the difference in hazards between the treated and the control groups, conditional on  $Z$.
Under the well-known semiparametric additive hazards model \citep{LinYing94}
\begin{equation}\label{model:aalen}
\haz_i(t;D_i,\bZ_i)= \lambda_0(t) + \theta D_i + \bbeta^\top\bZ_i,
\end{equation}
where $\bbeta \in \mathbb{R}^p$, where $p$ can be   much larger than the sample size $n$.
More generally we might write:
\begin{equation}\label{def:ate}
\theta(t; \bZ) = \haz(t; D=1,\bZ)- \haz(t; D=0,\bZ).
\end{equation}
Under model \eqref{model:aalen} we  first focus on homogeneous treatment effects so that $\theta(t; \bZ)  $ does not depend on $\bZ$.  Extensions to heterogeneous treatment effects  are considered in Section \ref{section:ite}.  
Other extensions beyond model \eqref{model:aalen}  are discussed in Section \ref{section:ml}.
Our goal here is to draw inference, construct confidence sets, for the treatment effect  $\theta$, in the presence of high-dimensional nuisance parameter $\bbeta$ as well as that of the baseline hazard $\lambda_0(t) $.

\subsection{Orthogonal Score}\label{section:score}

An initial look at the high-dimensional model \eqref{model:aalen} may lead one to think of regularization or other machine learning methods. However, they are known to give a biased estimate of $\theta$ even if $\theta$ itself is not penalized in the regularization process; e.g., see  \cite{BelloniEtal13} which illustrates this point. The effects of regularization propagate, create bias, and prevent root-$n$ inference. As it turns out, a model for the treatment selection mechanism can often be leveraged to remove the shrinkage bias saliently. Orthogonal scores were proven useful to this effect; see, e.g., \citep{ChernozhukovEtal17}. We are primarily interested in binary treatments. Therefore, we assume a logistic regression as a working model for the binary treatment assignment
\begin{equation}\label{model:D}
  \P(D_i=1|\bZ_i) 
  =
  \frac
  {\exp(\bgr^\top{\tilde Z}_i)} {   1+\exp(\bgr^\top{\tilde Z}_i ) }
 : = \expit(\bgr^\top{\tilde Z}_i ),
\end{equation}
where ${\tilde Z}_i= (1,Z_i^\top)^\top$. 

 Orthogonal scores have been highly effective in the presence of very many nuisance parameters;  the estimation of the treatment effect is  then hopefully not greatly affected by
the estimation error of the nuisance parameters
 \citep{Newey94}.
Because of this, orthogonal scores have been   useful in high-dimensional inference.
The orthogonality of a score function is defined as the local invariance of the score to
a small perturbation in the nuisance parameter space.
 Under models \eqref{model:aalen} and \eqref{model:D},  we denote the nuisance parameter as
 $$\eta = (\bbeta,\Haz,\bgr)$$
 where $\Haz$ is the cumulative baseline hazard.
 A score function  $\psi (\theta, \eta)$ is  an orthogonal score for $\theta$,  if the G\^{a}teaux derivative with respect to $\eta$ is zero; that is,
\begin{equation} \label{eq:ortho}
 \left.\frac{\partial}{\partial r}\E\{ \psi (\theta_0; \eta_0+r\deta) \} \right|_{r=0} = 0,
\end{equation}
where $\theta_0$ and $\eta_0 = (\bbeta_0, \Haz_0,\bgr_0)$ are the true parameter values, respectively, and $ \deta= \eta - \eta_0 $.

A common approach for obtaining orthogonal scores,  at least under the linear outcome models,  is to consider the product of  two residuals: one  from the outcome model and one  from the  treatment model  \citep{RobinsRotnitzky95}.
Under our models two natural candidates 
would be  the martingale residual
\begin{equation}\label{def:mart}
  M_i(t; \theta, \bbeta, \Haz) = N_i (t)  -   \int_0^t  \left( D_i   \theta  +  \bbeta^\top\bZ_i    \right) Y_i(u)  du -   Y_i(u) d\Haz(u),
\end{equation}
and the logistic regression residual
 $D_i - \expit(\bbeta^\top Z_i)$.
Under the additive hazards model \eqref{model:aalen},
 the martingale residual evaluated at the true parameter values,
 $M_i(t; \theta_0, \bbeta_0, \Haz_0)$,
   is a martingale with respect to the filtration $\Ftn$.

 However,  it is not difficult to verify that the product  of the above two residuals would  not give an orthogonal score according to the definition \eqref{eq:ortho}, ultimately due to the dependence between the at-risk process  and the treatment assignment.
In the following we show that, if we are willing to assume
\begin{equation}\label{eq:CindD}
 C\indep  (T,  D)  | \bZ,
 \end{equation}
we can make a simple correction, by `decoupling' the at-risk process  and the treatment assignment. 
For a single copy of the data, let
\begin{equation}\label{score:inference}
  \scorei(\theta; \bbeta,\Haz,\bgr) = \int_0^\tau \exp{(D \theta t)}  \left (D -\expit(\bgr^\top \bZI[])\right )d M(t;\theta, \bbeta, \Haz).
\end{equation}
\begin{lemma}\label{lem:score}
Under models \eqref{model:aalen} and \eqref{model:D},
and  assumption  \eqref{eq:CindD},
the score \eqref{score:inference}
identifies  the true parameters $(\theta_0; \bbeta_0,\Haz_0,\bgr_0)$, i.e.,
$
\mathbb{E} [ \scorei_i (\theta_0; \bbeta_0,\Haz_0,\bgr_0) ] =0.$
Moreover,  $\scorei_i$ is an orthogonal score for $\theta$ in the sense of
\eqref{eq:ortho}.
\end{lemma}
Using score \eqref{score:inference} we obtain $\hat \theta$ by solving $n^{-1}\sum_{i=1}^n \scorei_i(\theta; \hat \bbeta, \hat \Haz, \hat \bgr)=0$ for $\theta$. Note that
the Neyman orthogonality enables us to
 plug in   some initial estimates of the nuisance parameters $(\hat \bbeta,\hat \Haz, \hat \bgr)$.

While  \eqref{eq:CindD} is stronger than  the usual non-informative censoring, $C \indep T  | (D, \bZ)$, we observe that  under the two models \eqref{model:aalen} and \eqref{model:D},
the pair $(T,D)$
 plays the role of  `response' in relation to $Z$.
 From this perspective, the assumption \eqref{eq:CindD} is in line with the typical non-informative censoring assumption.
 Nonetheless, we are also able to decouple $Y_i(t)$ and $D$ under the weaker condition $C \indep  \  T  \ |  \ (D, \bZ)$,  and  similarly construct an orthogonal score, as stated in the following lemma.
 \begin{lemma}\label{lemma:score-Sc}
 Under models \eqref{model:aalen} and \eqref{model:D},
and  whenever
  $C \indep T  | (D, \bZ)$,
the score
\begin{equation*}\label{score:inference00}
  \scorei_{C}(\theta; \bbeta,\Haz,\bgr,S_C) = \int_0^\tau \exp{(D\theta t)}   S_C^{-1}(t| D, \bZ)\left (D-\expit(\bgr^\top \bZI[])\right )d M(t;\theta, \bbeta, \Haz)
\end{equation*}
identifies the true parameters $(\theta_0; \bbeta_0,\Haz_0,\bgr_0, S_{C,0})$,  i.e.,
$
\mathbb{E} [ \scorei_{C,i} (\theta_0; \bbeta_0,\Haz_0,\bgr_0,S_{C,0}) ] =0.$ Here $S_{C,0}$ is the true value of the additional nuisance parameter $ S_C(t; d, z) = \mathbb{P}(C \geq t| D=d, \bZ=z)$. Moreover,
 $ \scorei_{C}$ is an orthogonal score for $\theta$ in the sense of
\eqref{eq:ortho}.
\end{lemma}

Note that a practical implication of Lemma \ref{score:inference00} is the need to estimate  $ S_C(t | D,Z) $, at least consistently. Nonparametric estimation of $S_C(t | D, Z) $ in the presence of high-dimensional $Z$ might be attempted using machine learning methods as described in Section 6.1 later, but this would lead to the (interesting) problem of multiple robustness, and is not further pursued in this paper. From here on, throughout the paper, we continue by assuming \eqref{eq:CindD}. Note that \eqref{eq:CindD} is easily satisfied in case of administrative censoring, i.e., caused by the end of a study.

\subsection{From orthogonality to double robustness}\label{section:score-dr}


In classical low dimensional models, there has often been a close connection between local efficiency and double robustness; this was explicitly discussed in a fundamental work of \cite{RobinsRotnitzky01}. There, local efficiency is achieved using the efficient influence function or efficient score, which is a particular case of the orthogonal score. For high-dimensional models, however, only recently did  \cite{ChernozhukovDML-RR,ChernozhukovEtal-l2RR} establish such a connection in the case of linear outcome models. For survival outcomes, such connections have not even been studied.  The following lemma establishes double robust property of the proposed orthogonal score in low-dimensions.  Our orthogonal score matches that of \cite{Dukes:etal:19} obtained through semiparametric efficiency. 


\begin{lemma}\label{lem:doubly_robust}
\begin{enumerate}[label=(\alph*),ref=\ref{lem:doubly_robust}(\alph*)]
  \item \label{lem:Drobust} Suppose that 
  model \eqref{model:aalen}  holds, whereas the
      treatment assignment $D$ follows a nonparametric model
      \begin{equation}\label{model:Dgen}
        \P(D=1|\bZ) = m(\bZ).
      \end{equation}
      Then, for any  given $\bgr^*$,
      $\theta=\theta_0$ is the root of the equation
      \begin{equation*}\label{eq:Drobust}
      \E[\scorei(\theta;\bbeta_0,\Lambda_0,\bgr^*)]=0.
      \end{equation*}
  \item \label{lem:Trobust} Suppose that 
  model \eqref{model:D} holds, whereas the outcome
      $T$ follows a partially linear, additive hazards model 
      \begin{equation}\label{model:aalen_pl}
        \haz(t;D,\bZ) = D\theta + g(t;\bZ),
      \end{equation}
      with $g$ an unspecified function of both time and covariates.
        Under assumption  \eqref{eq:CindD}, for any given $\bbeta^*$ and $\Lambda^*$ with bounded total variation,
      $\theta=\theta_0$ is the root of the equation
      \begin{equation*}\label{eq:Trobust}
      \E[\scorei(\theta;\bbeta^*,\Lambda^*,\bgr_0)]=0.
      \end{equation*}
\end{enumerate}
\end{lemma}
While typical double robustness might refer to a nonparametric outcome model in part (b) above, since our treatment effect $\theta$ implies a constant hazard difference over time, a partially linear model \eqref{model:aalen_pl} is necessary, i.e.~the hazard function for the control group is a nonparametric function $ g_0(t;\bZ) $.

In the Section \ref{section:ext}, we will study the double robust properties of score function $\phi$ in high-dimensional settings.

\subsection{One-shot  inference for treatment effect  }\label{section:inf-phi}


In this subsection we present   one-shot inference  results, where all of the data, rather than subsamples, are utilized twice: first to estimate the unknown nuisance parameters and then to solve the score equation below.  Define $\hat \theta$ as the solution to
\begin{equation}\label{eq:thetahat}
 \frac{1}{n}\sum_{i=1}^{n}\scorei_i(\theta; \hat \beta, \hat \Lambda, \hat \gamma) =0.
\end{equation}
In the above,  $\hat \beta, \hat \Lambda, \hat \gamma$   are only required to satisfy the conditions specified below.
In particular, we allow the baseline hazard estimator $\hHaz(t;\theta)$ to be a
function of $\theta$ itself.
By allowing the nuisance parameter or its estimator to  depend on $\theta$,
we have generalized the existing approaches that assume the nuisance parameter
can be estimated without input of $\theta$ \citep{ChernozhukovEtal17}.
While the Neyman orthogonality allows $\hHaz$ here to depend on some initial estimate of $\theta$, embedding the unknown $\theta$ in $\hHaz(t;\theta)$ is novel to our best knowledge.

 Before stating the assumptions, we define the following measures of
 estimation error under models
\eqref{model:aalen} and \eqref{model:D} respectively:
\begin{align*}
  \msenb^2(\hbeta, \bbeta_0) &=
  {n^{-1}\sum_{i=1}^n \int_0^\tau \left\{(\hbeta-\bbeta_0)^{\top}\bZ_i \right\}^2Y_i(t)dt } ,
  \\
  \mseng^2(\hgr, \bgr_0)&=
 { n^{-1}\sum_{i=1}^n  \left\{\expit\left(\hgr^{\top}\bZ_i\right)
  -\expit\left(\bgr_0^{\top}\bZ_i\right)\right\}^2 }.
\end{align*}
Note that $ \msenb^2(\hbeta, \bbeta_0)$ is the symmetrized Bregman divergence used in \cite{GaiffasGuilloux12},
and $\mseng^2(\hgr, \bgr_0)$ is the excess risk used  in \cite{vdGeer08}.


\begin{assumption}\label{assume:inf}
For constants $K_Z$, $\varepsilon$, 
$K_\theta$, $K_\Lambda$ and $K_v$ independent of $p$ and $n$, we assume

\begin{enumerate}[label = (\roman*), ref = \ref{assume:inf}-\roman*]
  \item \label{assume:Z} bounded covariates: $\P\left(\sup_{i=1,\dots,n}\|\bZ_i\|_\infty < K_Z\right) = 1$;

  \item \label{assume:varD} positivity of propensity, at risk and event
   rates:
  $\Var(D|\bZ) \ge \varepsilon>0$, 
 $E\{Y(t)\} \ge \varepsilon>0$,
 and $E\{N(t)\} \ge \varepsilon>0$;

    \item\label{assume:hlamTV}
     the total variation of  the candidate estimate $\hlam(\cdot;\theta)$,
     is bounded by $K_v$ uniformly in $\theta \in [\theta_0-K_\theta, \theta_0+K_\theta]$ with probability tending to one as $n\rightarrow\infty$, where $ \theta_0$ is the true value of the parameter $\theta$;

    \item \label{assume:hlamlim}
     $\hlam(t;\theta)$ is
     approximately linear in $\theta$ in the neighborhood of $\theta_0$,
      with respect to the total variation:  with arbitrary partition $0=t_0<\dots<t_N = \tau$ and $N \in \N$,
     \begin{equation*}\label{def:Ehlamth}
      \bigvee_{t=0}^\tau \left\{\hlam(t;\theta) - \hlam(t;\theta_0)\right\}
       = \sup_{\substack{0=t_0<\dots<t_N = \tau \\ N\in\N} }
       \sum_{j=1}^N \left|\hlam(t_{j-1};\theta) - \hlam(t_j;\theta_0)\right|= O_p(|\theta-\theta_0|);
     \end{equation*}

  \item\label{assume:rate-inf} the rates of estimation errors satisfy:
  \begin{align}
   \sqrt{\log(p)} \bigl\|\hbeta - \bbeta_0 \bigl\|_1 &+ \sup_{t \in [0,\tau]}|\hlam(t;\theta_0) - \Lambda_0(t)|
  + \|\hgr-\bgr_0\|_1 \notag \\
  +& \sqrt{n}\mseng(\hgr,\bgr_0)\left(\msenb(\hbeta,\bbeta_0)+\sup_{t \in [0,\tau]}|\hlam(t;\theta_0) - \Lambda_0(t)|\right) = o_p(1),  \label{eq:rate-inf}
  \end{align}
  and the estimation error of the baseline hazard satisfies
  \begin{equation}\label{eq:rate-hlam}
    \int_0^\tau H(t) d \left \{\hlam(t;\theta_0) - \Haz_0(t) \right\}
    = o_p(1)
  \end{equation}
  for any process $H(t)$ with $\sup_{t\in[0,\tau]}|H(t)| = O_p(1)$  adapted to the filtration $\Ftn$.
\end{enumerate}
\end{assumption}

Under Assumption \ref{assume:inf} above, we have the following result regarding the asymptotic distribution of $\hat \theta$ in the presence of both high-dimensional and infinite-dimensional nuisance parameters.
The result below does not require exact  
 model recovery   for either the outcome or the treatment model, and is based solely on estimation
 consistency of the nuisance parameters, which is achieved using the existing machine learning methods.
\begin{theorem}\label{thm:aalen}
Under Assumption \ref{assume:inf}, $\hth$ that solves 
  $n^{-1}\sum_{i=1}^{n}\scorei_i(\theta; \hbeta,\hlam,\hgr) = 0$ in \eqref{score:inference}
converges in distribution to a  normal random variable at $\sqrt{n}$-rate,
\begin{equation*}\label{norm:aalen}
  \hsig^{-1}\sqrt{n}(\atei-\theta_0) \leadsto N(0,1)
\end{equation*}
when both $n,p \to \infty$,
where `$ \leadsto $' denotes convergence in distribution, and the variance estimator takes the closed form $\hsig^2: =  \sigma^2(\atei)$ with
\begin{equation}\label{hat:aalen-var}
 \sigma^2(\theta)= \frac{n^{-1}\sum_{i=1}^{n}\delta_i \{D_i-\expit(\hgr^\top\bZI[i])\}^2e^{2\theta D_i X_i}}
  {\left\{n^{-1}\sum_{i=1}^n(1-D_i)\expit(\hgr^\top\bZI[i])X_i\right\}^2}.
\end{equation}
\end{theorem}

A few comments are in order.
Bounded norm of the covariates 
appear commonly in   non-linear   high dimensional models including the  generalized linear model \citep{vdGeerEtal14}, the Cox proportional hazards model \citep{HuangEtal13},  as well as the additive hazards models \citep{LinLv13}.
Positivity assumption, or the overlap assumption, that  requires the propensity scores for all subjects to be  bounded away from zero, as well as positive at-risk and event rates, are necessary assumptions for causal and survival inferences, respectively.  


The conditions  in Assumption \ref{assume:inf}, collectively, play a similar role as the commonly used, high-level requirement in semiparametric literature where each estimator is required to converge  at $n^{-1/4}$ rate or faster \citep{Robinson88,BelloniEtal13,Farrell15}.
Here, \eqref{eq:rate-inf}  allows one or two (but not all simultaneously) of our estimators   to  converge arbitrarily slowly, in $l_1$ sense (which in low-dimensions is equivalent to $l_2$ sense), as long as the products
\[
\| \hat \gamma - \gamma_0 \|_2 \| \hat \beta - \beta_0\|_2  \qquad {\mbox{and}} \qquad  \| \hat \gamma - \gamma_0 \|_2\sup_{t \in [0,\tau]}|\hlam(t) - \Lambda_0(t) |
\]
converge  at the $n^{-1/2} $ rate.
The latter  product is specific for the semiparametric additive hazards model and is new in the literature, namely, the interference of the infinite-dimensional $\Lambda$  and the treatment model coefficients.
If the distribution of the covariates $\bZ$ can be learned under suitable model assumptions,
we may further improve the rate assumptions by leveraging the information about $Z$ \citep{RobinsEtal17}.
Another new aspect is the  rate $ \sqrt{\log(p)}\| \hat \beta - \beta_0 \|_1$ in \eqref{eq:rate-inf},
 which in turn is a result  of a high-dimensional dependencies in the risk-sets.
 Section \ref{section:cf} discusses cross-fitting and there we showcase rate double robustness of our procedure.

 The proof of Theorem \ref{thm:aalen}  is given in the Supplementary Materials. We note that
$\bbeta_0^\top\bZ_i$ is allowed to grow arbitrarily large, as $\|\bbeta_0\|_1$ grows with $p$ and $n$.
 To the best of our knowledge, Theorem \ref{thm:aalen}  is the first such result
  with unbounded conditional hazard function of a survival outcome,
 which distinguishes us from the existing works   \citep{HouEtal19,YuEtal2018}.


Finally we provide in the next lemma a sufficient condition under which
the root of $n^{-1}\sum_{i=1}^n\scorei_i(\theta; \hbeta,\hlam,\hgr) =0$ is unique over any compact set $[-K_\theta,K_\theta]$
containing $\theta_0$.
\begin{lemma}\label{lem:unique}
Suppose that the cumulative baseline  hazard estimator  is linear in $\theta$, i.e.
\begin{equation}\label{def:hHaz-lin}
  \hHaz(t;\theta) = \hHaz(t;\theta_0) + (\theta - \theta_0) \hat{D}(t),
\end{equation}
for some differentiable $\hat{D}(t)$ whose derivative $\hat{d}(t) = \hat{D}'(t)$
satisfy $0 \le \hat{d}(t) \le 1$ for $t\in[0,\tau]$
and
$\sup_{t\in [0,\tau]} |\hat{d}(t) - d_0(t)| = o_p(1)$, for some  $d_0(t)$.
Under Assumption \ref{assume:inf}, the root of $n^{-1}\sum_{i=1}^n\scorei_i(\theta; \hbeta,\hlam,\hgr)=0$
is unique with probability tending to one on any compact set $[-K_\theta,K_\theta]$ containing $\theta_0$.
\end{lemma}
Note that the commonly used cumulative
baseline hazard estimator under the additive hazards model is linear in $\theta$.
This is also the case for the 
estimators considered in Section \ref{section:lasso}
and later in Section \ref{section:new}.
The closed-form estimator of $\theta$ 
 in Section \ref{section:check} is also unique by definition.

\subsection{ An illustrative example: Lasso regularizers }\label{section:lasso}


Whenever models \eqref{model:aalen} and \eqref{model:D} are sparse high-dimensional models, many regularization methods can be employed for obtaining $\hat \beta$ and $\hat \gamma$. We provide illustration  for a simple case.
 Let $s_\beta =\|\beta_0\|_0$ and  $s_\gamma =\|\gamma_0\|_0,$ denote the sparsity of the outcome and treatment assignment, respectively. 
For the additive hazards survival outcome model \eqref{model:aalen}, a simple and widely used approach is the Lasso regularized estimator of
\cite{LengMa07},
defined as
\begin{equation}\label{init:beta}
  (\hth_{l},\hbeta^\top)^\top = \argmin_{(\theta_{l},\bbeta^\top)^\top \in \R^{p+1}} (\theta_{l},\bbeta^\top) H_n (\theta_{l},\bbeta^\top)^\top -
  2 (\theta_{l},\bbeta^\top) \bfh_n + \rho (\|\bbeta\|_1+|\theta_{l}|),
\end{equation}
where
\begin{align*}
  H_n = &
  n^{-1}
\sum_{i=1}^{n} \int_0^\tau
 \left\{\left(D,\bZ_i^\top\right)^\top - \left(\Dbar(t),\Zbar(t)^\top\right)^\top\right\}^{\otimes 2} Y_i(t)dt , \; v^{\otimes 2} = v v ^\top, \\
  \bfh_n = &
n^{-1}  \sum_{i=1}^{n} \int_0^\tau
 \left\{\left(D,\bZ_i^\top\right)^\top - \left(\Dbar(t),\Zbar(t)^\top\right)^\top\right\}dN_i(t) ,
 \end{align*}
with
$
 \Dbar(t) = 
 {\sum_{i=1}^n D_i Y_i(t)}/{\sum_{i=1}^n Y_i(t)}
 $
 and
$
 \Zbar(t) = 
 {\sum_{i=1}^n \bZ_i Y_i(t)}/{\sum_{i=1}^n Y_i(t)}
 $.
 Since $\theta$ is part of the parameter under model \eqref{model:aalen},
 we obtain an initial estimator $\hth_l$.
 Due to the bias induced by the regularization,
 $\hth_l$ cannot be used to draw inference.
 As shown later, $\hth_l$ can still be useful in the construction
 of the baseline hazard estimator $\hHaz$.

The estimation of $\bgr$ under
 model \eqref{model:D} is similar. We may use the Lasso estimator under the logistic regression model \citep{Tibshirani96}:
\begin{equation}\label{init:gr}
  \hgr = \argmin_{\bgr \in \R^{p+1}}- \frac{1}{n}\sum_{i=1}^{n}
  \left\{ D_i\bgr^\top\bZI[i] - \log(1+\exp\{\bgr^\top \bZI[i]\})\right\}
  + \rho  \| \gamma\|_1.
\end{equation}
We used the same notation $\rho$ for the tuning parameter  in \eqref{init:beta} and \eqref{init:gr} for simplicity.

 Following   \cite{GaiffasGuilloux12},   among others,
under  restricted eigenvalue conditions \cite[for example]{BickelEtal09},
the above estimators satisfy
\begin{gather}
 \|\hbeta-\bbeta_0\|_1 = O_p\left( s_\beta \sqrt{\log(p)/n}\right), \,
  \|\hgr-\bgr_0\|_1 = O_p\left(s_\gamma \sqrt{\log(p)/n}\right), \notag \\
\mathcal{D}_{\bbeta}(\hbeta, \bbeta_0) = O_p\left( \sqrt{ s_\beta \log(p)/n}\right),
  \,  \mseng(\hgr,\bgr_0) = O_p\left( \sqrt{s_\gamma\log(p)/n}\right),\label{rate:Lasso}
\end{gather}
whenever $\rho > C \sqrt{\log(p)/n}$ and $\log(p)/n=o(1)$, i.e.~allowing the dimension $p$ to be much larger than $n$.

Finally, our estimation has one additional nonparametric nuisance parameter,  the cumulative baseline hazard function. Let
\begin{equation}\label{def:hlam_inf}
\hlam(t; \hat \beta, \theta) = \int_0^t \frac{\sum_{i=1}^{n}\{dN_i(u)-Y_i(u)(\hbeta^\top\bZ_i + \theta D_i)du\}}
  {\sum_{i=1}^{n}Y_i(u)}
\end{equation}
be the well-known Breslow estimator  \citep{LinYing94}.
  We observe that
\begin{equation*}
\hlam(t;\hbeta,\hth_{l}) = \Haz_0(t) +
\int_0^t\frac{\sum_{i=1}^n dM_i(u;\theta_0,\bbeta_0,\Haz_0)}{\sum_{i=1}^n Y_i(u)}
- \int_0^t \{(\hbeta-\bbeta_0)^\top \Zbar(u)+ (\hth_{l}-\theta_0)\Dbar(u)\} du,
\end{equation*}
 Under Assumptions 
  \ref{assume:Z} and \ref{assume:varD},
 and applying similar  rates as in \eqref{rate:Lasso}, we have
 $$
  \sup_{t\in [0,\tau]} |\hlam(t;\hbeta,\hth_{l})-\Haz_0(t)|
  =  O_p\left((s_\beta+1) \sqrt{\log(p)/n}\right).
 $$
This way Assumption \ref{assume:rate-inf} is satisfied with a proper choice of the sparsities below.
In addition, direct calculation verifies that Assumptions \ref{assume:hlamTV} and \ref{assume:hlamlim} are satisfied with \eqref{def:hlam_inf}.

 A sufficient condition  for Theorem \ref{thm:aalen}   is  that the sparsities of the  outcome and treatment models satisfy
$$ \max \{ s_\beta, s_\gamma\} = o(\sqrt{n}/\log(p)  )  \mbox{ and } s_\beta s_\gamma = o\left(n/\log(p)^2  \right).$$
We note that the sparsity  conditions above are comparable to those of  \cite{BelloniEtal13, Farrell15, vdGeerEtal14,avagyan2017honest}.
In the above
  the first two sparsity conditions are needed to satisfy $l_1$ consistency at any rate, and the last is required to satisfy $\| \hat \gamma - \gamma_0 \|_2 \| \hat \beta - \beta_0\|_2 = o(n^{-1/2})$, a condition equivalent to that in the   fundamental works of \cite{RobinsRotnitzky95},
    among others.


\section{Double robustness in high-dimensions}\label{section:ext}

In the following we begin with a special case of our estimator  from Section \ref{section:inf},
which  has a closed-form expression and can be seen as directly estimating the hazards difference (HDi).
In Section \ref{section:cf} we introduce the cross-fitting scheme, leading to a relaxation of sparsity conditions under which we establish  rate double robustness.

\subsection{Hazards Difference (HDi) estimator}\label{section:check}

Here we introduce an estimator that utilizes covariate balancing, to directly estimate the hazard difference (HDi) under models \eqref{model:aalen} and \eqref{model:D}. Covariate balancing weights have been shown to be a useful approach
and have gained substantial interest in the recent causal literature; see for example \mbox{\cite{ImaiRatkovic14}} 
 and \cite{LiEtal18}.
An advantage of covariate balancing is that it eliminates confounding without having to rely on a correctly specified propensity score model
or when  $\bgr_0$ cannot be consistently estimated under model \eqref{model:D}.

Define a set of covariate balancing weights
$$
w^0_i(\bgr) = (1-D_i) P(D_i=1|Z_i) 
, \,\,\, w^1_i(\bgr) =D_i P(D_i=0|Z_i),
$$
 where $\gamma$ is the same as defined in model \eqref{model:D}.
It is straightforward to verify that the above weights are covariate balancing in the following sense.
Consider the weighted empirical cumulative distribution functions of the covariates:
\begin{equation*}\label{def:balance-cdf}
F_{d,n}(\bz) = 
\frac{n^{-1}\sum_{i=1}^{n}w^d_i(\bgr_0)I(\bZ_i \le \bz)}{n^{-1}\sum_{i=1}^{n}w^d_i(\bgr_0)}, \;
d = 0,1,
\end{equation*}
 where 
 $I(\bZ \le \bz) = \prod_{j=1}^p I(\bZ^j \le \bz^j)$
with $\bZ^j$ and $\bz^j$ being the $j$-th component of  $\bZ$ and $\bz$, respectively.
 It is straightforward to show that $F_{d,n}(\bz)$ converges in  probability to
 ${\E\{\Var(D|\bZ)I(\bZ \le \bz)\}} / {\E\{\Var(D|\bZ)\}}$ for
 $d=0,1$;
i.e.~the distributions of covariates in both treatment arms are approximately the same
after weighting.

Now define
\begin{align*}
  & \clam^k(t;\bbeta,\bgr) = \int_0^t \frac{\sum_{i=1}^n w^k_i(\bgr) \{dN_i(u)-Y_i(u)\bbeta^\top\bZ_i du\}}{
    \sum_{i=1}^n w^k_i(\bgr) Y_i(u)}, \qquad k=0,1.
\end{align*}
Under the additive hazards model \eqref{model:aalen}, $\clam^1 $ and $\clam^0 $ can be seen to estimate
 $\Haz_0(t)+\theta t$ and $\Haz_0(t)$, respectively.
 It is then immediate that the HDi estimator can be defined as the weighted difference of $\clam^0$ and $\clam^1$,
\begin{equation}\label{check:contrast}
  \cth = \Bigl\{ \sum_{i=1}^{n}w^0_i(\hgr)X_i\Bigl\} ^{-1}  \sum_{i=1}^{n} \int_0^\tau w^0_i(\hgr) Y_i(t)d\left\{ \clam^1(t;\hbeta,\hgr) - \clam^0(t;\hbeta,\hgr) \right\}.
\end{equation}
In the following we show that 
$ \cth $ is also a special case of our class of estimators from Section \ref{section:inf}.


With the weights defined above, the orthogonal score in \eqref{score:inference} can be written
\begin{equation}\label{def:score_w}
  \frac{1}{n}\sum_{i=1}^{n}\scorei_i(\theta; \bbeta,\Haz,\bgr) = \frac{1}{n} \sum_{i=1}^{n} \int_0^\tau
\{  e^{\theta t}  w^1_i(\bgr) - w^0_i(\bgr) \}
  d M_i(t;\theta, \bbeta, \Haz).
\end{equation}
Setting
$ \sum_{i=1}^{n}\int_0^t w^1_i(\bgr) d M_i(u; \theta, \bbeta, \Haz) =0$,
we have
 the weighted Breslow  estimator:
\begin{equation}\label{def:clam}
 \clam(t,\theta;\bbeta,\bgr) = \int_0^t \frac{\sum_{i=1}^n w^1_i(\bgr) \{dN_i(u)-Y_i(u)(D_i\theta+\bbeta^\top\bZ_i)du\}}{
    \sum_{i=1}^n w^1_i(\bgr) Y_i(u)}.
\end{equation}
We note that \eqref{def:clam} only uses weighted observations from treatment group `1', and
it  appears to be the only obvious choice so that we obtain an orthogonal score  linear in $\theta$.
Using \eqref{def:clam} in \eqref{def:score_w} the orthogonal score 
becomes a linear function of $\theta$:
\begin{align}
 & \frac{1}{n} \sum_{i=1}^{n} \scorei_i\big(\theta; \bbeta,\clam(\cdot,\theta; \bbeta,\bgr),\bgr \big) \notag \\
=& - \frac{1}{n} \sum_{i=1}^{n} \int_0^\tau w^0_i(\bgr)
  d M_i(t;\theta, \bbeta, \Haz) \notag\\
=&  -\frac{1}{n}\sum_{i=1}^{n} (1-D_i)\expit(\bgr^\top\bZI[i])  \int_0^\tau \left(dN_i(u)-Y_i(u)
\left[ \bbeta^\top\{\bZ_i-\Ztil(u;\bgr)\}du+d\Ntil(u;\bgr) \right]\right) 
\notag\\
&-\frac{\theta}{n}\sum_{i=1}^{n}(1-D_i)\expit(\bgr^\top\bZI[i])X_i,\label{score:DR-linear}
\end{align}
where we used the weighted processes
$$\Ztil(t;\bgr) = {\sum_{i=1}^n \bZ_i w^1_i(\bgr)Y_i(t)}/{\sum_{i=1}^n w^1_i(\bgr)Y_i(t)}, 
\quad
    \mbox{ and }
    \quad
   d\Ntil(t;\bgr) = {\sum_{i=1}^n w^1_i(\bgr)dN_i(t)}/{\sum_{i=1}^n w^1_i(\bgr)Y_i(t)}.$$
We then have
\begin{equation}\label{check:aalen-ate}
  \cth  = \Bigl[ \sum_{i=1}^{n}w^0_i(\hgr)X_i   \Bigl]^{-1} \int_0^\tau{\sum_{i=1}^{n} w^0_i(\hgr)\left[Y_i(u) \Bigl(\hat{\bbeta}^\top\{\bZ_i-\Ztil(u; \hat{\bgr})\}du+d\Ntil(u; \hat{\bgr}) \Bigl ) - dN_i(u)\right] }{ },
\end{equation}
which is readily verified to be equal to \eqref{check:contrast}.


For the HDi estimator we have the following.

\begin{theorem}\label{thm:2}
Let Assumption  \ref{assume:Z} and \ref{assume:varD} hold. If
 \begin{equation*}
  \sqrt{\log(p)}\|\hbeta - \bbeta_0\|_1
  + \|\hgr-\bgr_0\|_1 \notag \\
  + \sqrt{n}\mseng\left(\hgr,\bgr_0\right) \msenb \bigl(\hbeta , \bbeta_0\bigl) = o_p(1),
  \end{equation*}
then the HDi estimator $\cth$
satisfies
\begin{equation*}
  \check{\sigma}^{-1}\sqrt{n}(\cth-\theta_0) \leadsto N(0,1)
\end{equation*}
when both $n,p \to \infty$,
where  $  \check{\sigma}^2 = \sigma^2(\cth)$ with $\sigma(\cdot)$ in
\eqref{hat:aalen-var}.
\end{theorem}
In simulation studies we show that the HDi estimator has good empirical   properties even in the presence of model misspecifications. Theoretical investigation under model misspecification is beyond the scope of this work, and some preliminary results can be found in Section 3.3. of \cite{HouEtal_arxiv}.

\subsection{Rate double robustness}\label{section:cf}

In this section, we  develop a cross-fitted estimator (HDi or \eqref{eq:thetahat}) that achieves rate double robustness, as defined in \cite{SmuclerRotnitzkyRobins19}.
An estimator  is rate double robust if it is $l_2$ consistent and asymptotically normal
whenever: 1) the product of the estimation error of two nuisance parameters decay
faster than $n^{-1/2}$, and 2) the estimator of one nuisance parameter is allowed  arbitrarily slow $\ell_2 $ convergence rate.
In Theorems \ref{thm:aalen} and \ref{thm:2}, we have shown that the one-shot estimator and the HDi estimator satisfy the first condition of rate double robustness while assuming the estimator of uisance parameter is consistent in $\ell_1$-norm.
Now, we proceed to contruct the rate double robust estimator that allow for estimator of uisance parameter consistent in $\ell_2$-norm.

In order to achieve rate double robustness, we apply
a cross-fitting procedure that utilizes one part of the sample to estimate all the nuisances and uses the remaining part to find the zero of the score.
Algorithm \ref{alg:cf}  summarizes the estimation procedure.
Cross-fitting induces independence between the score and the estimated  nuisance parameters,
 further reducing the effect of the nuisance parameters on the treatment effect estimation.
This is  in addition to  the orthogonality of the score function.
Note that, cross-fitting is not necessary in low-dimensional problems; double-robustness is guaranteed by the score alone. In high-dimensional settings this is no longer the case.
For the number of cross-fitting folds $k$ in the Algorithm \ref{alg:cf},
any choice  produces (in theory) asymptotically equivalent results.
In practice we recommend $k$ to be at  $5$ or $10$ based on our own experience.
\begin{algorithm}
 \caption{Estimation of the Treatment Effect via $k$-fold Cross-fitting
 }\label{alg:cf}
\SetAlgoLined
\KwData{split the data into $k$ folds of equal size with the indices set $\fold[1], \fold[2], \dots, \fold[k]$}
\For{each fold indexed by $j$}{
1. estimate the nuisance parameters
  $\left(\hbj,\hlamj, \hgrj\right)$ using the   out-of-$j$-fold   data indexed by $\fold[-j] = \{1,\dots, n\}\setminus\fold[j]$ \;
2.  construct the cross-fitted score using the in-fold samples:
\begin{align}
     & \scorei_i\left(\theta; \hbj,\hlamj, \hgrj\right) \notag  \\
&  = \int_0^\tau e^{D_i\theta t}
\bigl[D_i - \expit(\hgrjt\bZI[i])\bigl]
  \left[dN_i(t)-Y_i(t)\Bigl( \bigl(D_i\theta+ \hbjt\bZ_i\bigl)dt + d \hlamj(t;\theta)\Bigl) \right] \label{score:setj}
\end{align}
}
\KwResult{Obtain the estimated treatment effect $\atei_{cf}$ by solving
    \begin{equation}\label{hat:aalen-ate-cf}
     \frac{1}{k} \sum_{j=1}^{k} \frac{1}{|\fold[j]|}\sum_{i \in \fold[j]}\scorei_i\left(\theta; \hbj,\hlamj, \hgrj\right)  = 0.
    \end{equation}
}
\end{algorithm}

In order to present theoretical results we need an  additional  set of notation.
Let  $(X_*, \delta_*, D_*, \bZ_*)$ be an independent copy   of  the original data.
We denote $\E_*$  as the expectation taken with respect to the independent data copy
conditionally on the other random variables.
We define the {\it average testing deviance}  as
\begin{gather}
  \mseb(\hbeta, \bbeta_0) =
  \left( \E_*\left[\int_0^\tau \left\{(\hbeta-\bbeta_0)^{\top}\bZ_* \right\}^2Y_*(t)dt\right] \right)^{1/2}
   , \notag  \\
  \mseg(\hgr, \bgr_0)=
  \left( \E_* \left[\left\{\expit (\hgr^{\top}\bZ_* )
  -\expit (\bgr_0^{\top}\bZ_* )\right\}^2\right] \right)^{1/2}.
\label{eq:mse}
\end{gather}
We note that the average testing deviance is a relaxed measure of convergence, compared to the one
in $l_1$-norm; see Section \ref{section:inf}.
The average testing deviance has a convergence rate that grows slower  than $l_1$ norm,
\begin{align*}
  \mseb(\hbeta, \bbeta_0)
   \le \|\hbeta-\bbeta_0\|_2 \sqrt{C\tau}, \qquad
  \mseg(\hgr, \bgr_0)
  \le \|\hgr-\bgr_0\|_2\sqrt{C}
\end{align*}
where $\lambda_{\max}(\cdot)$ denotes the maximal eigenvalue and a constant $C=\lambda_{\max} (\Var_*(\bZ_*) )+K_Z^2$,
 where  $K_Z$ is defined in Assumption \ref{assume:Z}.
Below we state Assumption \ref{assume:hat-cf}, under which we have the inference result for $\theta$ under models \eqref{model:D} and \eqref{model:aalen}, using the cross-fitted estimator $\atei_{cf}$ defined in \eqref{hat:aalen-ate-cf}.

\begin{assumption}\label{assume:hat-cf}
Suppose that Assumption  \ref{assume:Z}, \ref{assume:varD}
and \ref{assume:hlamlim} 
 hold with $(\hbeta,\hlam,\hgr) = \left(\hbj,\hlamj, \hgrj\right)$
for all $j=1,\dots,k$.
Assume additionally, that  there exist  positive sequences of real numbers $\gamma_n, \beta_n \to 0$ as $n,p\to \infty$ such that  $\mseg\left(\hgrj,\bgr_0\right)=o(\gamma_n)$ and $\mseb\left(\hbj,\bbeta_0\right) =o(\beta_n)$, and
  \begin{equation}
  \sqrt{n} \gamma_n \Bigl(\beta_n+ \sup_{t \in [0,\tau]}\bigl|\hlamj(t;\theta_0) - \Lambda_0(t)\bigl|\Bigl) = o_p(1). \label{eq:rate-inf-cf}
  \end{equation}
\end{assumption}

\begin{theorem}\label{thm:aalen-cf}
Under Assumption \ref{assume:hat-cf}, for $\atei_{cf}$
defined in \eqref{hat:aalen-ate-cf}
we have
\begin{equation*}
  \hsig_{cf}^{-1}\sqrt{n}(\atei_{cf}-\theta_0) \leadsto N(0,1)
\end{equation*}
when both $n,p \to \infty$, with the closed-form variance estimator
\begin{equation*}
  \hsig_{cf}^2 = \frac{n^{-1}\sum_{j=1}^{k} \sum_{i \in \fold[j]}\delta_i \{D_i-\expit(\hgr^{(j)\top}\bZI[i])\}^2e^{2\hth_{cf} D_i X_i}}
  {\left\{n^{-1}\sum_{j=1}^{k} \sum_{i \in \fold[j]} (1-D_i)\expit(\hgr^{(j)\top}\bZI[i])X_i\right\}^2}.
\end{equation*}
\end{theorem}

We make a few remarks for the above result.
The cross-fitted score \eqref{hat:aalen-ate-cf} can handle
a larger number of covariates, less sparse models
and less restrictive estimators of the baseline hazard  in that Assumption  \ref{assume:hlamTV} is not required.
The removal of condition \eqref{eq:rate-hlam}
allows various estimation methods of the baseline hazards besides the Breslow type estimators, e.g.~
 splines.
In addition, \eqref{eq:rate-inf-cf} allows a larger dimension without the extra $\log(p)$ factor
of \eqref{eq:rate-inf}.
Lastly,  \eqref{eq:rate-inf-cf}  implies a sufficient condition for Theorem \ref{thm:aalen-cf}:
\begin{equation}\label{eq:s2}
 \max\{s_\beta, s_\gamma\} =o(n/\log(p)) \mbox{ and } s_\beta s_\gamma = o\left(n/\log(p)^2  \right),
\end{equation}
which is weaker  than those discussed in Section \ref{section:inf}.
Note that  $ \sqrt{n}/\log(p) \leq s \leq {n}/\log(p)$ (Theorem \ref{thm:aalen-cf})  is known as the moderately sparse regime,  whereas  $s  \leq \sqrt{n}/\log(p)$ (Theorems \ref{thm:aalen} and \ref{thm:2}) is known as the exactly or strictly sparse regime.
Our condition \eqref{eq:s2}  parallels that of  rate double robustness  as defined in \cite{SmuclerRotnitzkyRobins19}.
Our  results confirm that the same rate condition is achievable with censored data as well as with
the additional nonparametric component of the baseline hazard.

\section{Numerical experiments}\label{section:simulation}

We considered $n=p=300$ and $n=p=1500$.
To ensure that the baseline hazard is non-negative,
the covariates $
Z_1,\dots, Z_{p}$ were independently generated from $\mathcal{N}(0,1)$ conditioned
on the event that $\beta^\top Z_i \ge 0.25$.
The censoring time $C$ was generated as
the smaller between $\tau$ and Uniform $(0,c_0)$;
the parameters $\tau$ and $c_0$ were chosen such that
$n/10$ treated subjects were expected to be at-risk at $t=\tau$, and the censoring rate was around $30\%$.
We repeated simulation $500$ times.
We considered the four proposed estimators:
$\hth$ obtained from \eqref{score:inference} with the Breslow estimator \eqref{def:hlam_inf},
the HDi estimator $\cth$ in \eqref{check:aalen-ate},
 and their cross-fitted counterparts $\hth_{cf}$ as described in Algorithm \ref{alg:cf} and $\cth_{cf}$,
 $$
     \cth_{cf} =
 \sum_{j=1}^k\sum_{i\in\fold[j]}   \frac{\int_0^\tau w^0_i\left(\hgrj\right) Y_i\left(t\right)}
  {\sum_{j=1}^k\sum_{i\in\fold[j]}w^0_i\left(\hgrj\right)X_i}d\left\{\clam^{1(j)}\left(t;\hbj,\hgrj\right) - \clam^{0(j)}\left(t;\hbj,\hgrj\right) \right\}.
 $$
For comparison, we also present the result of
the Lasso estimator, $\tilde{\theta}$, under the additive hazards model with 
 the penalty for $\theta$ set to be zero.
 We note that  \eqref{init:gr} is obtained
by the R-package \emph{glmnet}, whereas  \eqref{init:beta}
by the R-package \emph{ahaz}. For $\hth$ and $\cth$
the penalty parameters were selected by 10-fold cross-validation.
For cross-fitting  the number of folds was set as 10,
and
within each fold we used 9-fold cross-validation.

\subsection{Finite sample inference}

The event time $T$ was generated from the additive hazards   \eqref{model:aalen} with $\theta_0 = -0.25$, and $\lambda_0(t) = 0.25$, whereas
  the treatment assignment $D$ from the logistic regression   \eqref{model:D} with
the intercept   chosen so that the marginal probability $\P(D=1)=0.5$.
We considered the following three sparsity levels:
\begin{align*}
&\text{very sparse $s_\beta=2$: }
 \bbeta = (1,0.1, \underbrace{0,\dots,0}_{p-2}),  \,\,\,
s_\gamma =1:  \bgr = (1, \underbrace{0,\dots,0}_{p-1}); \\
&\text{sparse $s_\beta=6$: }
\bbeta = (1,\underbrace{0.1,\dots,0.1}_{5}, \underbrace{0,\dots,0}_{p-6}), \,\,\,
s_\gamma =3: \bgr = (1, 0.05,0.05, \underbrace{0,\dots,0}_{p-3}); \\
&\text{moderately sparse $s_\beta=15$: }
  \bbeta = (1,\underbrace{0.1,\dots,0.1}_{13},
    \underbrace{0,\dots,0}_{p-15}), \\
&\quad\qquad\qquad \qquad
s_\gamma =10: \bgr = (1, 1, \underbrace{0.05,\dots,0.05}_{8},
    \underbrace{0,\dots,0}_{p-10}).
\end{align*}
We consider  four pairs of sparsities:
$(s_\beta=2, s_\gamma=1)$,
$(s_\beta=2, s_\gamma=10)$, $(s_\beta=15, s_\gamma=1)$
and $(s_\beta=6, s_\gamma=3)$.

\begin{table}[h]
\caption{Inference results under correctly specified models:
true $\theta = -0.25$, $30\%$
censoring.
$\hth$, $\cth$, $\hth_{cf}$ and $\cth_{cf}$ are
the four proposed estimators, where the subscript `cf' denotes the cross-fitted version. The naive Lasso estimator $\tilde{\theta}$ penalized only the covariate effects $\beta$ but not $\theta$.
 `CP' is the coverage probability of the nominal 95\% confidence intervals.
} \label{tab:inf}
\begin{center}
\scriptsize
\begin{tabular}{lllllllllllllll}
\hline
\hline
\multicolumn{2}{c}{Sparsity} && \multicolumn{2}{c}{Lasso $\tilde{\theta}$}
&&   \multicolumn{4}{c}{$\hat{\theta}$ }
 &&  \multicolumn{4}{c}{HDi $\check{\theta}$}\\
\cline{1-2} \cline{4-5} \cline{7-10} \cline{12-15}
$s_\beta$ & $s_\gamma$ && Bias & SD & & Bias & SD & SE & CP
& & Bias & SD & SE & CP\\

\hline
\multicolumn{15}{c}{n=p=300} \\
\hline
 2 & 1 &  &  0.054 & 0.097 &  &  0.029 & 0.097 & 0.091 & 92.4 \% &  &  0.032 & 0.094 & 0.091 & 93.0 \% \\
  6 & 3 &  &  0.071 & 0.094 &  &  0.050 & 0.095 & 0.093 & 92.0 \% &  &  0.052 & 0.092 & 0.093 & 93.0 \% \\
  15 & 1 &  &  0.088 & 0.135 &  &  0.051 & 0.123 & 0.128 & 93.6 \% &  &  0.049 & 0.122 & 0.128 & 94.0 \% \\
  2 & 10 &  &  0.099 & 0.094 &  &  0.050 & 0.099 & 0.094 & 89.6 \% &  &  0.052 & 0.096 & 0.094 & 89.8 \% \\

\hline
\multicolumn{15}{c}{n=p=1500} \\
\hline
 2 & 1 &  &  0.031 & 0.040 &  &  0.009 & 0.041 & 0.041 & 94.2 \% &  &  0.011 & 0.041 & 0.041 & 94.0 \% \\
  6 & 3 &  &  0.033 & 0.042 &  &  0.015 & 0.043 & 0.042 & 93.0 \% &  &  0.017 & 0.042 & 0.042 & 93.6 \% \\
  15 & 1 &  &  0.047 & 0.063  &  &  0.019 & 0.064 & 0.058 & 91.4 \% &  &  0.020 & 0.063 & 0.058 & 91.2 \% \\
  2 & 10 &  &  0.077 & 0.041 &  &  0.019 & 0.043 & 0.043 & 91.6 \% &  &  0.022 & 0.043 & 0.043 & 91.6 \% \\

\hline
\hline

\multicolumn{2}{c}{Sparsity}  &&  &
&&  \multicolumn{4}{c}{$\hat{\theta}_{cf}$ }
 &&  \multicolumn{4}{c}{HDi $\check{\theta}_{cf}$}\\
\cline{1-2}  \cline{7-10} \cline{12-15}
$s_\beta$ & $s_\gamma$ && & & & Bias & SD & SE & CP
& & Bias & SD & SE & CP\\

\hline
\multicolumn{15}{c}{n=p=300} \\
\hline
 2 & 1 &  &   &  &  &   0.012 & 0.100 & 0.090 & 91.0 \% &  &  0.011 & 0.093 & 0.090 & 93.4 \% \\
  6 & 3 &  &   &  &  &  0.027 & 0.100 & 0.092 & 92.4 \% &  &  0.028 & 0.089 & 0.092 & 94.6 \% \\
  15 & 1 &  &   &  &  & 0.018 & 0.134 & 0.127 & 93.8 \% &  &  0.013 & 0.123 & 0.127 & 95.8 \% \\
  2 & 10 &  &   &  &  &  0.032 & 0.106 & 0.094 & 89.4 \% &  &  0.032 & 0.097 & 0.094 & 93.2 \% \\

\hline
\multicolumn{15}{c}{n=p=1500} \\
\hline
 2 & 1 &  &   &  &  &  0.006 & 0.042 & 0.041 & 94.8 \% &  &  0.009 & 0.040 & 0.041 & 95.4 \% \\
  6 & 3  &  &   &  &  &  0.010 & 0.044 & 0.041 & 92.8 \% &  &  0.014 & 0.041 & 0.042 & 94.2 \% \\
  15 & 1  &  &   &  &  &   0.006 & 0.064 & 0.058 & 92.4 \% &  &  0.012 & 0.061 & 0.058 & 93.0 \% \\
  2 & 10  &  &   &  &  &  0.017 & 0.044 & 0.043 & 92.4 \% &  &  0.019 & 0.042 & 0.043 & 92.8 \% \\

\hline
\hline
\end{tabular}
\end{center}
\end{table}

We present the inference results in Table \ref{tab:inf}.
Using the orthogonal score approach, all four estimators largely reduced the bias of the naive Lasso $\tilde \theta$, especially visible for the larger sample size $ n=p=1500$.
All  four estimators achieved reasonably good coverage rates of the nominal 95\% confidence intervals with the larger sample size $n=p=1500$, while cross-fitted HDi estimator outperformed the rest in small sample sizes.
With the smaller sample size $n=p=300$,
$\cth_{cf}$ the cross-fitted HDi estimator also
had good coverage properties.

\subsection{Exploring Model Misspecification}

We test the robustness of our method when model or sparsity assumption
is violated in either the  propensity or the outcome  model.
We  consider violations of sparsity (above the dashed lines in Table \ref{tab:dr}) as well as model misspecification (below the dashed lines in Table \ref{tab:dr}). We simulated from the models above
with dense coefficients:
\begin{align*}
\text{Dense $s_\beta=30$: }
  &\bbeta = (\underbrace{1,\dots,1}_{4},\underbrace{0.1,\dots,0.1}_{26},
    \underbrace{0,\dots,0}_{p-30}),\\
\text{Dense 
$s_\gamma =20$: }
    &\bgr = (\underbrace{1,\dots,1}_{4}, \underbrace{0.05,\dots,0.05}_{16},
    \underbrace{0,\dots,0}_{p-20}).
\end{align*}
Two pairs of very sparse - dense combinations,
$(s_\beta=2, s_\gamma=20)$
and $(s_\beta=30, s_\gamma=1)$,
are studied.
We considered model misspecification in three different scenarios, denoted by `E', `P' and `D' which stand for `Exponential', `Probit' and `Deterministic', respectively.
\begin{itemize}
\item In scenario `E',  the event time was generated
with exponential link:
\begin{equation}\label{model:sim-quad}
\lambda(t|D,\bZ) =  -0.25 D + \exp\left(\bbeta^\top\bZ\right) + 0.25,
\end{equation}
while the logistic treatment model \eqref{model:D} was correct.
The coefficients were set as in
$(s_\beta=2, s_\gamma=1)$.
\item In scenario `P',  we considered the misspecified treatment model with probit link:
\begin{equation}\label{model:sim-probit}
\P(D=1|\bZ) = \rm{probit}\left(\bgr^\top\tilde{\bZ}  \right),
\end{equation}
while the additive hazards model \eqref{model:aalen}  for the event time was correct.
The coefficients were also set as in
$(s_\beta=2, s_\gamma=1)$.
\item In scenario `D', we considered another misspecified treatment model with
deterministic treatment assignment given $Z$:
\begin{equation}\label{model:sim-fix}
D|\bZ =I\left(\bgr^\top\tilde{\bZ} > \mu\right),
\end{equation}
where $\mu$ is the median of $\bgr^\top\tilde{\bZ}$. Again the additive hazards model \eqref{model:aalen}  for the event time was correct, and
the coefficients were  set as in
$(s_\beta=2, s_\gamma=1)$.  Since $\Var(D|Z)=0$ for all $Z$, i.e.~there was no overlap at all,  Assumptions \ref{assume:varD} 
was violated in this scenario.
\end{itemize}

We present the estimation results   in Table \ref{tab:dr}.
All four estimators had smaller bias than the naive Lasso $\tilde{\theta}$,
most notably under  scenarios `D' and `P' where
the naive Lasso failed completely. Moreover, the bias decreased significantly with larger sample sizes.
Additionally,
$\hth_{cf}$ and $\cth_{cf}$ demonstrated the advantages of cross-fitting
with even smaller biases. HDi estimator, $\cth_{cf}$
had often smaller SD than $\hth_{cf}$  leading to the improved MSE especially seen in smaller sample sizes.  In  scenario `D' with $n=p=300$, $\cth_{cf}$ showed advantage over
$\hth_{cf}$ which could not be  defined in one of the simulation runs where the score equation appeared to have  no numerical root in the  neighborhood of interest.

\begin{table}
\caption{Estimation  results under dense coefficients ($s_\beta=30$, $s_\gamma=20$) or misspecified models (exponential link `E', probit link `P' and deterministic treatment assignment `D').
True $\theta = -0.25$, $30\%$
censoring.
$\hth$, $\cth$, $\hth_{cf}$ and $\cth_{cf}$ are
the four proposed estimators, where the subscript `cf' denotes the cross-fitted version. The naive Lasso estimator $\tilde{\theta}$ penalized only the covariate effects $\beta$ but not $\theta$.} \label{tab:dr}
\begin{center}
\scriptsize
\begin{threeparttable}
\begin{tabular}{llllllllllllll}
\hline
\hline
\multicolumn{2}{c}{Sparsity} && \multicolumn{3}{c}{Lasso $\tilde{\theta}$}
&&  \multicolumn{3}{c}{$\hat{\theta}$}
 &&  \multicolumn{3}{c}{HDi $\check{\theta}$}\\
\cline{1-2} \cline{4-6} \cline{8-10} \cline{12-14}
$s_\beta$ & $s_\gamma$ && Bias & sd & $\sqrt{MSE}$
&& Bias & sd & $\sqrt{MSE}$ && Bias & sd & $\sqrt{MSE}$\\

\hline
\multicolumn{14}{c}{n=p=300} \\
\hline
 30 & 1 &  &  0.141 & 0.202 & 0.247  &  &  0.080 & 0.169 & 0.187 &  &  0.074 & 0.169 & 0.185 \\
  2 & 20 &  &  0.078 & 0.090 & 0.119  &  &  0.051 & 0.099 & 0.111 &  &  0.052 & 0.096 & 0.109 \\
 \cdashline{4-14}
  E & 1 &  &  0.233 & 0.397 & 0.461  &  &  0.117 & 0.375 & 0.393 &  &  0.106 & 0.366 & 0.381 \\
  2 & P &  &  0.095 & 0.102 & 0.139&  &  0.041 & 0.101 & 0.109 &  &  0.043 & 0.097 & 0.106 \\
  2 & D &  & -0.598 & 0.204 & 0.632&  & -0.103 & 0.217 & 0.240 &  & -0.124 & 0.223 & 0.255 \\

\hline
\multicolumn{14}{c}{n=p=1500} \\
\hline
 30 & 1 &  &  0.057 & 0.083 & 0.101 &  &  0.018 & 0.082 & 0.084 &  &  0.018 & 0.081 & 0.083 \\
  2 & 20 &  &  0.061 & 0.042 & 0.074  &  &  0.020 & 0.048 & 0.052 &  &  0.022 & 0.047 & 0.052 \\
   \cdashline{4-14}
  E & 1 &  &  0.132 & 0.169 & 0.214 &  &  0.049 & 0.167 & 0.174 &  &  0.049 & 0.163 & 0.171 \\
  2 & P &  &  0.050 & 0.043 & 0.066 &  &  0.012 & 0.045 & 0.046 &  &  0.014 & 0.044 & 0.046 \\
  2 & D &  & -0.308 & 0.092 & 0.321 &  & -0.032 & 0.104 & 0.109 &  & -0.040 & 0.107 & 0.114 \\

\hline\hline

\multicolumn{2}{c}{Sparsity}  && &&
&&   \multicolumn{3}{c}{$\hat{\theta}_{cf}$ }
 &&  \multicolumn{3}{c}{HDi $\check{\theta}_{cf}$}\\
\cline{1-2} \cline{8-10} \cline{12-14}
$s_\beta$ & $s_\gamma$ && & &
&& Bias & sd & $\sqrt{MSE}$ && Bias & sd & $\sqrt{MSE}$\\

\hline
\multicolumn{14}{c}{n=p=300} \\
\hline
 30 & 1 &  &  && &  &  0.049 & 0.197 & 0.203 &  &  0.026 & 0.177 & 0.179 \\
  2 & 20 &  &  && &  &  0.036 & 0.106 & 0.112 &  &  0.034 & 0.099 & 0.104 \\
   \cdashline{8-14}
  E & 1 &  &  && &  &  0.054 & 0.396 & 0.400 &  &  0.001 & 0.364 & 0.364 \\
  2 & P &  &  && &  &  0.021 & 0.105 & 0.107 &  &  0.019 & 0.097 & 0.099 \\
  2 & D &  &  && &  & -0.216$^*$ & 0.506$^*$  & 0.550$^*$ &  & -0.133 & 0.258 & 0.290 \\

\hline
\multicolumn{14}{c}{n=p=1500} \\
\hline
 30 & 1 &  &  && &  & 0.000 & 0.085 & 0.085 &  &  0.004 & 0.080 & 0.080 \\
  2 & 20 &  &  && &  & 0.018 & 0.049 & 0.052 &  &  0.020 & 0.047 & 0.051 \\
   \cdashline{8-14}
  E & 1 &  &  && &  & 0.027 & 0.169 & 0.171 &  &  0.024 & 0.163 & 0.165 \\
  2 & P&  &  && &  &  0.008 & 0.045 & 0.046 &  &  0.011 & 0.043 & 0.045 \\
  2 & D&  &  && &  & -0.031 & 0.107 & 0.111 &  & -0.040 & 0.116 & 0.123 \\

%
\hline
\hline
\end{tabular}
$^*$One divergent iteration was removed from the summary.
\end{threeparttable}
\end{center}
\end{table}

We  also investigated the average testing deviance defined in the Section \ref{section:ext} and the {\it magnitude of  estimation} under the possibly misspecified models
\begin{align*}
   \magjb^2 (\hbeta ) = &
   \int_0^\tau \hbeta^\top\E_*\left[ \left\{\bZ_*-\mu(t)\right\}^{\otimes 2} Y_*(t)\right] \hbeta dt
   ,   \
   \magjg\left(\hgr\right) =
    \left[\E_*\left\{w^0_*\left(\hgr\right) X_*\right\}\right]^{-1}
   + \left[\E_*\left\{w^1_*\left(\hgr\right) Y_*(\tau)\right\}\right]^{-1},
\end{align*}
with $\mu(t) =\E_*(\bZ_*)/\E_*\{Y_*(t)\}$.
 We used the sample average $n^{-1}\sum_{j=1}^k\sum_{i\in\fold[k]}$
to approximate the expectation $\E_*$ for the former.
Under the Scenario `E', we evaluated
\begin{equation}
  \widehat{\mseb} =
  n^{-1}\sum_{j=1}^k\sum_{i\in\fold[k]}\left[\int_0^\tau \left\{\cbjt\bZ_i
  - \exp(\bbeta_0^\top\bZ_i)\right\}^2Y_i(t)dt\right]
\label{eq:hatmse-E}.
\end{equation}
For the latter we used the sample average $k/n\sum_{i\in\fold[k]}$
to approximate the expectation $\E_*$,
and took the maximum across all folds.
In Table \ref{tab:nuisance}, we present the estimation error, deviance and magnitude  of the nuisance parameters.
The uniform error columns contain the estimation errors from Lasso in $l_1$-norm
and the Breslow estimator in $l_\infty$-norm; these are compared to the true parameters
when the models are correctly specified.  
The deviance columns contain the mean over simulation runs
of the empirical deviances.
The magnitude columns contain the median over simulation runs
of the empirical magnitudes.
When the true model is dense, in our setups either $s_\beta = 30$ or $s_\gamma = 20$, we observe from Table \ref{tab:nuisance} that the Lasso estimators deviated substantially
from the underlying true coefficients in $l^1$-norm,
 suggesting that the Lasso was not concentrated around the true parameters.
Regardless, our proposed method achieved consistent estimation of the treatment effect.
When the Assumption \ref{assume:varD} held, 
i.e.~for all scenarios except `D',
the magnitudes were well controlled,
with no indication that the magnitudes might blow up in high dimensions.

\begin{table}
\caption{Estimation error, deviance and magnitude of the nuisance parameters; the settings are as described for the previous two tables.}\label{tab:nuisance}
\begin{center}
\scriptsize
\begin{threeparttable}
\begin{tabular}{llllllllllll}
\hline
\hline
& & \multicolumn{5}{c}{Outcome Model} && \multicolumn{4}{c}{Treatment Model} \\
\cline{3-7} \cline{9-12}

\\[-8 pt]
n \& p& &$s_\beta$  & $\hbeta$ in $l_1$ &  $\hHaz$ in $l_\infty$&
$\widehat{\mseb}\left(\cbeta\right)$ & $\widehat{\magjb}\left(\cbeta\right)$ &
&$s_\gamma$&$\hgr$ in $l_1$ & $\widehat{\mseg}\left(\cgr\right)$
&  $\widehat{\magjg}\left(\cgr\right)$ \\
\hline
 300 &  & 2 & 0.61 & 0.27 & 0.34 & 0.37 &  & 1 & 1.38 & 0.09 &   36.73 \\
  1500 &  & 2 & 0.42 & 0.13 & 0.20 & 0.41 &  & 1 & 0.80 & 0.05 &   29.78 \\
   \hline
300 &  & 6 & 1.13 & 0.31 & 0.46 & 0.35 &  & 3 & 1.75 & 0.10 &   29.84 \\
  1500 &  & 6 & 0.90 & 0.17 & 0.27 & 0.40 &  & 3 & 1.13 & 0.05 &   23.49 \\
   \hline
300 &  & 15 & 2.87 & 0.46 & 0.67 & 0.38 &  & 10 & 2.97 & 0.12 &   35.11 \\
  1500 &  & 15 & 2.43 & 0.28 & 0.45 & 0.44 &  & 10 & 2.02 & 0.07 &   27.16 \\
   \hline
300 &  & 30 & 6.22 & 0.60 & 0.96 & 0.42 &  & 20 & 5.01 & 0.15 &   48.72 \\
  1500 &  & 30 & 5.19 & 0.36 & 0.63 & 0.48 &  & 20 & 3.75 & 0.09 &   34.97 \\
   \hline
300 &  & E & -- & -- & 0.71 & 0.58 &  & P & -- & 0.09 &   29.19 \\
  1500 &  & E & -- & -- & 0.47 & 0.65 &  & P & -- & 0.05 &   22.91 \\
   \hline
300 &  &  & &  &  & &  & D & -- & 0.17 & $>100$\tnote{**} \\
  1500 &  &  & & &  &  &  & D & -- & 0.13 & $>100$\tnote{**}\\
   \hline
\hline
\end{tabular}
\begin{tablenotes}
\item[*] The dashed entries are not well-defined due to
misspecification;
\item[**] The divergence of magnitude is expected because
setup ``D" violates Assumption \ref{assume:varD}.
\end{tablenotes}
\end{threeparttable}
\end{center}
\end{table}


\section{Illustration with SEER-Medicare data}\label{section:data}

Clinical databases such as the United States National Cancer Institute's Surveillance, Epidemiology, and End Results (SEER) typically contain disease specific variables, but with only limited information on the subjects' health status, such are comorbidities for example. In studying causal treatment effects, this can lead to unobserved confounding \citep{HadleyEtal10JNCI,YingEtal19}. On the other hand, the availability of information from insurance claims databases could make up for some of these otherwise 'unobserved' variables;   \cite{RiviereEtal19}  shows that they contain much information about comorbidities in particular.

For prostate cancer while radical prostatectomy is quite effective in reducing cancer related deaths.
With improvement in diagnosis, treatment and management for the disease
other causes have become  dominant 
for the overall death of this patient population. 
 Studying the comparative effect of radical prostatectomy on
overall survival of the patient using  observational data calls for proper control of confounding.
Due to the lack of tools for handling the high-dimensional claims code data,
existing work on the topic either have not made use, 
or made very limited 
through summary statistics, 
of the rich information on the patients'
health status.

We considered 49973 prostate cancer patients diagnosed during 2004-2013
as recorded in the SEER-Medicare linked database.
The data contained the survival times of the patients,  treatment information, demographic information,
clinical variables and the federal Medicare insurance claims codes.
More specifically, we included in our analysis age, race, martital status, tumor stage, tumor grade,
prior Charlson comorbidity score,
and 6397 claims codes possessed by at least 10 patients
during the 12 months before their diagnosis of prostate cancer.
The summary statistics of these variables
are presented in Table \ref{tab:descr}.
Our main focus is the treatment effect of surgery (radical prostatectomy) on the overall survival of the patients.
In our sample, 17614 (35.25 \%) patients received surgery
while 32359 (64.75 \%) patients received other types of treatment without surgery.
As can be seen, many of the variables are not balanced between the surgery and no surgery groups;
in particular, patients who received surgery tended to be younger, married, white, T2, poorly differentiated tumor grade, zero comorbidity, and diagnosed in 2008 or earlier.
Among all  patients 5375 (10.76 \%) deaths were observed while the rest of the patients were still alive
by the end of year 2013.
The Kaplan-Meier curves for the two groups
are presented in Figure \ref{fig:KM}.

\begin{figure}[h]
\centering
\includegraphics[width=0.45\textwidth]{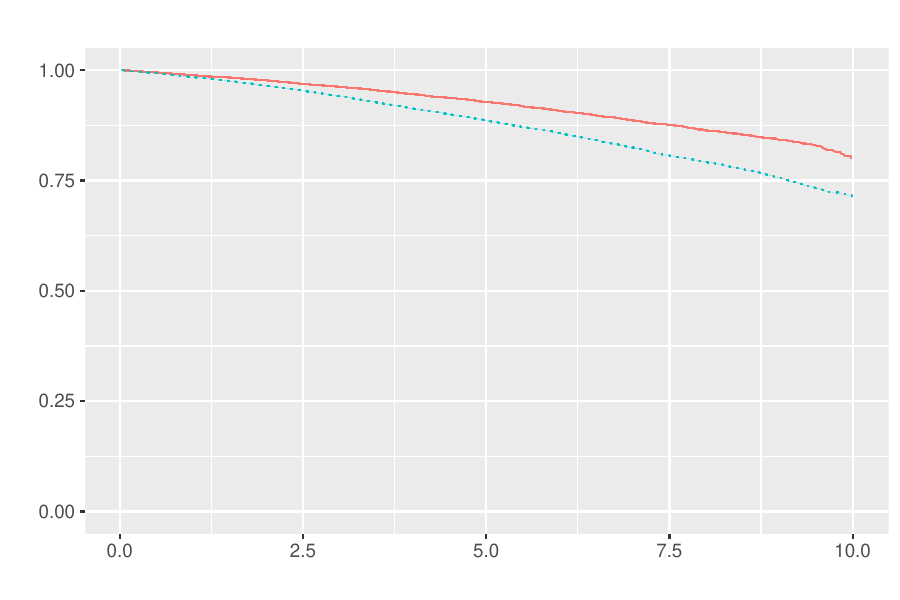}
\caption{Kaplan-Meier curves
for surgery (solid red) versus no surgery (dashed blue).
Y axis: survival probability; X axis: year(s) since diagnosis. }\label{fig:KM}
\end{figure}

\begin{table}[h]
\caption{Summary of the  linked SEER-Medicare data }\label{tab:descr}
\begin{center}
\scriptsize
\begin{tabular}{l l c c}
\hline
\hline
& &  No Surgery &  Surgery\\
\hline
Variable & Value & $n=32359$ &  $n=17614$ \\
\hline
 Age & 66-69 & 12460 (38.5 \%) &  8790 (49.9 \%) \\
   & 70-74 & 19899 (61.5 \%) &  8824 (50.1 \%) \\
   \hline
Marital status & Married & 21464 (66.3 \%) & 13439 (76.3 \%) \\
   & Divorced &  2200 ( 6.8 \%) &   820 ( 4.7 \%) \\
   & Single &  2490 ( 7.7 \%) &  1207 ( 6.9 \%) \\
   & Other &  6205 (19.2 \%) &  2148 (12.2 \%) \\
   \hline
Race & White & 26019 (80.4 \%) & 15035 (85.4 \%) \\
   & Black &  4501 (13.9 \%) &  1527 ( 8.7 \%) \\
   & Asian &   467 ( 1.4 \%) &   284 ( 1.6 \%) \\
   & Hispanic &   327 ( 1.0 \%) &   204 ( 1.2 \%) \\
   & Other &  1045 ( 3.2 \%) &   564 ( 3.2 \%) \\
   \hline
Tumor stage & T1 & 20314 (62.8 \%) &  3866 (21.9 \%) \\
   & T2 & 12045 (37.2 \%) & 13748 (78.1 \%) \\
   \hline
Tumor grade & Well differentiated &   381 ( 1.2 \%) &   214 ( 1.2 \%) \\
   & Moderately differentiated & 16549 (51.1 \%) &  7340 (41.7 \%) \\
   & Poorly differentiated & 15374 (47.5 \%) & 10024 (56.9 \%) \\
   & Undifferentiated &    55 ( 0.2 \%) &    36 ( 0.2 \%) \\
   \hline
Prior Charlson & 0 & 20238 (62.5 \%) & 11890 (67.5 \%) \\
  comorbidity score & 1 &  7067 (21.8 \%) &  3699 (21.0 \%) \\
   & $\ge 2$ &  5054 (15.6 \%) &  2025 (11.5 \%) \\
   \hline
Year & 2004 &  3076 ( 9.5 \%) &  1674 ( 9.5 \%) \\
   & 2005 &  3003 ( 9.3 \%) &  1653 ( 9.4 \%) \\
   & 2006 &  3365 (10.4 \%) &  1879 (10.7 \%) \\
   & 2007 &  3419 (10.6 \%) &  2027 (11.5 \%) \\
   & 2008 &  3315 (10.2 \%) &  1937 (11.0 \%) \\
   & 2009 &  3382 (10.5 \%) &  1843 (10.5 \%) \\
   & 2010 &  3315 (10.2 \%) &  1884 (10.7 \%) \\
   & 2011 &  3568 (11.0 \%) &  1924 (10.9 \%) \\
   & 2012 &  2964 ( 9.2 \%) &  1430 ( 8.1 \%) \\
   & 2013 &  2952 ( 9.1 \%) &  1363 ( 7.7 \%) \\
   \hline

   \hline
Total claims codes & Mean (SD) & 44.3 (34.0) & 45.9 (31.7) \\

\hline
\hline
\end{tabular}
\end{center}
\end{table}

The causal diagrams of the analyses are illustrated in Figure \ref{fig:diag}.
In  Analysis I, we adjusted for the potential confounding effects
from the clinical and demographic variables and the high-dimensional claims codes.
After excluding claims codes with less than 10 occurrences,
we had $6533$ covariates in Analysis I.
For the additive hazards model,
we were under the `$p>n$' scenario as the number of covariates exceeded
the number of observed events $5375$.

\begin{figure}[h]
\centering
\begin{subfigure}[b]{0.8\textwidth}
\centering
\includegraphics[width=\textwidth]{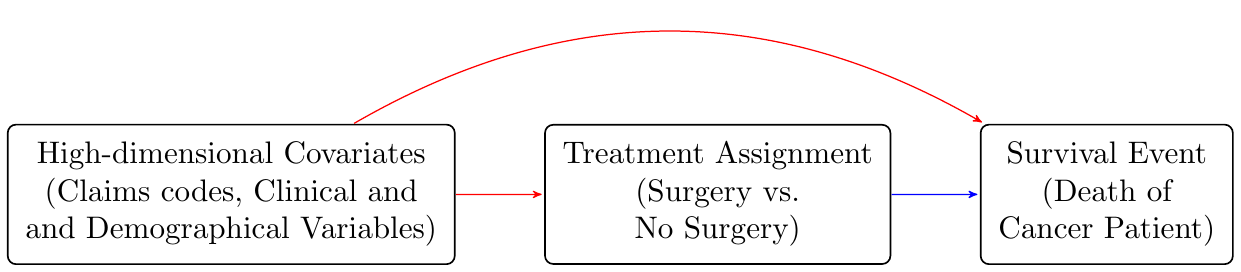}
\caption{Analysis I: adjust for the clinical and demographic variables and the claims codes. }
\label{fig:diag-I}
\end{subfigure}
\begin{subfigure}[b]{0.8\textwidth}
\centering
\includegraphics[width=\textwidth]{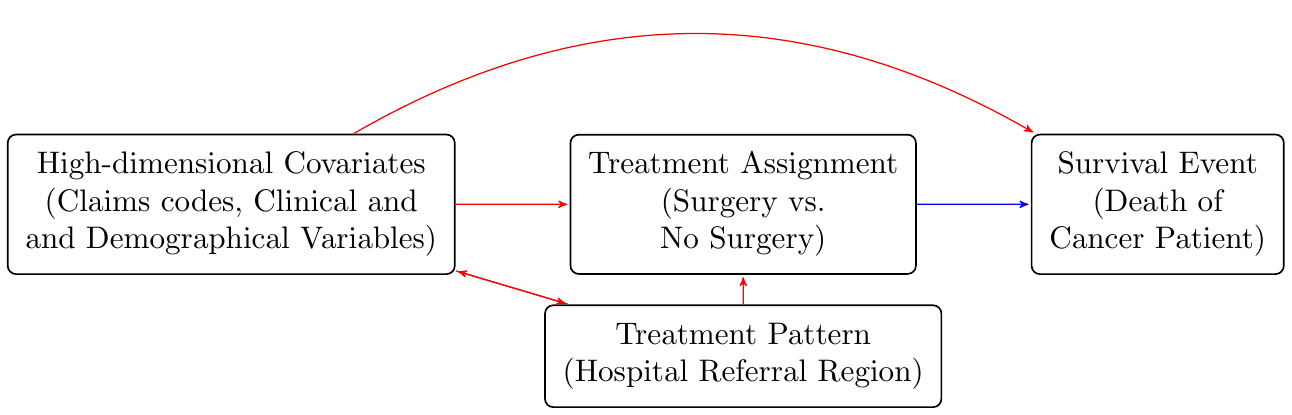}
\caption{Analysis II: accounting for heterogeneity in
treatment pattern as reflected in hospital referral regions (HRR) and its
interactions with the covariates from Analysis I in the PS model.
 }
\label{fig:diag-II}[h]
\end{subfigure}
\caption{Causal diagrams of the two analyses.
}\label{fig:diag}
\end{figure}


In Figure \ref{fig:hist-PS}, we plot the distribution
of the estimated propensity scores from both groups.
We note that the range of PS for the surgery group was $0.04 \sim 0.93$,
and for the no surgery group was $0.03 \sim 0.90$.

\begin{figure}
\centering
\begin{subfigure}[t]{0.45\textwidth}
\centering
\includegraphics[height=0.67\textwidth]{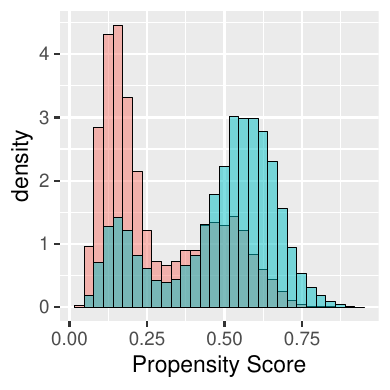}
\caption{Analysis I}
\end{subfigure}
\begin{subfigure}[t]{0.45\textwidth}
\centering
\includegraphics[height=0.67\textwidth]{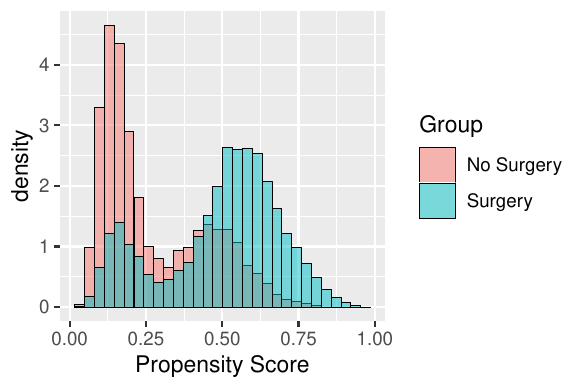}
\caption{Analysis II}
\end{subfigure}
\caption{Distribution of the estimated propensity scores.
}\label{fig:hist-PS}
\end{figure}

We report the analysis results in Table \ref{tab:analysis}.
Low dimensional covariates methods include
the crude analysis  without adjusting for any covariates; the regression adjustment directly adjusted for the
clinical and demographical variables in the additive hazards model;
and the inverse probability weighting (IPW) with propensity score estimated
using R package `\emph{twang}'.
 When including the  claims codes data,
we report the results of
 the naive additive hazards model Lasso estimate $\tilde\theta$
that did not penalize the treatment effect as well as our four methods.

\begin{table}[h]
\footnotesize
\caption[Estimated treatment effect ($\times 10^{-3}$) from the linked SEER-Medicare data. ]{Estimated treatment effect ($\times 10^{-3}$) from the linked SEER-Medicare data.
Crude analysis did not adjust for any covariates.
Lasso estimator $\tilde{\theta}$ penalized only the covariate effects $\bbeta$ but not $\theta$.
$\hth$, $\cth$, $\hth_{cf}$ and $\cth_{cf}$ are
the four proposed estimators, where the subscript `cf' denotes the cross-fitting. }\label{tab:analysis}
\begin{center}
\begin{threeparttable}
\begin{tabular}{p{2.5in}  c c c c }
\hline
\hline
Approach & Estimate & SE  & 95 \% CI
& $p$-value\\
\hline
\multicolumn{5}{l}
{Low-dimensional analysis: 
$p=23$.}\\
 Crude  & -9.971 & 0.605 & [ -11.157 , -8.786 ] & $< 0.001$ \\
  Regression adjustment & -6.151 & 0.722 & [  -7.567 , -4.735 ] & $< 0.001$ \\
  IPW with  PS & -4.408 & 0.757 & [  -5.893 , -2.924 ] & $< 0.001$ \\

\hline
\multicolumn{5}{l}
{Analysis I: 
$p=6533,\hat{s}_\beta = 378,
\hat{s}_\gamma = 309$.
}\\

 Lasso $\tilde\theta$\tnote{*} & -5.598 & -- & -- & -- \\
  $\hth$  & -4.193 & 0.730 & [  -5.624 , -2.761 ] & $< 0.001$ \\
   $\cth$ & -4.187 & 0.730 & [  -5.619 , -2.756 ] & $< 0.001$ \\
  $\hth_{cf}$  & -3.851 & 0.733 & [  -5.288 , -2.414 ] & $< 0.001$ \\
   $\cth_{cf}$  & -4.310 & 0.732 & [  -5.746 , -2.875 ] & $< 0.001$ \\

\hline
\multicolumn{5}{p{5.5in}}
{Analysis II: 
$p_\beta = 6533, \hat{s}_\beta = 378,
p_\gamma=43466,\hat{s}_\gamma = 883$.}\\

 Lasso $\tilde\theta$\tnote{*} & -5.598 & -- & -- & -- \\
  $\hth$  & -4.151 & 0.735 & [  -5.591 , -2.711 ] & $< 0.001$ \\
   $\cth$ & -4.149 & 0.735 & [  -5.589 , -2.708 ] & $< 0.001$ \\
  $\hth_{cf}$  & -3.784 & 0.738 & [  -5.231 , -2.337 ] & $< 0.001$ \\
  $\cth_{cf}$  & -4.271 & 0.738 & [  -5.717 , -2.826 ] & $< 0.001$ \\

\hline
\end{tabular}
\begin{tablenotes}
\item[*] Inference is not directly available. Only the estimates are reported.
\end{tablenotes}
\end{threeparttable}
\end{center}
\end{table}

All analysis results suggest that surgery improved overall survival
compared to no surgery. The magnitude of the estimated treatment effect, however,
 varied according to the approach used. The crude analysis had the largest estimated treatment effect of almost 0.01 reduction in hazard rate. Regression adjustment gave an estimated reduction of 0.006, while IPW with low dimensional PS, as well as the four proposed estimators gave  reduction around  0.004.

In addition to the above, heterogeneity in treatment pattern has been noticed
across geographic regions \citep{HarlanEtal01}.
In our data,  geographic region is described by the
hospital referral region (HRR).
%
In Analysis II, we included in  the treatment propensity
score model all the interactions between HRR
and the covariates in Analysis I.
After excluding claims codes and binary interaction terms
with less than 10 occurrences,
we have $43466$ covariates for the PS model in Analysis II.
The large number of additional interaction terms
puts the  PS model in the $p \approx n$ scenario.
However, the  results were numerically stable and quantitatively similar to Analysis I above, despite the dramatically increased number of covariates.


\section{Beyond simple additive hazards}\label{section:new}


\subsection{Modern nonparametrics estimates}\label{section:ml}


With the advancement in machine learning (ML), nonparametric estimation under the general models \eqref{model:Dgen}
and \eqref{model:aalen_pl} becomes both feasible and even attractive.
The use of nonparametric estimation in  causal inference has been considered
under settings without censoring \citep{ChernozhukovEtal17,FarrellEtal19}.
For censored survival data,  hazard function is often used in the specific models.
Here we
divide into two types of approaches, based on whether
the nonparametric estimation is on the hazard scale
or on the cumulative hazard scale.
The former is an immediate extension of Section \ref{section:cf}, while the latter has more readily available
 methodologies and theories established in the literature.

\subsubsection*{\underline{Machine-Learning of the hazard function}}

There are multiple ways that the settings of \eqref{model:aalen} and \eqref{model:D} can be  extended.  A natural extension would consider
\begin{equation}\label{eq:aalan1}
   \P(D=1|Z) = m(Z), \;  \lambda(t|D,Z) = \lambda_0(t) + D\theta + g(Z),
\end{equation}
for  unknown functions $g$ and $m$.
Suppose we have consistent estimators $\hat{g}$ and $\hat{m}$ satisfying
$$
\sqrt{\E_*[\hat{g}(Z_*) - g(Z_*)]^2\E_*[\hat{m}(Z_*) - m(Z_*)]^2} = o_p\left(n^{-1/2}\right).
$$
Using $\hat{g}$, we may estimate the baseline cumulative hazard by
$$
\hHaz(t;\theta) = \int_0^t \frac{\sum_{i=1}^n[dN_i(u)-Y_i(u) \{\theta+\hat{g}(Z_i)\} du]}{
    \sum_{i=1}^n Y_i(u)}.
$$
The orthogonal score \eqref{score:inference}   takes the form
\begin{equation}\label{score:ml-haz}
 \scorei(\theta; m, g, \Haz) = \int_0^\tau \exp{(D\theta t)}  \left \{D-m(\bZ)\right \}\left[
   dN(t) - Y(u)\{\theta dt + g(Z) + d \Haz(t;\theta)\}\right].
\end{equation}
We can estimate the treatment effect $\theta$ with the score \eqref{score:ml-haz} through the crossfitting scheme in Section \ref{section:cf}

While the nonparametric estimation of hazard with high-dimensional covariate is largely an  open problem,
we offer one such solution. 
If we ponder the existence of an orthonormal basis $\{ f_j\}_{j \in \mathbb{N}}$ of $L^2([-K,K]^p)$,  and a $\kappa$ and a permutation $\pi$ of the first $\kappa$ elements of the basis such that
\[
g(Z) = \sum_{j=1}^\kappa \beta_j f_j(Z)  + \sum_{j=\kappa+1}^\infty \beta_j f_j(Z).
\]
Then, one can estimate $g$ by any series estimator $\hat g(Z)$ using a  least-squares regression  on the first $\kappa$ elements  $f_1(Z),\dots, f_\kappa(Z)$, or a  least-squares regression of the likes of \eqref{init:beta} on a large number of elements of the basis but now with possibly group regularizations. 
Observe that in the above, $\kappa$  plays an equivalent role to that of the sparsity level.
Similarly we can obtain an estimate $\hat m(Z)$.

\subsubsection*{\underline{Machine-Learning of the cumulative hazard function}}

A    generalization of  model   \eqref{eq:aalan1}   considers a  further loosening of  the additive  assumption between $Z$ and $\lambda_0(t)$:
\begin{equation}\label{eq:aalan2}
   \P(D=1|Z) = m(Z), \;  \lambda(t|D,Z) =   D\theta + g(t;Z),
\end{equation}
as in \eqref{model:Dgen} and \eqref{model:aalen_pl}.
However, most of the ML approaches for cumulative hazard cannot effectively isolate the time-independent effect $\theta$, and are mostly designed to estimate the cumulative hazard
  $$G(t;Z) = \int_0^t g(t;Z) dt.$$
The HDi type estimator, as specified in \eqref{check:contrast}, is no longer generally applicable.  Below we propose a different approach that mimics the orthogonal score $\phi_{i}$.

The orthogonal score \eqref{score:inference}   takes the form
\begin{equation}\label{score:ml}
  \scorei_{G}(\theta; m, G) =\int_0^\tau \exp{(D\theta t)}  \left \{D-m(\bZ)\right \}\left[
   dN(t) - Y(u)\{\theta dt + d G(t;\bZ)\}\right].
\end{equation}
We proceed by estimating the cumulative hazard functions separately in each of the two groups.
Call these estimates $\hat G _0 (t;Z)$ and $\hat G_1(t;Z)$ corresponding to $D=0$ and $1$, respectively. Some survival ML approaches for achieving this include \cite{CuiEtal19pArXiv,LiBradic19pArXiv,IshwaranEtal08,IshwaranKogalur10,IshwaranEtal10,katzman2018deepsurv,BradicRava20}.
Whenever the machine learning method produces consistent estimates
we expect to have
$$
\hat{G}_{0}(t;Z) \to G(t;Z), \text{ and }
\hat{G}_{1}(t;Z)  \to \theta_0 t + G(t;Z).
$$
As score \eqref{score:ml} requires a unified estimator of the cumulative hazard, we provide a simple estimator that combines $\hat{G}_{0} $ and $\hat{G}_{1} $:  
\begin{equation}\label{def:ghat}
 \hat{G}(t;Z, \theta)  =\bar{D}(t)\{\hat{G}_{1}(t;Z)-\theta t\} + \{1-\bar{D}(t)\}\hat{G}_{0}(t;Z),
\end{equation}
where
\begin{equation}
  \bar{D}(t) = \sum_{i=1}^{n} D_i Y_i(t) / \sum_{i=1}^{n} Y_i(t).
\end{equation}
 We present the cross-fitting strategy in Algorithm \ref{alg:ml}.
We expect that under conditions mimicking those of Section \ref{section:cf},
we can guarantee asymptotically normal estimator of $\theta$; i.e., $\hat \theta_{ml}$ of \eqref{hat:aalen-ate-ml} satisfies
\begin{equation*}
  \hsig_{ml}^{-1}\sqrt{n}(\atei_{ml}-\theta_0) \leadsto N(0,1),
\end{equation*}
with
\begin{equation*}
  \hsig_{ml}^2 = \frac{n^{-1}\sum_{j=1}^{k} \sum_{i \in \fold[j]}\delta_i \{D_i-\hat{m}^{(j)}(Z_i)\}^2e^{2\atei_{ml} D_i X_i}}
  {\left\{n^{-1}\sum_{j=1}^{k} \sum_{i \in \fold[j]} (1-D_i)\hat{m}^{(j)}(Z_i)X_i\right\}^2},
\end{equation*}
as long as the ML estimates satisfy
$$
\sqrt{\E_*\left[\left\{\hat{G}\left(X_*;Z_*\right)
- G\left(X_*;Z_*\right)\right\}^2\mid D_*=0 \right]\E_*[\{\hat{m}(Z_*) - m(Z_*)\}^2]}
= o_p\left(n^{-1/2}\right). 
$$

\begin{algorithm}[]
 \caption{Estimation of the Treatment Effect  with Cross-fitted
  Machine-learning estimates.  }\label{alg:ml}
\SetAlgoLined
\KwData{split the data into $k$ folds of equal size with the indices set $\fold[1], \fold[2], \dots, \fold[k]$; stratify the data into two treatment arms with the indices set $\fold[trt]$ and $\fold[ctr]$. }
\For{each fold indexed by $j$}{
1. estimate the propensity model $\hat{m}^{(j)}(z)$ using the out-of-fold (out of $j$-th fold)  data indexed by $\fold[-j] = \{1,\dots, n\}\setminus\fold[j]$ \;
2. estimate the conditional cumulative hazards $\hat{G}^{(j)}_{1}(t;z)$ and $\hat{G}^{(j)}_{0}(t;z)$ with the two out-of-fold treatment arms
$\fold[-j] \cap \fold[trt]$ and $\fold[-j] \cap \fold[ctr]$\;
3. denote $\bar{D}^{(j)}(t) = \sum_{i \in \fold[-j]} D_i Y_i(t) / \sum_{i \in \fold[-j]} Y_i(t)$
and construct the estimator for the partially linear additive hazards model as
$$\hat{G}^{(j)}(t;z, \theta) = \bar{D}^{(j)}(t)\{\hat{G}^{(j)}_{1}(t;z)-\theta t\}+ \{1-\bar{D}^{(j)}(t)\}\hat{G}^{(j)}_{0}(t;z);$$
4.  construct the cross-fitted score using the in-fold samples:
\begin{align}
  &  \scorei_{G,i}\left(\theta; \hat{m}^{(j)} ,\hat{G}^{(j)}_{1}, \hat{G}^{(j)}_{0}\right)
     \notag \\
 &   \qquad = \int_0^\tau e^{D_i\theta t} \left\{D_i - \hat{m}^{(j)}(\bZ_i)\right\}  \Bigg(dN_i(t) \notag \\
&
\qquad \qquad - Y_i(t)\Big[\theta  \{D_i - \bar{D}^{(j)}(t) \}dt +\bar{D}^{(j)}(t)d\hat{G}^{(j)}_{1}(t;\bZ_i)+
( 1-\bar{D}^{(j)}(t) ) d\hat{G}^{(j)}_{0}(t;\bZ_i)\Big] \Bigg). \label{score:setj-ml}
\end{align}
}
\KwResult{Obtain the estimated treatment effect $\atei_{ml}$ by solving
    \begin{equation}\label{hat:aalen-ate-ml}
     \frac{1}{k} \sum_{j=1}^{k}\frac{1}{|\fold[j]|} \sum_{i \in \fold[j]}  \scorei_{G,i}\left(\theta; \hat{m}^{(j)} ,\hat{G}^{(j)}_{1}, \hat{G}^{(j)}_{0}\right)  = 0.
    \end{equation}
}
\end{algorithm}

\newpage
\subsection{Heterogeneous effects}\label{section:ite}

So far we have made the assumption that the treatment effect is homogeneous given the covariates in model \eqref{model:aalen}.
It is natural to consider the average treatment effect (ATE) when heterogeneity exists among
treatment effect for each individual.
In the setting without censoring, the augmented inverse probability weighting (AIPW) is often used to
study the ATE with hetergenous treatment effect \citep{RobinsEtal94,ChernozhukovEtal17}.
The typical AIPW cannot be directly applied to survival data as the missingness from censoring
can trigger confounding by itself.
Even in low-dimensional scenario, existing doubly robust estimation methods of ATE with survival data
require fully parameteric heterogeneous treatment effect as interaction terms \citep{JiangEtal17}, consistent estimation of censoring distribution \citep{ZhaoEtal15} or both \citep{LuEtal16}.
 Incorporation of the estimation of censoring distribution, currently not a part of our methodology, requires  systematic changes beyond the scope of the paper.
Therefore, we limit our extension of model \eqref{model:aalen} to include interaction terms
with treatment:
\begin{equation}\label{model:hetero}
   \haz(t; D,\bZ) = \lambda_0(t) + (\theta + \alpha^\top W) D + \beta^\top Z,
\end{equation}
where  $W \in \R^q$ contains some known functions of $\bZ$. The choice of $W$ is typically informed by the scientific context, and we assume that the dimension of $W$ is fixed, $q \ll n$.
The goal is now to draw inference jointly on $(\theta, \alpha^\top)$.

Under model \eqref{model:hetero} the orthogonal score for $(\theta,\alpha^\top)^\top$ is:
\begin{equation}\label{score:hetero}
 \phi_{h}(\theta, \alpha; \bbeta,\Haz,\bgr) = \int_0^\tau
 \exp\{(\theta +\alpha^\top W)D t\} \tilde{W}  \left (D-\expit(\bgr^\top \bZI[])\right )d M(t;\theta, \alpha, \bbeta, \Haz),
\end{equation}
where $\tilde{W} = (1,W^\top)^\top$, and
$$
M(t;\theta, \alpha, \bbeta, \Haz) =
 N(t) - \int_0^t Y_i(u)[d\Haz(t) + \{ (\theta + \alpha^\top W)D + \beta^\top Z \}dt ].
$$
By orthogonality and under conditions similar to Assumption \ref{assume:inf},
the solution $(\hat{\theta}, \hat{\alpha})$ to the system of equations
$$\frac{1}{n}\sum_{i=1}^{n} \phi_{h,i}(\theta, \alpha; \hbeta,\hHaz,\hgr) = 0$$
  is asymptotically normal. Let us emphasize that here the nuisance parameters, $\hat \beta$, $\hat \Lambda$ and $\hat \gamma$, need to be estimated on a hold-out dataset whereas the solution to the score equation above needs to be found on the remaining data, such as in the cross-fitting scheme. The estimate of the variance  can be found using  $ \hat{\Sigma}^{-1} \hat{\mathcal{V}}\hat{\Sigma}^{-1}/n$, where
\begin{eqnarray*}
\hat{\Sigma} &=& \frac{1}{n}\sum_{i=1}^n \int_0^\tau e^{D_i(\hat{\theta}+\hat{\alpha}^{\top} W_i) t}  \{D_i-\expit(\hgr^\top \bZI[i])\}\{D_i - \bar{D}^{(j)}(t)\} \tilde{W}_i  \tilde{W}_i^\top
  Y_i(t)dt, \\
  \hat{\mathcal{V}} &=& \frac{1}{n}\sum_{i=1}^n \delta_i \{D_i-\expit(\hgr^\top \bZI[i])\}^2  \tilde{W}_i \tilde{W}^\top_i e^{2D_i(\hat\theta +\hat{\alpha}^{\top} W_i)}.
\end{eqnarray*}

The above interaction terms are special cases of heterogeneous treatment effects. More generally, heterogeneous treatment effects may be written $\theta(W)$, or $\theta(Z)$. It is not straightforward to extend our orthogonal score method here to estimate a nonparametric $\theta(\cdot)$ or its average $\E\{\theta(Z)\}$, and possibly with high-dimensional $Z$.  These are, however,  important problems for future research.

\section{Concluding Remarks}\label{section:discuss}

In this paper, we have developed a novel orthogonal score-based approach under the additive hazard model with high-dimensional covariates. Thus, the resulting estimate of the treatment effect is consistent and asymptotically normal at root-$n$ rate even with (biased) input from regularized regression. We have presented the results with high-level assumptions such as Assumption 1 and discussed under which sparsity settings our high-level assumptions are satisfied for Lasso type estimators. Such discussion holds,  equivalently,  for any regularized estimators for which rates of estimation have been established: non-convex penalties, group, and hierarchical or smoothing penalties, etc.  On the other hand, corresponding rates of estimation have not been manifested for modern machine learning methods, including boosting, random forests, or neural networks; some of those methods have yet to be developed for censored data.


We highlight that the notion of double robustness, though clearly defined in low-dimensional problems, is no longer uniquely defined in the high-dimensional setting. Although semiparametric efficiency theory can be used to design doubly robust scores in low dimensions, as in \cite{Dukes:etal:19}, the resulting scores   
do not necessarily satisfy any of the notions of
double robustness in high-dimensional settings. The problem here becomes rich with potentially many different aspects.
We have shown that our cross-fitted orthogonal score has the {rate double robustness} property.

The above notion of rate 
double robustness assumes that the models are, in fact, correctly specified.
\cite{Tan18} and  \cite{SmuclerRotnitzkyRobins19} relaxed this assumption for what is referred to as `model double robustness'. As pointed out by a reviewer, their  assumptions
 are very similar  and  are essentially that the parameter defined as the
population minimizer of the (possibly misspecified) loss function
 is suitably sparse. This in turn
implies that the corresponding $L_1$ regularized estimators
of the nuisance functions converge somewhere, and that it is in fact quite strong an assumption. 
   Simulation studies presented in this paper indicate that our procedure has the potential to be valid in similar aspects.
In \cite{HouEtal_arxiv}, we show that consistent estimation 
 can 
be achieved {when one (but not both) of the two models 
 is misspecified.}
   However, inference under those settings requires substantial future work on the suitable estimation process for the nuisance parameters.
   
\processdelayedfloats

 \bibliography{ATEsurvHD-r1}

\newpage

\newpage 

\appendix

  \begin{center}
    {\LARGE\bf Supplementary Materials}
\end{center}
  \medskip
  
{\centering{This document contains  complete  proofs of the theoretical claims made in the main text.    Appendix \ref{appendix:main_proof} delineates proofs of Theorems whereas auxiliary results, including classical and new concentration results,
are stated and proved in Appendix \ref{appendix:auxiliary}.
}}

\renewcommand\thefigure{A\arabic{figure}}
\renewcommand\theequation{A.\arabic{equation}}
\renewcommand\thetable{A\arabic{table}}
\renewcommand\thelemma{A\arabic{lemma}}
\renewcommand\thesubsection{\Alph{section}\arabic{subsection}}
\renewcommand\thedefinition{A\arabic{definition}}
\renewcommand\theremark{A\arabic{remark}}

\setcounter{equation}{0}
\setcounter{lemma}{0}

\section{Proof of the main results}\label{appendix:main_proof}
We give the proofs of the Theorems and Lemmas in the order of appearance in the main text.
We begin by introducing useful notation used throughout the document.
 Notice that the Assumption \ref{assume:varD} implies
  \begin{gather*}
  \E\Bigl[ \mathbb{P}\bigl(X \ge \tau |\bZ;D=0\bigl)\Var(D|\bZ) \Bigl] \ge e^{-(\theta_0 \vee 0)\tau}\varepsilon^2>0, \\
  \E\Bigl[ \mathbb{P} ( T \leq C, X \le \tau |\bZ;D=0)\Var(D|\bZ)\Bigl]  \ge e^{-(\theta_0 \vee 0)\tau} \varepsilon^2>0,
  \end{gather*}
  for some $\varepsilon>0$.
For the simplicity of notation, we define two lower bounds as
\begin{equation}\label{eq:assume-varD}
  \begin{gathered}
  \E\Bigl[ \mathbb{P}\bigl(X \ge \tau |\bZ;D=0\bigl)\Var(D|\bZ) \Bigl] \ge \varepsilon_Y>0,\\
  \E\Bigl[ \mathbb{P} ( T \leq C, X \le \tau |\bZ;D=0)\Var(D|\bZ)\Bigl]  \ge \varepsilon_N>0.
  \end{gathered}
\end{equation}

\subsection{
 Proof of Lemma \ref{lem:score} } 

  We first verify the identifiability of the true parameters. At the true parameters $(\theta_0, \bbeta_0, \Haz_0, \bgr_0)$,
 $M(t;\theta_0,\bbeta_0,\Haz_0)$ is a martingale with respect to filtration $\mathcal{F}_{t} = \sigma\{N(u),Y(u),D,\bZ: u\le t\}$.
 Since the other elements $D_i$ and $\bZ$
 are all measurable with respect to  $\mathcal{F}_{n,t}$,
 the martingale integral $\scorei(\theta_0;\bbeta_0,\Haz_0,\bgr_0)$ is also a $\mathcal{F}_{n,t}$-martingale.
 Therefore, $\E\{\scorei(\theta_0;\bbeta_0,\Haz_0,\bgr_0)\} = 0$.

 To show the orthogonality, we define the directional perturbations
 \begin{equation*}
   \bbeta_r = \bbeta_0 + r\dbeta , \,
   \Haz_r(t) = \Haz_0(t) + r\dHaz(t)\text{ and }
   \bgr_r = \bgr_0 + r\dgr.
 \end{equation*}
 We decompose the expected directional derivative in nuisance parameters evaluated at the true parameters
 into 2 terms,
 \begin{align*}
    & \left.\frac{\partial}{\partial r}\E\{\scorei(\theta_0;\bbeta_r,\Haz_r,\bgr_r)\}\right|_{r=0}  \\
   = & -\E\left[\left\{D-\expit(\bgr_0^\top \bZI[])\right\}\int_0^\tau e^{D\theta_0 t}Y(t)\left\{\dbeta^\top \bZ dt+d\dHaz(t)\right\}\right] \\
    & - \E\left[\frac{e^{\bgr_0^\top \bZI[]}}{\left(1+e^{\bgr_0^\top \bZI[]}\right)^2} \dgr^\top \bZI[]\int_0^\tau e^{D\theta_0 t}d M(t;\theta_0, \bbeta_0, \Haz_0)\right].
 \end{align*}
 The effect of treatment $D$ on the conditional expectation of the at-risk process $$\E\{Y(t)|D,\bZ\} = \P(T \ge t|D,\bZ)\P(C \ge t|D,\bZ)$$
 has two components, the effect on the event-time and that on the censoring time.
 Under Assumption \eqref{eq:CindD},
\begin{equation*}
  \P(C \ge t | D,\bZ) =\P(C \ge t |\bZ)
\end{equation*}
is $\sigma\{\bZ\}-$measurable.
Under model \eqref{model:aalen_pl},
\begin{equation*}
  \P(T \ge t | D,\bZ) = e^{-D\theta_0t-\int_0^t g_0(u;\bZ)du}.
\end{equation*}
Therefore, we have the following representation
\begin{equation*}
  \E[e^{D\theta_0 t}Y(t)|D,\bZ] = \P(C \ge t |\bZ)e^{-\int_0^t g_0(u;\bZ)du}
  = \E\{Y(t)|\bZ,D=0\},
\end{equation*}
which is obviously $\sigma\{\bZ\}-$measurable.
%
 Using the fact $\E\left\{D-\expit(\bgr_0^\top \bZI[])|\bZ\right\} = 0$ under model \eqref{model:D},  we apply the tower property of conditional expectation to calculate that the first term equals zero,
 \begin{equation*}
   \int_0^\tau \E\left[\E\left\{D-\expit(\bgr_0^\top \bZI[i])|\bZ\right\} \E\left\{e^{D\theta_0 t}Y(t)|D,\bZ\right\}\left\{\dbeta^\top \bZ dt+d\dHaz(t)\right\}\right]
   = 0.
 \end{equation*}
 The second term is again a $\mathcal{F}_{t}$-martingale, so it also has mean zero. This completes the proof.

\subsection{Proof of Theorem \ref{thm:aalen}} 

Proof of Theorem \ref{thm:aalen}  is split into five Steps.
First, we utilize orthogonality to establish useful decomposition of the score $n^{-1}\sum_{i=1}^{n}\scorei_i$.
Second, we show that $\hth$ is consistent for $\theta_0$.
Next, we establish the asymptotic normality of the score $n^{-1}\sum_{i=1}^{n}\scorei_i$ at true parameter.
Fourth, we obtain the $\sqrt{n}$-tightness and the asymptotic distribution of $\hth-\theta_0$.
Lastly, we show that the variance estimator is consistent.

\subsubsection*{Step 1: orthogonality}

We use the orthogonality of the score  \eqref{score:inference}
to establish
  \begin{equation}\label{def:psin_taylor}
 \begin{aligned}
    &n^{-1/2}\sum_{i=1}^{n}\scorei_i\left(\theta;\hbeta,\hlam(\cdot,\theta),\hgr\right)  \\
   =& n^{-1/2}\sum_{i=1}^{n}\scorei_i(\theta_0;\bbeta_0,\Lambda_0,\bgr_0)  -\frac{1}{\sqrt{n}}(\theta-\theta_0) \sum_{i=1}^n D_i\{1-\expit(\bgr_0^\top\bZI[i])\}
  (e^{\theta_0X_i}-1)/\theta_0\\
  & + o_p(1+\sqrt{n}|\theta-\theta_0|) + O_p(\sqrt{n}|\theta-\theta_0|^2)
  + O_p(\sqrt{n}|\theta-\theta_0|^3).
  \end{aligned}
    \end{equation}
    under Assumption \ref{assume:inf}.
    The proof of \eqref{def:psin_taylor} involves tedious calculation,
    so we present the details separately in Lemma \ref{propA:thm-aalen-taylor}.

When the dimension of covariates $\bZ$ is fixed,
the representation \eqref{def:psin_taylor} immediately leads to
asymptotic normality through mere formality.
However, the growing dimension of covariates in our high-dimensional setting
may cause the violation of the
classical boundedness assumptions on the summands of $\scorei(\theta_0;\bbeta_0,\Lambda_0,\bgr_0)$.

\subsubsection*{Step 2: local consistency}

We first show that there exists one $\sqrt{n}$-consistent root of score $\scorei\left(\theta;\hbeta,\hlam(\cdot,\theta),\hgr\right)$
in the local neighborhood of $\theta_0$.
Under model \eqref{model:aalen},
\begin{equation*}
  n^{-1}\sum_{i=1}^{n}\scorei_i(\theta_0;\bbeta_0,\Lambda_0,\bgr_0) =\frac{1}{n}\sum_{i=1}^{n} \int_0^\tau \{D_i-\expit(\bgr_0^\top\bZI[i])\}e^{\theta_0D_i t}dM_i(t)
\end{equation*}
is a martingale with respect to filtration $\Ftn=\sigma\{N_i(u),Y_i(u),D_i,\bZ_i: u\le t, i=1,\dots,n\}$
$$
\frac{1}{n}\sum_{i=1}^{n} \int_0^t \{D_i-\expit(\bgr_0^\top\bZI[i])\}e^{\theta_0D_i t}dM_i(t)
$$
evaluated at $t=\tau$.
Its expectation is thus zero,
  \begin{equation}\label{eq:Ephi0}
    \E\{\scorei(\theta_0;\bbeta_0,\Lambda_0,\bgr_0)\} = 0.
  \end{equation}
  The true  $\theta_0$ is thus identified by the estimating equation $\scorei(\theta; \hbeta,\hlam(\theta),\hgr) = 0$.
  From \eqref{eq:Ephi0}, we may apply the concentration result of Lemma \ref{lemA:mart}, getting
  \begin{equation}\label{eq:aalen_consistency_num}
  n^{-1}\sum_{i=1}^{n}\scorei_i(\theta_0;\bbeta_0,\Lambda_0,\bgr_0,W_i) = o_p(1).
  \end{equation}
    Under the assumption $C \indep (T,D)\mid Z$ \eqref{eq:CindD}, we use the martingale property of $M(t)$, defined by \eqref{def:mart}, and Lemma \ref{lem:eY} to calculate the derivative with respect to $\theta$ at $\theta_0$
      \begin{eqnarray}
    && \frac{\partial}{\partial \theta}\E\left\{\scorei_i(\theta;\bbeta_0,\Lambda_0,\bgr_0)\right\} \bigg|_{\theta=\theta_0}\notag \\
    & =& \E\E\left(\{D-\expit(\bgr_0^\top\bZI))\}D\E\left[\int_0^\tau e^{D\theta_0 t} \{ t dM(t)-Y(t)dt\}\bigg|D,\bZ\right]\bigg|\bZ\right)\notag\\
    & =& -\E\left[D_i\{1-\expit(\bgr_0^\top\bZI[i])\}
  (e^{\theta_0X_i}-1)/\theta_0\right]. \label{eq:dEpsi0}
  \end{eqnarray}
  Notice the discontinuity of $(e^{\theta_0X_i}-1)/\theta_0$
  at $\theta_0=0$ can be removed as $\lim_{\theta_0 \to 0}(e^{\theta_0X_i}-1)/\theta_0 = X_i$.
 We have under the logistic regression model
  $$
  \E[D_i\{1-\expit(\bgr_0^\top\bZI[i])\}|\bZ_i] = \Var(D_i|\bZ_i),
  $$
  and under the additive hazards model
  $$
  \E\{ (e^{\theta_0X_i}-1)/\theta_0| \bZ_i\}
  = \E\left\{ \int_0^\tau e^{D_i\theta_0t} Y_i(t)dt| \bZ_i\right\}
  = \E\{Y(t)|\bZ;D=0\},
  $$
  so we have an alternative representation of \eqref{eq:dEpsi0},
  $$
  -\E\left[D_i\{1-\expit(\bgr_0^\top\bZI[i])\}
  (e^{\theta_0X_i}-1)/\theta_0\right]
  = -\int_0^\tau \E[\E\{Y(t)|\bZ;D=0\}\Var(D|\bZ)]dt.
  $$
  The at risk process $Y(t)$ is decreasing in time with minimum at  $Y(\tau)$.
  Combining such fact and Assumption \ref{assume:varD} (see also \eqref{eq:assume-varD}), we establish a lower bound for \eqref{eq:dEpsi0}
  \begin{equation*}
    \int_0^\tau \E[\E\{Y(t)|\bZ;D=0\}\Var(D|\bZ)]dt
    \ge  \int_0^\tau \E[\E\{Y(\tau)|\bZ;D=0\}\Var(D|\bZ)]dt
    \ge  \tau \varepsilon_Y,
  \end{equation*}
  with $\varepsilon_Y$ defined in Assumption \ref{assume:varD} (see also \eqref{eq:assume-varD}).   Based on the bound
  \begin{equation}\label{eq:aalen_dpsi_bound}
    |D_i\{1-\expit(\bgr_0^\top\bZI[i])\}
  (e^{\theta_0X_i}-1)/\theta_0| \le e^{\tau\theta_0}\tau,
  \end{equation}
  we can use the Hoeffding's inequality (as in Lemma \ref{lemma:Hoeffding}) to establish a lower bound
  \begin{equation}\label{eq:aalen_consistency_denom}
    \P\left(n^{-1}\sum_{i=1}^n D_i\{1-\expit(\bgr_0^\top\bZI[i])\}
  (e^{\theta_0X_i}-1)/\theta_0 > \varepsilon_Y/2\right) > 1- e^{-\frac{n\varepsilon_Y^2}{8e^{2\tau\theta_0}\tau^2}}.
  \end{equation}

Now, we locate the root $\hth$ in a $\delta_n$-neighborhood of $\theta_0$ with large probability.
The radius $\delta_n$ determines the rate of convergence for $\hth$.
We choose $\delta_n= n^{-1/4}$.
Denote the boundaries of the $\delta_n$-neighborhood of $\theta_0$ as $\theta_+ = \theta_0 + \delta_n$ and $\theta_- = \theta_0 - \delta_n$.
Using \eqref{def:psin_taylor}, 
we have
\begin{gather*}
  n^{-1}\sum_{i=1}^{n}\scorei_i\left(\theta_+;\hbeta,\hlam(\cdot,\theta_+),\hgr\right)
  = n^{-1}\sum_{i=1}^{n}\scorei_i(\theta_0;\bbeta_0,\Lambda_0,\bgr_0)
  - \frac{\delta_n}{n} \sum_{i=1}^n D_i\{1-\expit(\bgr_0^\top\bZI[i])\}+o_p(1), \\
  n^{-1}\sum_{i=1}^{n}\scorei_i\left(\theta_-;\hbeta,\hlam(\cdot,\theta_-),\hgr\right)
  = n^{-1}\sum_{i=1}^{n}\scorei_i(\theta_0;\bbeta_0,\Lambda_0,\bgr_0)
  + \frac{\delta_n}{n} \sum_{i=1}^n D_i\{1-\expit(\bgr_0^\top\bZI[i])\}+o_p(1).
\end{gather*}
Applying   \eqref{eq:aalen_consistency_num} and  \eqref{eq:aalen_consistency_denom},
we have $n^{-1}\sum_{i=1}^{n}\scorei_i(\theta_0;\bbeta_0,\Lambda_0,\bgr_0) = o_p(1)$
and $n^{-1} \sum_{i=1}^n D_i\{1-\expit(\bgr_0^\top\bZI[i])\}$ bounded away from zero.
Thus,  we have
$$n^{-1}\sum_{i=1}^{n}\scorei_i\left(\theta_+;\hbeta,\hlam(\cdot,\theta_+),\hgr\right) < 0<
n^{-1}\sum_{i=1}^{n}\scorei_i\left(\theta_-;\hbeta,\hlam(\cdot,\theta_+),\hgr\right) $$
with probability tending to one.
By  continuity of $\scorei$ with respect to $\theta$, the root $\hth$ must be in
$(\theta_-, \theta_+)$,
implying $\hth-\theta_0 = O_p(\delta_n)=O_p(n^{-1/4})$.
 Plugging this
 to \eqref{def:psin_taylor},
we obtain
  \begin{equation*}
    \sqrt{n}  (\hth-\theta_0) = \frac{n^{-1/2}\sum_{i=1}^{n}\scorei_i(\theta_0;\bbeta_0,\Lambda_0,\bgr_0)}
      {n^{-1}\sum_{i=1}^n D_i\{1-\expit(\bgr_0^\top\bZI[i])\}
  (e^{\theta_0X_i}-1)/\theta_0} + o_p(1).
    \end{equation*}

\subsubsection*{Step 3: asymptotic normality}
 Let $X_{(1)}, \dots, X_{(n)}$ be the order statistics of the observed times
 and
 \begin{eqnarray}
   M^1_k &=& \frac{1}{n} \sum_{i=1}^{n} \int_0^{X_{(k)}}
   D_i\{1-\expit(\bgr_0^\top\bZI[i])\} e^{\theta_0t}dM_i(t), \notag \\
    M^0_k &=& \frac{1}{n} \sum_{i=1}^{n} \int_0^{X_{(k)}}
   (1-D_i)\expit(\bgr_0^\top\bZI[i])dM_i(t), \label{eq:psi-mart}
 \end{eqnarray}
 for $k = 0, \dots, n$.
 We note that the score $n^{-1}\sum_{i=1}^{n}\scorei_i$ with true parameters can be alternatively
 expressed as
 \begin{equation*}
     n^{-1}\sum_{i=1}^{n}\scorei_i(\theta_0;\bbeta_0,\Lambda_0,\bgr_0)
   = M^1_n- M^0_n.
 \end{equation*}
 Since the integrands of both $M^1_k$ and $M^0_k$, $D_i\{1-\expit(\bgr_0^\top\bZI[i])\} e^{\theta_0t}$
 and $(1-D_i)\expit(\bgr_0^\top\bZI[i])$, in \eqref{eq:psi-mart} are nonnegative and bounded by
 $\tau (1\vee e^{\theta_0 \tau})$, we
 can apply Lemma \ref{lemA:martdiff} to get that both $M^1_k$ and $M^0_k$, hence $M^1_k-M^0_k$, are martingales under filtration
 $\FMk = \sigma\{N_i(u),Y_i(u+),D_i,\bZ_i: u \in [0,t_k], i=1,\dots,n\}$
 satisfying, most importantly,
 \begin{equation*}
    \max\left\{\E\left\{(M^1_k - M^1_{k-1})^2|\FMk\right\},
    \E\left\{(M^0_k - M^0_{k-1})^2|\FMk\right\}\right\}
  \le 8\tau^2 (1\vee e^{\theta_0 \tau})^2/n^2.
 \end{equation*}
  By the Cauchy-Schwartz inequality, we have
 \begin{equation}\label{eq:psi-mart-CS}
   (M^1_k -M^0_k - M^1_{k-1}+ M^0_{k-1})^2 \le 2(M^1_k - M^1_{k-1})^2+2(M^0_k - M^0_{k-1})^2.
 \end{equation}
 Hence for the quadratic variation of $n^{-1}\sum_{i=1}^{n}\scorei_i(\theta_0;\bbeta_0,\Lambda_0,\bgr_0)$,
 we have
 \begin{equation*}
   \E\left\{(M^1_k - M^1_{k-1}-M^0_k + M^0_{k-1})^2|\FMk\right\} \le
   32\tau^2 (1\vee e^{\theta_0 \tau})^2/n^2.
 \end{equation*}
 As a result, the variance
  \begin{equation}\label{def:sig_psi}
   \sigma_\scorei^2 = \Var\{n^{-1/2}\sum_{i=1}^{n}\scorei_i(\theta_0;\bbeta_0,\Lambda_0,\bgr_0)\}
   = n \E\left[\sum_{i=1}^n \E\left\{(M^1_k -M^0_k - M^1_{k-1}+ M^0_{k-1})^2|\FMk\right\}\right]
 \end{equation}
 is finite, bounded by $32\tau^2 (1\vee e^{\theta_0 \tau})^2$.

 Now, we verify the Lindeberg condition for the martingale central limit theorem \citep{Brown71}.
The event
\begin{equation*}
  \sqrt{n}|M^1_k - M^1_{k-1} -M^0_k+ M^0_{k-1}|>\varepsilon
\end{equation*}
occurs only if
\begin{equation*}
  \sqrt{n}|M^1_k - M^1_{k-1}|>\varepsilon/2 \text{ or }
  \sqrt{n}|M^0_k - M^0_{k-1}|>\varepsilon/2
\end{equation*}
occurs.
Let $I(\cdot)$ be the binary event indicator.
Thus, we must have the following inequality
\begin{eqnarray*}
&&I(\sqrt{n}|M^1_k - M^1_{k-1} -M^0_k+ M^0_{k-1}|>\varepsilon) \notag \\
&\le& I(\sqrt{n}|M^1_k - M^1_{k-1}|>\varepsilon/2)+I(\sqrt{n}|M^0_k - M^0_{k-1}|>\varepsilon/2).
\end{eqnarray*}
 Along with \eqref{eq:psi-mart-CS}, we have
 \begin{eqnarray}
  && n\sum_{i=1}^n \E\left\{(M^1_k - M^1_{k-1} -M^0_k+ M^0_{k-1})^2 I\left(\sqrt{n}|M^1_k -M^0_k - M^1_{k-1}+ M^0_{k-1}|>\varepsilon\right)\right\} \notag \\
  &\le& 2n\sum_{i=1}^n \E\left\{(M^1_k- M^1_{k-1})^2 I\left(\sqrt{n}|M^1_k - M^1_{k-1}|>\varepsilon/2\right)\right\}\notag \\
  &&+2n\sum_{i=1}^n \E\left\{(M^0_k- M^0_{k-1})^2 I\left(\sqrt{n}|M^0_k - M^0_{k-1}|>\varepsilon/2\right)\right\}.
  \label{eq:psi-mart-clt}
 \end{eqnarray}
 By Lemma \ref{lemA:martdiff},
 the right hand side in \eqref{eq:psi-mart-clt} decays to zero
 when $n$ approaches $\infty$.
 Hence, we can apply the martingale central limit theorem to
 \begin{equation}\label{norm:aalenscore}
n^{-1/2}\sigma_\scorei^{-1}\sum_{i=1}^{n}\scorei_i(\theta_0;\bbeta_0,\Lambda_0,\bgr_0,W_i)
   \leadsto N(0,1).
 \end{equation}

\subsubsection*{Step 4: tightness}
 We define the asymptotic variance of $\sqrt{n}(\hth-\theta_0)$ as
 \begin{align}
   \sigma^2 = & (\sigma_\scorei/\E[D\{1-\expit(\bgr_0^\top\bZ_{1})\}
  (e^{\theta_0X}-1)/\theta_0])^2 \notag \\
  = & \frac{\Var\{n^{-1/2}\sum_{i=1}^{n}\scorei_i(\theta_0;\bbeta_0,\Lambda_0,\bgr_0)\}}
  {\E[D\{1-\expit(\bgr_0^\top\bZ_{1})\}
  (e^{\theta_0X}-1)/\theta_0])^2}\label{def:sig}
 \end{align}
 where $\sigma_\scorei$ is the square root of \eqref{def:sig_psi}.
 Since $\hth$ solves $n^{-1}\sum_{i=1}^{n}\scorei_i(\theta; \hbeta,\hlam(\theta),\hgr) = 0$,
 we have along with Lemma \ref{propA:thm-aalen-taylor}
 \begin{eqnarray}
    && \sqrt{n}\sigma^{-1}(\hth-\theta_0)
    -\frac{1}{\sqrt{n}\sigma_\scorei}\sum_{i=1}^{n}\scorei_i(\theta_0;\bbeta_0,\bgr_0,W_i)\notag \\
    &=&  \frac{\sqrt{n}(\hth-\theta_0)}{\sigma_\scorei}
    \bigg(\E[D\{1-\expit(\bgr_0^\top\bZ_{1})\}
  (e^{\theta_0X}-1)/\theta_0] \notag \\
  &&\quad -\frac{1}{n}\sum_{i=1}^{n}D_i\{1-\expit(\bgr_0^\top\bZI[i])\}
  (e^{\theta_0X_i}-1)/\theta_0\Bigg) \notag \\
    &&+ o_p(1+\sqrt{n}|\hth-\theta_0|)+
    O_p(\sqrt{n}|\hth-\theta_0|^2)
  + O_p(\sqrt{n}|\hth-\theta_0|^3). \label{eq:aalenapprox}
 \end{eqnarray}
 Applying the $\sqrt{n}$-consistency of $\hth$ from step 3,
 we obtain that the $o_p$ and $O_p$ terms in \eqref{eq:aalenapprox} are asymptotically negligible.
 Again using the bound \eqref{eq:aalen_dpsi_bound},
 we apply the Hoeffding's inequality (as in Lemma \ref{lemma:Hoeffding}) to establish that
 \begin{equation}\label{eq:aalen_Edpsi}
   \E[D\{1-\expit(\bgr_0^\top\bZ_{1})\}
  (e^{\theta_0X}-1)/\theta_0]-\frac{1}{n}\sum_{i=1}^{n}D_i\{1-\expit(\bgr_0^\top\bZI[i])\}
  (e^{\theta_0X_i}-1)/\theta_0
 \end{equation}
 is of order $O_p(n^{-1/2})$.
 Hence, the right hand side of \eqref{eq:aalenapprox} is of order $o_p(1+\sqrt{n}|\hth-\theta_0|)$.
 Along with the normality \eqref{norm:aalenscore}, we establish the $\sqrt{n}$-tightness of
 the estimation error
 \begin{equation*}
   |\hth-\theta_0| = O_p\big(n^{-1/2}\big).
 \end{equation*}
 Plugging in the rate of estimation error into the righthand side of \eqref{eq:aalenapprox}, we
 obtain the asymptotic equivalence
 \begin{equation}\label{eq:aalenerror}
   \sqrt{n}\sigma^{-1}(\hth-\theta_0)
    -\sigma_\scorei^{-1}n^{-1/2}\sum_{i=1}^{n}\scorei_i(\theta_0;\hbeta,\hgr,\bW) = o_p(1).
 \end{equation}

\subsubsection*{Step 5: variance estimation}
To show that $\hsig^{-1}$ defined by
\[
 \hsig^2 = \frac{n^{-1}\sum_{i=1}^{n}\delta_i \{D_i-\expit(\hgr^\top\bZI[i])\}^2e^{2\hth D_i X_i}}
  {\left\{n^{-1}\sum_{i=1}^n(1-D_i)\expit(\hgr^\top\bZI[i])X_i\right\}^2}.
\]
 is a consistent estimator
for $\sigma^{-1}$,
we decompose the numerator of $\hsig^2$ into
\begin{eqnarray}
  && \frac{1}{n}\sum_{i=1}^{n}\delta_i \{D_i-\expit(\hgr^\top\bZI[i])\}^2e^{2\hth D_i X_i} \notag \\
  &=& \frac{1}{n}\sum_{i=1}^{n}\left[\delta_i \{D_i-\expit(\hgr^\top\bZI[i])\}^2e^{2\hth D_i X_i} - \delta_i\{D_i-\expit(\bgr_0^\top\bZI[i])\}^2e^{2\theta_0 D_i X_i} \right]\notag \\
  && + \frac{1}{n}\sum_{i=1}^{n} \int_0^\tau
  [\{D_i-\expit(\bgr_0^\top\bZI[i])\}e^{\theta_0 D_i t}]^2dN_i(t). \label{eq:hsig-num}
\end{eqnarray}
By the mean value theorem, the first term in the righthand side of \eqref{eq:hsig-num}
can be written in terms of $\theta_\xi = (1-\xi)\theta_0+\xi\hth$
and $\bgr_\xi = (1-\xi)\bgr_0 +\xi\hgr$ with some $\xi \in (0,1)$,
\begin{eqnarray}
  && \frac{(\bgr_0-\hgr)^\top}{n}\sum_{i=1}^{n}\frac{\delta_i \{D_i-\expit(\bgr_\xi^\top\bZI[i])\}e^{\bgr_\xi^\top\bZI[i]}e^{2\theta_\xi D_i X_i}}{\left(1+e^{\bgr_\xi^\top\bZI[i]}\right)^2} \notag \\
 && + \frac{(\hth-\theta_0)}{n}\sum_{i=1}^{n} 2\delta_i D_iX_i \{1-\expit(\bgr_\xi^\top\bZI[i])\}^2e^{2\theta_\xi D_i X_i} \notag \\
 &=& O_p(\|\hgr-\bgr_0\|_1+|\hth-\theta_0|). \label{eq:hsig-dnum}
\end{eqnarray}
We repeatedly use $\xi$ in all mean value theorem expansions
for convenience of notation.
The second term on the righthand side of \eqref{eq:hsig-num} is
the optional quadratic variation of $n^{-1}\sum_{i=1}^{n}\scorei_i(\theta_0;\bbeta_0,\bgr_0,\Lambda_0)$  bounded by $e^{2\theta_0\tau}$ \cite[(5.17) p. 159 and (5.26) p. 162]{KalbfleischPrentice02}. By the Hoeffding's inequality (as in Lemma \ref{lemma:Hoeffding}), we have the concentration of the second term around the variance of $n^{-1}\sum_{i=1}^{n}\scorei_i(\theta_0;\bbeta_0,\bgr_0,\Lambda_0)$,
\begin{eqnarray}
&&\frac{1}{n}\sum_{i=1}^{n} \int_0^\tau
  [\{D_i-\expit(\bgr_0^\top\bZI[i])\}e^{\theta_0 D_i t}]^2dN_i(t)
  \notag \\
  &=& \E\left(\int_0^\tau
  [\{D_i-\expit(\bgr_0^\top\bZI[i])\}e^{\theta_0 D_i t}]^2dN_i(t)\right) +O_p(n^{-1/2}) \notag \\
  &=& \sigma_\scorei^2 + o_p(1). \label{eq:hsig-Enum}
\end{eqnarray}
Putting \eqref{eq:hsig-dnum} and \eqref{eq:hsig-Enum} together,
we have the numerator of $\hsig^2$  equals $\sigma_\scorei^2 + o_p(1)$.
Similarly, we decompose the denominator of $\hsig$ into
\begin{align}
& n^{-1}\sum_{i=1}^n(1-D_i)\expit(\hgr^\top\bZI[i])X_i \notag \\
= & \E \left[D\{1-\expit(\bgr_0^\top\bZ_{1})\}
  (e^{\theta_0X}-1)/\theta_0 \right] \notag \\
& -\left(\E \left [D\{1-\expit(\bgr_0^\top\bZ_{1})\}
  \frac{e^{\theta_0X_i}-1}{\theta_0}\right] -\frac{1}{n}\sum_{i=1}^{n}D_i\{1-\expit(\bgr_0^\top\bZI[i])\}
  \frac{e^{\theta_0X_i}-1}{\theta_0}\right) \notag\\
& +
  \frac{1}{n}\sum_{i=1}^{n}\left[D_i\{1-\expit(\hgr^\top\bZI[i])\}
  \frac{e^{\hth X_i}-1}{\hth}-D_i\{1-\expit(\hgr^\top\bZI[i])\}
  \frac{e^{\theta_0X_i}-1}{\theta_0}\right] \notag \\
  & +
  \frac{1}{n}\sum_{i=1}^{n}\left[D_i\{\expit(\bgr_0^\top\bZI[i])-\expit(\hgr^\top\bZI[i])\}
  \frac{e^{\theta_0 X_i}-1}{\theta_0}\right] \notag \\
= & Q_1 + Q_2 + Q_3 + Q_4\label{eq:hsig-ddenom}.
\end{align}
$Q_1$ is the leading term.
$Q_2$ has been studied in Step 3 as \eqref{eq:aalen_Edpsi}.
We define a compact neighborhood of $\theta_0$ $O(\theta_0,1) \subset \R$.
By the consistency of $\hth$,
we have
$\lim_{n \to 1}\P(\hth \in O(\theta_0,1)) = 1$.
The expression $$\frac{\theta e^{\theta X_i}X_i - e^{\theta X_i}+1}{\theta^2}$$
requires some additional attention due to the potential singularity at $\theta = 0$.
As a function of $\theta$, we show the singularity at $\theta = 0$
is removable by the L'Hospital's rule,
$$
\lim_{\theta \to 0} \frac{\theta e^{\theta X_i}X_i - e^{\theta X_i}+1}{\theta^2}
= \lim_{\theta \to 0} \frac{
\theta  e^{\hth_\xi X_i}X_i^2}{\theta}  = X_i^2 \le \tau^2
$$
while being continuous elsewhere.
As a function of $X_i$, we utilize the monotonicity of $ x e^x - e^x$
$$
0 \le \frac{\theta e^{\theta X_i}X_i - e^{\theta X_i}+1}{\theta^2}
\le \frac{\theta e^{\theta \tau}\tau - e^{\theta \tau}+1}{\theta^2}.
$$
Thus, we have a uniform bound over the compact set
$$
\sup_{\theta \in O(\theta_0,1)} \left|\frac{\theta e^{\theta X_i}X_i - e^{\theta X_i}+1}{\theta^2}\right|
\le \sup_{\theta \in O(\theta_0,1)} \frac{\theta e^{\theta \tau}\tau - e^{\theta \tau}+1}{\theta^2} < \infty.
$$
We note that the constant bound above does not change with dimension.
Then for sufficiently large $n$, we may apply the mean value theorem bound for $Q_3$
\begin{align*}
  \left|Q_3\right|
 \le & \sup_{\theta_\xi in O(\theta_0,1)}\left|\frac{1}{n}\sum_{i=1}^{n}D_i\{1-\expit(\hgr^\top\bZI[i])\} \frac{\hth_{\xi}e^{\hth_\xi X_i}X_i - e^{\hth_\xi X_i}+1}{\hth^2_{\xi}}\right| |\hth-\theta_0| \\
 \le &|\hth-\theta_0| \sup_{\theta \in O(\theta_0,1)} \frac{\theta e^{\theta \tau}\tau - e^{\theta \tau}+1}{\theta^2}.
\end{align*}
This shows that $Q_3$ is of order $O_p(|\hth-\theta_0|)$.
By another mean-value theorem argument, we establish a bound for $Q_4$
\begin{align*}
|Q_4| =  & \left|\frac{1}{n}\sum_{i=1}^{n}\left[D_i\expit(\bgr_\xi^\top\bZI[i])\{1-\expit(\bgr_\xi^\top\bZI[i])\}
  \frac{e^{\theta_0 X_i}-1}{\theta_0}\right]\bZI[i]^\top (\hgr-\bgr_0)\right|  \\
  \le  & \frac{e^{|\theta_0|\tau}-1}{|\theta_0|} K_Z \|\hgr-\bgr_0\|_1.
\end{align*}
This shows $Q_4$ is of order $O_p(\|\hgr-\bgr_0\|_1)$.
From our analysis of \eqref{eq:hsig-ddenom}, we have established
$$
n^{-1}\sum_{i=1}^n(1-D_i)\expit(\hgr^\top\bZI[i])X_i
= \E[D\{1-\expit(\bgr_0^\top\bZ_{1})\}
  (e^{\theta_0X}-1)/\theta_0] + o_p(1),
$$
which implies the denominator of the variance estimator $\hsig^2$
converges to  the asymptotic variance $\sigma^2$.

Under the additive hazards model \eqref{model:aalen},
we must have a nonnegative hazard among the control subjects
\begin{equation*}
  \bbeta_0^\top \bZ + d\Lambda_0(t) \ge 0
\end{equation*}
for all $\bZ$ such that $\Pr(D = 0|\bZ)>0$.
Under the assumption $C \indep (T,D) | Z$ \eqref{eq:CindD} and Assumption \ref{assume:varD} (see also \eqref{eq:assume-varD}), we can establish a lower bound for $\sigma_\scorei$
\begin{eqnarray*}
  \sigma_\scorei &=& \E\left[\int_0^\tau \{D-\expit(\bgr_0^\top\bZ)\}^2e^{2D\theta_0t} Y(t)
  \{(D\theta_0+\bbeta_0^\top\bZ)dt + d\Lambda_0(t)\}\right] \notag \\
  &=& \E\left[\int_0^\tau \{D-\expit(\bgr_0^\top\bZ)\}^2e^{D\theta_0t} \E\{Y(t)|\bZ;D=0\}
  \{(D\theta_0+\bbeta_0^\top\bZ)dt + d\Lambda_0(t)\}\right] \notag \\
  &=& \E\left[\int_0^\tau D\{1-\expit(\bgr_0^\top\bZ)\}^2e^{\theta_0t} \E\{Y(\tau)|\bZ;D=0\}
  \theta_0dt\right] \notag \\
  && + \E\left[\int_0^\tau \{D-\expit(\bgr_0^\top\bZ)\}^2e^{D\theta_0t}
  d\E\{N(t)|\bZ;D=0\}\right] \notag \\
  &\ge& 0 + e^{1\wedge\theta_0\tau}\E[\Var(D|\bZ)\E\{N(\tau)|\bZ;D=0\}] \notag \\
  &\ge& e^{1\wedge\theta_0\tau} \varepsilon_N.
\end{eqnarray*}
Hence,
the asymptotic variance $\sigma^2$  \eqref{def:sig} is bounded
away from zero
\begin{equation*}
  \sigma^{2} =
  \left(\frac{\sqrt{\E[\delta\{D-\expit(\bgr_0^\top\bZ_1)\}^2e^{2D\theta_0X}]} }{\E\left[D\{1-\expit(\bgr_0^\top\bZ_1)\}\frac{e^{\theta_0X}-1}
  {\theta_0}\right]}\right)^2
  \ge \left(\frac{\sqrt{e^{1\wedge\theta_0\tau} \varepsilon_N}}{\tau e^{\theta_0 \tau}}\right)^2.
\end{equation*}
Therefore, we have
\begin{equation}\label{eq:hsig-sig-op1}
  \hsig^{-1} = \sigma^{-1} + o_p(1)
\end{equation}
by continuous mapping theorem.

Combining the results \eqref{norm:aalenscore}, \eqref{eq:aalenerror} and \eqref{eq:hsig-sig-op1},
we obtain
\begin{equation*}
 \sqrt{n}\hsig^{-1}(\hth-\theta_0) \leadsto N(0,1).
 \end{equation*}

\subsection{Proof of Lemma \ref{lem:unique}} 
To prove the uniqueness of the root for $\scorei(\theta; \hbeta,\hlam,\hgr)$,
we first establish the uniform concentration of score
around its limit
\begin{equation}\label{eq:unique-decomp}
 \sup_{\theta \in [-K_\theta,K_\theta]} \left|n^{-1}\sum_{i=1}^{n}\scorei_i(\theta; \hbeta,\hHaz,\hgr)-\E\{\scorei_i(\theta; \bbeta_0,\overline{\Haz},\bgr_0)\}\right|= o_p(1), \;
 \overline{\Haz}(t;\theta) = \Haz_0(t) + \int_0^t d_0(u) (\theta_0-\theta) du.
\end{equation}
The details for \eqref{eq:unique-decomp} are provided in Lemma \ref{propA:unique-decomp}.
Then, we show that the population score
\begin{equation}\label{def:Escore}
\E\{\scorei(\theta; \bbeta_0,\overline{\Haz},\bgr_0)\}
= \E\left[\{D - \expit(\bgr_0^\top\bZI[])\}\int_{0}^{\tau}e^{D\theta t}Y(t)\{D-d_0(t)\}(\theta-\theta_0)  dt\right]
\end{equation}
has a unique root.
Using the identity
$$
D - \expit(\bgr_0^\top\bZI[]) = D\{1- \expit(\bgr_0^\top\bZI[])\} - (1-D)\expit(\bgr_0^\top\bZI[]),
$$
we may characterize any root of the population  score as
\begin{align}
 0 =  & \E\left[\{D - \expit(\bgr_0^\top\bZI[])\}\int_{0}^{\tau}e^{D\theta t}Y(t)\{D-d_0(t)\}(\theta-\theta_0)  dt\right] \notag \\
 = & (\theta - \theta_0)\left(\E\left[D\{1 - \expit(\bgr_0^\top\bZI[])\}\int_{0}^{\tau}e^{\theta t}Y(t)\{1-d_0(t)\} dt\right]\right. \notag \\
& \qquad \qquad+ \left. \E\left[(1-D)\{\expit(\bgr_0^\top\bZI[])\}\int_{0}^{\tau}Y(t)d_0(t) dt\right]
 \right). \label{def:root-Escore}
\end{align}
Under Assumption \ref{assume:varD} along with the bounds $|\theta|<K_\theta$ and $0 \le d_0(t) \le 1$,
we have
\begin{gather*}
 \E\left[D\{1 - \expit(\bgr_0^\top\bZI[])\}\int_{0}^{\tau}e^{\theta t}Y(t)\{1-d_0(t)\} dt\right]
 \ge \varepsilon^2  e^{-K_\theta \tau} \mu(d_0 < 1-\varepsilon), \\
 \E\left[(1-D)\{\expit(\bgr_0^\top\bZI[])\}\int_{0}^{\tau}Y(t)d_0(t) dt\right]
 \ge \varepsilon^2  \mu(d_0 >\varepsilon),
\end{gather*}
where $\mu$ is the Lebesgue measure over $[0,\tau]$.
Without loss of generality, we may set $\varepsilon < 1/2$,
in which case we have
$$
\mu(d_0 >\varepsilon) +  \mu(d_0 < 1-\varepsilon) \ge \tau.
$$
Thus, we have shown that
$$
\E\left[D\{1 - \expit(\bgr_0^\top\bZI[])\}\int_{0}^{\tau}e^{\theta t}Y(t)\{1-d_0(t)\} dt+(1-D)\{\expit(\bgr_0^\top\bZI[])\}\int_{0}^{\tau}Y(t)d_0(t) dt\right]
 \ge \varepsilon^2  e^{-K_\theta \tau} \tau
$$
for any $\theta \in [-K_\theta, K_\theta]$.
Therefore, we may conclude from \eqref{def:root-Escore}
that $\theta = \theta_0$ is the unique root of the population  score
in $[-K, K]$.

\subsection{Proof of Theorem \ref{thm:2}}
To apply  Theorem \ref{thm:aalen},
we prove the claim  by establishing its conditions.
 Notice that the  weighted Breslow estimator \eqref{def:clam} satisfies
conditions in regards to how well it estimates $\Lambda_0$ in \ref{assume:hlamTV}, \ref{assume:hlamlim} and \ref{assume:rate-inf}.

Throughout the proof, we focus on the event
\begin{equation}\label{event:wY}
 \inf_{t\in[0,\tau]}\frac{1}{n}\sum_{i=1}^n w^1_i(\cgr)Y_i(t)> e^{-\Kth\tau}\varepsilon_Y/2
\end{equation}
with $\varepsilon_Y$ defined in Assumption \ref{assume:varD} (see also \eqref{eq:assume-varD}).
Such event occurs with probability tending to one by Lemma \ref{lemA:wY}.

To analyze $\clam$, we consider the following decomposition
\begin{align}\label{eq:decomp-clam}
  \clam(t,\theta;\hbeta,\hgr)
  = &  \int_0^t \frac{\sum_{i=1}^n w^1_i(\hgr) dM_i(u)}{
    \sum_{i=1}^n w^1_i(\hgr) Y_i(u)}
    - (\theta-\theta_0) t
     \nonumber\\ \nonumber
    & \qquad -
    (\hbeta-\beta_0)^\top \int_0^t \frac{\sum_{i=1}^n w^1_i(\hgr) Y_i(u)\bZ_i du}{
    \sum_{i=1}^n w^1_i(\hgr) Y_i(u)}   + \Lambda_0(t)\\
  = & Q_1(t) - Q_2(t)
   \nonumber\\
  &\qquad - Q_3(t) + \Lambda_0(t).
\end{align}

We first study the total variation of $\clam(t,\theta;\hbeta,\hgr)$.
On the event \eqref{event:wY},
$$
\frac{w^1_i(\hgr)}{n^{-1}\sum_{i=1}^n w^1_i(\hgr) Y_i(u)}
\le  2e^{\Kth\tau}/\varepsilon_Y.
$$
By Lemma \ref{lemA:mart} with $K_H = 2e^{\Kth\tau}/\varepsilon_Y$, the total variation of $Q_1(t)$ is bounded with probability
tending to one.

 Inside the compact neighborhood $|\theta - \theta_0|\le K_\theta$,
the total variation of $Q_2(t)$ is at most $K_\theta \tau$.
Applying the H\"{o}lder's inequality,
we have
$$
\bigvee_0^\tau Q_3(t) \le \sup_{i=1,\dots,n}|(\hbeta-\beta_0)^\top\bZ_i| \int_0^\tau du
\le
\tau \|\hbeta-\beta_0\|_1 \|Z_i\|_\infty \tau.
$$
Under Conditions \ref{assume:Z} and the rate condition $\sqrt{\log(p)}\|\hbeta-\beta_0\|_1=o_p(1)$    in the statement of the theorem,
the total variation of $Q_3(t)$ is also bounded with probability tending to one.
The cumulative baseline hazard $\Lambda_0(t)$ is monotone increasing,
so its total variation is bounded by $\Lambda_0(\tau)$.
Therefore, Assumption \ref{assume:hlamTV} is met.

We then study the approximate linearity in $\theta$.
This is obviously satisfied according to the decomposition \eqref{eq:decomp-clam}.

We next study the rate of convergence of $\clam$.
By Lemma \ref{lemA:mart}, $\sup_{t\in[0,\tau]}|Q_1(t)| = O_p(n^{-1/2})$.
The term $Q_2(t) = 0$ when $\theta = \theta_0$.
Note that $w^1_i(\hgr) \in [0,1]$.

Applying the Cauchy-Schwartz inequality twice,
we obtain a bound for $Q_3(t)$,
\begin{align*}
\sup_{t\in[0,\tau]} |Q_3(t)| \le &  \int_0^\tau \sqrt{\sum_{i=1}^{n} \{(\hbeta-\beta_0)^\top Z_i\}^2 Y_i(t)}\frac{\sqrt{\sum_{i=1}^{n} \{w^1_i(\cgr)\}^2Y_i(t)}}{\sum_{i=1}^{n} w^1_i(\cgr)Y_i(t)}dt \\
\le& \sqrt{\int_0^\tau \frac{1}{n}\sum_{i=1}^{n} \{(\hbeta-\beta_0)^\top Z_i\}^2 Y_i(t)dt }
\sqrt{\int_0^\tau \frac{n^{-1}\sum_{i=1}^{n} \{w^1_i(\cgr)\}^2Y_i(t)}{\left\{n^{-1}\sum_{i=1}^{n} w^1_i(\cgr)Y_i(t)\right\}^2}dt}.
\end{align*}
Notice the first factor on the right hand side is the definition of average training deviance
$\msenb(\hbeta,\bbeta_0)$
in Section \ref{section:inf-phi}.
On the event \eqref{event:wY}, we then have
$$
\sup_{t\in[0,\tau]} |Q_3(t)| \le \msenb(\hbeta,\bbeta_0)2\tau e^{K_\theta \tau}/\varepsilon_Y.
$$
Thus, $\sup_{t\in[0,\tau]} |Q_3(t)| = O_p(\msenb(\hbeta,\bbeta_0))$.

By a similar expansion,
$$
 \clam(t,\theta_0;\hbeta,\hgr) = Q_1(t)- Q_3(t) + \Lambda_0(t),
$$
with the same $Q_1$ and $Q_3$ defined in \eqref{eq:decomp-clam}.
Through the analyses of $Q_1$ and $Q_3$ above, we have that the estimation error of $\clam$  is dominated by the estimation error in $\hbeta$,
$$
\sup_{t\in[0,\tau]} |\clam(t,\theta_0;\hbeta,\hgr) - \Lambda_0(t)|
= O_p(\msenb(\hbeta,\bbeta_0)+ n^{-1/2}) = O_p(\msenb(\hbeta,\bbeta_0)).
$$
To verify \eqref{eq:rate-hlam},
we consider
\begin{align}\label{eq:decomp-clam-0}
& \int_0^\tau H(t)  d\{\clam(t,\theta_0;\hbeta,\hgr)-\Haz_0(t)\} \nonumber\\
  = &  \int_0^\tau H(t)\frac{\sum_{i=1}^n w^1_i(\hgr) dM_i(t)}{
    \sum_{i=1}^n w^1_i(\hgr) Y_i(t)}
   -
    (\hbeta-\beta_0)^\top \int_0^\tau H(t)\frac{\sum_{i=1}^n w^1_i(\hgr) Y_i(t)\bZ_i dt}{
    \sum_{i=1}^n w^1_i(\hgr) Y_i(t)}  \nonumber\\
  = & Q_1^H + Q_3^H.
\end{align}
By Lemma \ref{lemA:mart}, $Q_1^H = O_p(n^{-1/2})$.
Following the bound for $Q_3(t)$,
we have the bound
$$
Q_3^H \le  \sup_{t\in[0,\tau]} |H(t)|Q_3(\tau)  \le \sup_{t\in[0,\tau]} |H(t)|\msenb(\hbeta,\bbeta_0)2\tau e^{K_\theta \tau}/\varepsilon_Y
$$
on event  \eqref{event:wY}. This verifies \eqref{eq:rate-hlam}.

\subsection{Proof of Theorem \ref{thm:aalen-cf}}
The flow of the proof is the same as that of Theorem \ref{thm:aalen}.
We proceed in the same Steps 1-5 and only highlight parts that are substantially different.

\subsubsection*{Step 1: orthogonality}

The weaker Assumption \ref{assume:hat-cf} for Theorem \ref{thm:aalen-cf}
is ascribed to the improvement in orthogonality from cross-fitting, as advocated
by \cite{ChernozhukovEtal17}.
  Utilizing the independence between the in-fold data and out-of-fold estimators,
  we may obtain sharper bounds of various terms
  by first establish concentration of the terms be analyzed
  around their expectation over in-fold data.
  We provide the detail in the proof of Lemma \ref{propA:thm-aalen-taylor-cf}
that we may obtain the local approximation
 \begin{align}
    &\frac{\sqrt{n}}{|\fold[j]|}\sum_{i\in\fold[j]}\scorei_i\left(\theta;\hbj,\hlamj(\cdot,\theta),\hgrj\right)  \notag\\
   =& \frac{\sqrt{n}}{|\fold[j]|}\sum_{i\in\fold[j]}\scorei_i(\theta_0;\bbeta_0,\Lambda_0,\bgr_0)  -\frac{\sqrt{n}}{|\fold[j]|}(\theta-\theta_0) \sum_{i \in \fold[j]} D_i\{1-\expit(\bgr_0^\top\bZI[i])\}
  (e^{\theta_0X_i}-1)/\theta_0 \notag\\
  & + o_p(1+\sqrt{n}|\theta-\theta_0|) + O_p(\sqrt{n}|\theta-\theta_0|^2)
  + O_p(\sqrt{n}|\theta-\theta_0|^3).\label{def:psin_taylor-cf}
  \end{align}
  under the weaker Assumption \ref{assume:hat-cf}.
Since the number of folds is constant not changing with dimensionality, we may obtain
 \begin{align}
 &\frac{\sqrt{n}}{k}\sum_{j=1}^{k}\frac{1}{|\fold[j]|}\sum_{i\in\fold[j]}\scorei_i\left(\theta;\hbj,\hlamj(\cdot,\theta),\hgrj\right)  \notag\\
   =& n^{-1/2}\sum_{i=1}^n\scorei_i(\theta_0;\bbeta_0,\Lambda_0,\bgr_0)  - n^{-1/2}(\theta-\theta_0)\sum_{i=1}^n D_i\{1-\expit(\bgr_0^\top\bZI[i])\}
  (e^{\theta_0X_i}-1)/\theta_0 \notag\\
  & + o_p(1+\sqrt{n}|\theta-\theta_0|) + O_p(\sqrt{n}|\theta-\theta_0|^2)
  + O_p(\sqrt{n}|\theta-\theta_0|^3). \label{def:psin_taylor-cf-k}
\end{align}
   by summing over all the folds.
   According to \eqref{def:psin_taylor} and \eqref{def:psin_taylor-cf-k} and , the cross-fitted score
   admits the same asymptotic approximation as the one-shot score.

\subsubsection*{Step 2-4: consistency, normality and tightness}
The proofs are identical to the corresponding steps in the proof of Theorem \ref{thm:aalen}
because these Steps are the consequences of \eqref{def:psin_taylor}
and the shared Assumptions \ref{assume:Z}, \ref{assume:varD}
and \ref{assume:hlamlim}  (see also \eqref{eq:assume-varD}) concerning
the true distribution of the data.
Hence, we repeat those steps to obtain the asymptotic normality
$$
\sqrt{n}\sigma^{-1}(\hth_{cf}-\theta_0)
\leadsto N(0,1)
$$
with the asymptotic variance
 \begin{align}
   \sigma^2 = & (\sigma_\scorei/\E[D\{1-\expit(\bgr_0^\top\bZ_{1})\}
  (e^{\theta_0X}-1)/\theta_0])^2 \notag \\
  = & \frac{\Var\{n^{-1/2}\sum_{i=1}^{n}\scorei_i(\theta_0;\bbeta_0,\Lambda_0,\bgr_0)\}}
  {\E[D\{1-\expit(\bgr_0^\top\bZ_{1})\}
  (e^{\theta_0X}-1)/\theta_0])^2}\tag{\ref{def:sig}}.
 \end{align}

\subsubsection*{Step 5: variance estimation}
To show that the variance estimator
\begin{equation*}
  \hsig_{cf}^2 = \frac{n^{-1}\sum_{j=1}^{k} \sum_{i \in \fold[j]}\delta_i \{D_i-\expit(\hgr^{(j)\top}\bZI[i])\}^2e^{2\hth_{cf} D_i X_i}}
  {\left\{n^{-1}\sum_{j=1}^{k} \sum_{i \in \fold[j]} (1-D_i)\expit(\hgr^{(j)\top}\bZI[i])X_i\right\}^2}.
\end{equation*}
is consistent for the asymptotic variance $\sigma^2$ \eqref{def:sig}.
Like in the proof of Theorem \ref{thm:aalen}, we separately analyze the
numerator and the denominator of $\hsig_{cf}^2$.

We analyze the numerator of $\hsig_{cf}^2$ through the decomposition
\begin{align}
  & \frac{1}{n}\sum_{j=1}^k\sum_{i\in \fold[j]}\delta_i \{D_i-\expit(\hgr^{(j)\top}\bZI[i])\}^2e^{2\hth_{cf} D_i X_i} \notag \\
  =& \frac{1}{n}\sum_{j=1}^k\sum_{i\in \fold[j]}\left[\delta_i \{D_i-\expit(\hgr^{(j)\top}\bZI[i])\}^2(e^{2\hth D_i X_i} - 2e^{2\theta_0 D_i X_i}) \right]\notag \\
  & + \frac{1}{n}\sum_{j=1}^k\sum_{i\in \fold[j]}\left[\delta_i \{\expit(\bgr_0^\top\bZI[i])-\expit(\hgr^{(j)\top}\bZI[i])\}^2e^{2\theta_0 D_i X_i} \right]\notag \\
  & + \frac{1}{n}\sum_{i=1}^{n} \int_0^\tau
  [\{D_i-\expit(\bgr_0^\top\bZI[i])\}e^{\theta_0 D_i t}]^2dN_i(t) \notag \\
  = & Q^N_1 + Q^N_2 + Q^N_3.  \label{eq:hsig-num-cf}
\end{align}
We repeatedly use $\xi$ in all mean value theorem expansions
for convenience of notation.
By the mean value theorem, $Q^N_1$ in the righthand side of \eqref{eq:hsig-num-cf}
can be written in terms of $\theta_\xi = (1-\xi)\theta_0+\xi\hth_{cf}$
with some $\xi \in (0,1)$,
$$
  Q^N_1 =
  \frac{(\hth_{cf}-\theta_0)}{n}\sum_{j=1}^k\sum_{i\in \fold[j]} 2\delta_i D_iX_i \{1-\expit(\hgr^{(j)\top}\bZI[i])\}^2e^{2\theta_\xi D_i X_i}.
$$
Thus, we obtain the bound
\begin{equation*}
  |Q^N_1| \le |\hth_{cf}-\theta_0|\frac{1}{n}\sum_{j=1}^k\sum_{i\in \fold[j]} 2\tau e^{2\theta_\xi D_i X_i} = O_p(|\hth_{cf}-\theta_0|).
\end{equation*}
For $Q^N_2$, we analysis the term from each term separately,
$$
Q^N_{2,j} = \frac{1}{|\fold[j]|}\sum_{i\in \fold[j]}\left[\delta_i \{\expit(\bgr_0^\top\bZI[i])-\expit(\hgr^{(j)\top}\bZI[i])\}^2e^{2\theta_0 D_i X_i} \right], \,
Q^N_{2} = \sum_{j=1}^k \frac{|\fold[j]|}{n}Q^N_{2,j}.
$$
We bound each $Q^N_{2,j}$ through the expectation conditioning on the out-of-fold data
$$
\E(Q^N_{2,j}|\hgr^{(j)})
\le e^{2|\theta_0| \tau} \E\{\expit(\bgr_0^\top\bZI[i])-\expit(\hgr^{(j)\top}\bZI[i])|\hgr^{(j)}\}
=  e^{2|\theta_0| \tau} \mseg\left(\hgrj,\bgr_0\right).
$$
Under Assumption \ref{assume:hat-cf},
$\E(Q^N_{2,j}|\hgr^{(j)}) = o_p(1)$.
By the Markov's inequality, we obtain $Q^N_{2,j} = o_p(1)$.
Since the number of folds is constant, we may sum the rates up $Q^N_2 = o_p(1)$.
We have shown in \eqref{eq:hsig-Enum} as part of the proof of Theorem \ref{thm:aalen}
that $Q^N_3 = \sigma_\scorei^2 + o_p(1)$.
Putting our analyses of $Q^N_1$-$Q^N_3$ together,
we have shown the numerator of $\hsig^2$  equals $\sigma_\scorei^2 + o_p(1)$.

Similarly, we decompose the denominator of $\hsig_{cf}$ into
\begin{align*}
&\frac{1}{n}\sum_{j=1}^k\sum_{i\in \fold[j]}(1-D_i)\expit(\hgr^{(j)\top}\bZI[i])X_i \notag \\
= & \E[D\{1-\expit(\bgr_0^\top\bZ_{1})\}
  (e^{\theta_0X}-1)/\theta_0] \notag \\
& -\left(\E[D\{1-\expit(\bgr_0^\top\bZ_{1})\}
  \frac{e^{\theta_0X_i}-1}{\theta_0}]-\frac{1}{n}\sum_{i=1}^{n}D_i\{1-\expit(\bgr_0^\top\bZI[i])\}
  \frac{e^{\theta_0X_i}-1}{\theta_0}\right) \notag\\
& +
  \frac{1}{n}\sum_{j=1}^k\sum_{i\in \fold[j]}\left[D_i\{1-\expit(\hgr^{(j)\top}\bZI[i])\}
  \frac{e^{\hth X_i}-1}{\hth}-D_i\{1-\expit(\hgr^{(j)\top}\bZI[i])\}
  \frac{e^{\theta_0X_i}-1}{\theta_0}\right] \notag \\
  & +
  \frac{1}{n}\sum_{j=1}^k\sum_{i\in \fold[j]}\left[D_i\{\expit(\bgr_0^\top\bZI[i])-\expit(\hgr^{(j)\top}\bZI[i])\}
  \frac{e^{\theta_0 X_i}-1}{\theta_0}\right] \notag \\
= & Q^D_1 + Q^D_2 + Q^D_3 + Q^D_4
\end{align*}
$Q^D_1$ is the leading term.
$Q^D_2$ has been shown in \eqref{eq:aalen_Edpsi}
as diminishing to zero, $Q^D_2 = o_p(1)$.
We have established the bound of the expression $\frac{\theta e^{\theta X_i}X_i - e^{\theta X_i}+1}{\theta^2}$ in the Step 5 of the proof of Theorem \ref{thm:aalen},
$$
\sup_{\theta \in O(\theta_0,1)} \left|\frac{\theta e^{\theta X_i}X_i - e^{\theta X_i}+1}{\theta^2}\right|
\le \sup_{\theta \in O(\theta_0,1)} \frac{\theta e^{\theta \tau}\tau - e^{\theta \tau}+1}{\theta^2} < \infty,
$$
uniform over $O(\theta_0,1) \subset \R$, the radius one compact neighborhood of $\theta_0$.
Then for sufficiently large $n$, we may apply the mean value theorem bound for $Q^D_3$
\begin{align*}
  \left|Q^D_3\right|
 \le & \sup_{\theta_\xi in O(\theta_0,1)}\left|\frac{1}{n}\sum_{j=1}^k\sum_{i\in \fold[j]}D_i\{1-\expit(\hgr^{(j)\top}\bZI[i])\} \frac{\hth_{\xi}e^{\hth_\xi X_i}X_i - e^{\hth_\xi X_i}+1}{\hth^2_{\xi}}\right| |\hth-\theta_0| \\
 \le &|\hth-\theta_0| \sup_{\theta \in O(\theta_0,1)} \frac{\theta e^{\theta \tau}\tau - e^{\theta \tau}+1}{\theta^2}.
\end{align*}
This shows that $Q^D_3$ is of order $O_p(|\hth-\theta_0|)$.
Denote the terms for each fold in $Q^D_4$ as,
$$
Q^D_{4,j} = \frac{1}{|\fold[j]|}\sum_{i\in \fold[j]}\left[D_i\{\expit(\bgr_0^\top\bZI[i])-\expit(\hgr^{(j)\top}\bZI[i])\}
  \frac{e^{\theta_0 X_i}-1}{\theta_0}\right], \,
Q^D_{4} = \sum_{j=1}^k \frac{|\fold[j]|}{n}Q^D_{4,j}.
$$
By the triangle inequality, we establish a bound for each $Q^D_{4,j}$
$$
|Q^D_{4,j}| \le\frac{e^{|\theta_0| \tau}-1}{|\theta_0|} \sqrt{\frac{1}{|\fold[j]|}\sum_{i\in \fold[j]}
\{\expit(\bgr_0^\top\bZI[i])-\expit(\hgr^{(j)\top}\bZI[i])\}^2}.
$$
We have shown the righthand side above is of order $o_p(1)$
in our analysis of $Q^N_{2,j}$.
Thus, we obtain $Q_4=o_p(1)$ by summing the rates over all folds.
From our analysis of \eqref{eq:hsig-ddenom}, we have established
$$
n^{-1}\sum_{j=1}^k\sum_{i\in \fold[j]}(1-D_i)\expit(\hgr^{(j)\top}\bZI[i])X_i
= \E[D\{1-\expit(\bgr_0^\top\bZ_{1})\}
  (e^{\theta_0X}-1)/|\theta_0|] + o_p(1),
$$
which implies the denominator of the variance estimator $\hsig^2_{cf}$
converges to that of the asymptotic variance $\sigma^2$.

\section{Auxiliary Results}\label{appendix:auxiliary}
We state the auxiliary results in Appendices \ref{aux:prelim}-\ref{aux:other},
whose proofs are given in Appendix \ref{aux:proof}.
The results in Appendix \ref{aux:prelim} are technical preliminary steps in the proofs of the main results.
We state and prove them separately to promote the conciseness and readability of the proofs of the
main results.
Appendix \ref{aux:classical} contains the classical concentration equalities we use in our proofs.
We establish some new concentration results in Appendix \ref{aux:concentration}.
We put some minor but frequently used results in Appendix \ref{aux:other}.
The notations with letter $H$ are all generic and are replaced by suitable objects when we apply the results.

\subsection{Preliminary Results}\label{aux:prelim}

\begin{lemma}\label{propA:thm-aalen-taylor}
Under the Assumption \ref{assume:inf},
we have for $\theta$ in a compact neighborhood of $\theta_0$ such that $|\theta|\le \Kth$
 \begin{align*}
    &n^{-1/2}\sum_{i=1}^{n}\scorei_i\left(\theta;\hbeta,\hlam(\cdot,\theta),\hgr\right)  \\
   =&n^{-1/2}\sum_{i=1}^{n}\scorei_i(\theta_0;\bbeta_0,\Lambda_0,\bgr_0)  -\frac{1}{\sqrt{n}}(\theta-\theta_0) \sum_{i=1}^n D_i\{1-\expit(\bgr_0^\top\bZI[i])\}
  (e^{\theta_0X_i}-1)/\theta_0\\
  & + o_p(1+\sqrt{n}|\theta-\theta_0|) + O_p(\sqrt{n}|\theta-\theta_0|^2)
  + O_p(\sqrt{n}|\theta-\theta_0|^3).\tag{\ref{def:psin_taylor}}
  \end{align*}
\end{lemma}

\begin{lemma}\label{propA:unique-decomp}
Suppose the baseline cumulative hazard estimator $\hHaz$
satisfies the condition of Lemma \ref{lem:unique}
Under the Assumption \ref{assume:inf},
we have
\begin{equation*}
\sup_{\theta \in [-K_\theta,K_\theta]} \left|\frac{1}{\sqrt{n}}\sum_{i=1}^{n}\left[\scorei_i(\theta; \hbeta,\hHaz,\hgr)-\E\{\scorei_i(\theta; \bbeta_0,\overline{\Haz},\bgr_0)\}\right]\right|= o_p(1), \;
 \overline{\Haz}(t;\theta) = \Haz_0(t) + \int_0^t d_0(u) (\theta-\theta_0) du. \tag{\ref{eq:unique-decomp}}
\end{equation*}
\end{lemma}

\begin{lemma}\label{propA:thm-aalen-taylor-cf}
Suppose the $|\fold[j]|\asymp n$.
Under the Assumption \ref{assume:hat-cf},
we have for $\theta$ in a compact neighborhood of $\theta_0$ such that $|\theta|\le \Kth$
 \begin{align}
    &\frac{\sqrt{n}}{|\fold[j]|}\sum_{i\in\fold[j]}\scorei_i\left(\theta;\hbj,\hlamj(\cdot,\theta),\hgrj\right)  \notag\\
   =& \frac{\sqrt{n}}{|\fold[j]|}\sum_{i\in\fold[j]}\scorei_i(\theta_0;\bbeta_0,\Lambda_0,\bgr_0)  -\frac{\sqrt{n}}{|\fold[j]|}(\theta-\theta_0) \sum_{i \in \fold[j]} D_i\{1-\expit(\bgr_0^\top\bZI[i])\}
  (e^{\theta_0X_i}-1)/\theta_0 \notag\\
  & + o_p(1+\sqrt{n}|\theta-\theta_0|) + O_p(\sqrt{n}|\theta-\theta_0|^2)
  + O_p(\sqrt{n}|\theta-\theta_0|^3).\tag{\ref{def:psin_taylor-cf}}
  \end{align}
\end{lemma}

\subsection{Classical Concentration Inequalities}\label{aux:classical}
\begin{lemma}\label{lemma:Hoeffding}
\textbf{Hoeffding's Inequality} Theorem 2 p.4 in \cite{Hoeffding63}.
If $X_1,\dots,X_n$ are independent and $a_i \le X_i \le b_i$ $(i=1,2,\dots,n)$,
then for $t>0$
$$
\Pr(\bar{X}-\mu \ge t) \le \exp\left(-\frac{2n^2t^2}{\sum_{i=1}^n (b_i-a_i)^2}\right).
$$
\end{lemma}

\begin{lemma}\label{lemma:Azuma}
\textbf{A version of Azuma's Inequality } Theorem 1 p.3 and Remark 7 p.5 in \cite{Sason13}.
Let $\{X_k,\mathcal{F}_k\}_k=0^\infty$ be a discrete-parameter real-valued
martingale sequence such that for every $k$, the condition $|X_k-X_{k-1}|\le a_k$ holds almost surely
for some non-negative constants $\{a_k\}_{k=1}^\infty$. Then
$$
\Pr\left(\max_{k\in 1,\dots, n} |X_k - X_0| \ge t\right) \le 2 \exp\left(-\frac{t^2}{2\sum_{k=1}^n a_k^2}\right)
$$
\end{lemma}

\begin{lemma}\label{lemma:Bernstein}
\textbf{Bernstein Inequality for Sub-exponential Random Variables} Chapter 2 Sections 1.3 and 2.2 in \cite{wainwright2019}.
\begin{enumerate}[label = \alph*)]
\item For i.i.d. sample as in Chapter 2 Section 1.3 of \cite{wainwright2019}:

Let $X$ be a random variable with mean $\E(X)=\mu$.
If $X$ satisfies the Bernstein's condition with parameter $b$, i.e.
\begin{equation*}
  \left|\E\left\{(X-\mu)^k\right\}\right| \le \frac{1}{2}k! b^{k}, \; \text{for } k = 2,3,\dots,
\end{equation*}
the following concentration inequality holds for an i.i.d. sample $X_1,\dots, X_n$
\begin{equation*}
  \P\left(\left|\frac{1}{n}\sum_{i=1}^n X_i-\mu\right|\ge t\right) \le 2 \exp\left\{-\frac{n t^2}{2(b^2+ b t)}\right\}.
\end{equation*}
\item For martingale as in Chapter 2 Section 2.2 of \cite{wainwright2019}:
Let $M_1,\dots,M_n$ be a martingale series with respect to filtration $\mathcal{F}_1\subset\dots\subset\mathcal{F}_n$.
If the martingale differences satisfies the Bernstein's condition with parameter $b$, i.e.
\begin{equation*}
  \left|\E\left\{(M_{j+1}-M_j)^k|\mathcal{F}_j\right\}\right| \le \frac{1}{2}k! b^{k}, \; \text{for }
   j = 1,\dots,n-1
   \text{ and } k = 2,3,\dots,
\end{equation*}
the following concentration inequality holds
\begin{equation*}
  \P\left(\sup_{j=1,\dots,n} |M_j|\ge t\right) \le 2 \exp\left\{-\frac{ t^2}{2(nb^2+ b t)}\right\}.
\end{equation*}
\end{enumerate}
\end{lemma}

\begin{lemma}\label{lemma:DKW}
\textbf{Dvoretzky-Kiefer-Wolfowitz (DKW) Inequality} \citep{DKW56,massart1990}
Let $X_1,\dots,X_n$ be i.i.d. samples from a distribution with c.d.f. $F(x)$.
Define the empirical c.d.f. as $F_n(x) = n^{-1}\sum_{i=1}^{n}I(X_i \le x)$.
For any $\varepsilon>0$,
$$
\Pr\left(\sup_{x\in \R} \left|F_n(x)-F(x)\right|\right)\le 2e^{-2n\varepsilon^2}.
$$
\end{lemma}

\subsection{New Concentration Results}\label{aux:concentration}
All the concentration results are adapted to the cross-fitting scheme.
We repeated use the following two notations for index set and index set specific filtration.
\begin{definition}\label{def:fold}
We denote $\fold \subset \{1,\dots,n\}$ be a index set independent of observed data $\{W_i, i=1,\dots,n\}$
whose cardinality satisfies $|\fold| \asymp n$.
\end{definition}
\begin{definition}\label{def:FIt}
We define the filtration for index set $\fold$ as
$$\FIt = \sigma\left(\{N_i(u),Y_i(u+),D_i,Z_i:u\le t, i\in\fold \}\cup \{\delta_i,X_i,D_i,Z_i: i \in \fold^c\}\right).$$
\end{definition}
\begin{remark}\label{remark:FIt}
  The difference between $\FIt$ with $Y_i(u+)$ and the usual filtration defined with $Y_i(u)$ is that the former contains information about independent out of fold samples and
 the censoring times at present time $t$
so that the observed censoring times are stopping times with respect to $\FIt$.
On the other hand, we still have the
martingale property
\begin{equation*}
  \E\{M_i(t)|\mathcal{F}_{\fold, t-}\} =\E\{M_i(t)|\mathcal{F}^*_{\fold, t-}\} = M_i(t-)
\end{equation*}
because the extra censoring information at $t$ is not in $\mathcal{F}_{\fold, t-}$,
and out of fold samples are independent of $M_i(t)$ for $i \in \fold$.
\end{remark}

In the definition of model \eqref{model:aalen}, we implicitly assume the hazard function $\lambda_0(t)$
exists and is finite.
We denote its maximum as
\begin{equation}\label{def:Klam}
  \sup_{t\in [0,\tau]} \lambda_0(t) = K_\Lambda.
\end{equation}

\begin{lemma}\label{lemA:HbZ}
Define the filtration $\Fi = \sigma \left(\{N_i(u), Y_i(u),D_i,\bZ_i: u \le t\}\right)$.
Let $H_i(t)$ be a $\Fi[\tau]$-measurable random process, satisfying
$\P(\sup_{t\in[0,\tau]}|H_i(t)|<K_H) = 1$.
Under the model \eqref{model:aalen} (see also \eqref{def:Klam}),
we have for $x > K_H(K_\Lambda+\theta_0\vee 0)\tau$
\begin{equation*}
  \P\left(\int_0^\tau H_i(t)Y_i(t)\bbeta_0^\top\bZ_i dt>x\right) \le 2e^{-x/(2K_H)}
\end{equation*}
Moreover, we have
\begin{equation*}
\left|\int_0^\tau\E\{H_i(t)Y_i(t)\bbeta_0^\top\bZ_i\}dt\right| < 2K_H^2(K_\Lambda+\theta_0\vee 0)\tau +4K_H
\end{equation*}
and the concentration result
for all $\varepsilon \in \left[0, \sqrt{2}\right]$ and index set $\fold$ defined as in Definition \ref{def:fold}
\begin{equation*}
  \P\left(\left|\frac{1}{|\fold|}\sum_{i\in \fold} \int_0^\tau H_i(t)Y_i(t)\bbeta_0^\top\bZ_i dt-\int_0^\tau\E\{H_i(t)Y_i(t)\bbeta_0^\top\bZ_i\}dt\right| > K \varepsilon\right)
  < 4 e^{-|\fold|\varepsilon^2/2},
\end{equation*}
where $K = 2K_H(K_\Lambda+\theta_0\vee 0)\tau+2|\mu|+4K_H$.
\end{lemma}

\begin{lemma}\label{lemA:martdiff}
For an index set $\fold$ defined as in Definition \ref{def:fold}, we define the filtration
$\FIt$ as in Definition \ref{def:FIt}.
Let $M_i(t)$ be the martingale \eqref{def:mart} under model \eqref{model:aalen}
and $H_i(t)$ be a nonnegative $\FIt$-measurable random processes, satisfying
$\P(\sup_{t\in[0,\tau]}|H_i(t)|<K_H) = 1$.
Denote the order statistics of observed times as $X_{(1)},\dots, X_{(|\fold|)}$.
Then,
\begin{equation*}
  M^H_k = \frac{1}{|\fold|}\sum_{i\in\fold} \int_0^{X_{(k)}}H_i(t)dM_i(t),\, k=0,\dots,|\fold|
\end{equation*}
is a martingale with respect to $\FIt$,
and we have for $j \ge 2$
\begin{equation}\label{eq:mart_moment}
  \left|\E\left\{(M^H_k - M^H_{k-1})^j|\mathcal{F}^H_{k-1}\right\}\right|
  \le j! (2K_H/|\fold|)^j.
\end{equation}
Besides, for every $\varepsilon>K_H/\sqrt{|\fold|}$ we have
\begin{align}
&|\fold| \E\left\{(M^H_k - M^H_{k-1})^2; \sqrt{|\fold|}|M^H_k - M^H_{k-1}|>\varepsilon\right\} \notag\\
<&
(\varepsilon^2|\fold|+2K_H\sqrt{|\fold|}+2K_H^2) e^{-\varepsilon\sqrt{|\fold|}/K_H}.
\label{eq:mart_clt}
\end{align}
\end{lemma}

\begin{lemma}\label{lemA:mart}
For an index set $\fold$ defined as in Definition \ref{def:fold}, we define the filtration
$\FIt$ as in Definition \ref{def:FIt}.
Let $M_i(t)$ be the martingale \eqref{def:mart} under model \eqref{model:aalen}
and $H_i(t)$ be a $\FIt$-measurable random processes, satisfying
$\P(\sup_{t\in[0,\tau]}|H_i(t)|<K_H) = 1$.
Denote $X_{(1)},\dots, X_{(|\fold|)}$ be the order statistics of observed times.
For any $\varepsilon < 1$,
\begin{equation}
  \P\left(\left|\frac{1}{|\fold|}\sum_{i\in\fold} \int_0^\tau H_i(t)dM_i(t)\right| < 8K_H \varepsilon\right)
  > 1- 4 e^{-|\fold|\varepsilon^2/2}. \label{eq:lemA-mart-1}
\end{equation}
Moreover, we also have
\begin{equation}
 \bigvee_{t=0}^\tau \left\{\frac{1}{|\fold|}\sum_{i\in\fold} \int_0^t H_i(u)dM_i(u)\right\}
  \le \frac{1}{|\fold|}\sum_{i\in\fold} \bigvee_{t=0}^\tau \int_0^t H_i(u)dM_i(u)< 4K_H+8K_H \varepsilon \label{eq:lemA-mart-2}
\end{equation}
where $ \bigvee_{t=0}^\tau f(t)$ is the total variation of function $f(t)$ over $[0,\tau]$,
and
\begin{equation*}
  \sup_{t\in[0,\tau]}\left|\frac{1}{|\fold|}\sum_{i\in\fold} \int_0^t H_i(u)dM_i(u)\right| < 8K_H \varepsilon+2K_H/|\fold| 
\end{equation*}
whenever the event in \eqref{eq:lemA-mart-1} occurs.
\end{lemma}

\begin{lemma}\label{lemA:mart-rate}
For an index set $\fold$ defined as in Definition \ref{def:fold}, we define the filtration
$\FIt$ as in Definition \ref{def:FIt}.
Let $M_i(t)$ be the martingale \eqref{def:mart} under model \eqref{model:aalen}
and $H_i(t)$ be a $\FIt$-measurable random processes with tight supremum norm
$\max_{i=1,\dots,n}\sup_{t\in[0,\tau]}|H_i(t)| = O_p(1)$.
For any $\varepsilon < 1$,
\begin{equation}
 \left|\frac{1}{|\fold|}\sum_{i\in\fold} \int_0^\tau H_i(t)dM_i(t)\right|
 = O_p\left(n^{-\frac{1}{2}}\right). \label{eq:lemA-mart-rate}
\end{equation}
\end{lemma}

\begin{lemma}\label{lemA:HY}
Let $H_i$ be a random variable, satisfying
$\P\left(\sup_{i=1,\dots,n}|H_i|\le K_H\right) = 1$.
For an index set $\fold$ defined as in Definition \ref{def:fold}, we have the concentration result
\begin{equation*}
  \P\left(\sup_{t\in[0,\tau]}\left|\frac{1}{|\fold|}\sum_{i\in\fold} H_iY_i(t)-\E\{H_iY_i(t)\}\right| > 5K_H\varepsilon\right)
  < 8e^{-|\fold|\varepsilon^2/2}.
\end{equation*}
\end{lemma}

\begin{lemma}\label{lemA:ximart}
For an index set $\fold$ defined as in Definition \ref{def:fold}, we define the filtration
$\FIt$ as in Definition \ref{def:FIt}.
Let $M_i(t)$ be the martingale \eqref{def:mart} under model \eqref{model:aalen}
and $H_i(t)$ be $\FIt$-measurable random processes, satisfying
$\P(\sup_{t\in[0,\tau]}|H_i(t)|<K_H) = 1$.
Let $\mathcal{H}$ be a set of functions, potentially not $\FIt$-measurable, but satisfying
$\P\left(\sup_{\tilde{H}\in \mathcal{H}}\sup_{t\in[0,\tau]}|\tilde{H}(t)|<K_V\right) = 1$
and $\P\left(\sup_{\tilde{H}\in \mathcal{H}}\bigvee_0^\tau\tilde{H}(t) <K_V\right) = 1$,
where $\bigvee_0^\tau$ is the total variation on $[0,\tau]$.
Under Assumption \ref{assume:varD} (see also \eqref{eq:assume-varD}),
\begin{equation*}
  \P\left(\sup_{\tilde{H}\in \mathcal{H}}\left|\frac{1}{|\fold|}\sum_{i\in \fold}\int_0^\tau \tilde{H}(t)H_i(t)dM_i(t)\right|
  > 16K_HK_V\varepsilon +2K_HK_V/|\fold| \right)
  < 4e^{-|\fold|\varepsilon^2/2}.
\end{equation*}
\end{lemma}

\subsection{Other Auxiliary Results}\label{aux:other}

\begin{lemma}\label{lem:eY}
Under Assumption $C \indep (T,D)|Z$ \eqref{eq:CindD} and models \eqref{model:aalen}, or more general
partially linear additive risks model \eqref{model:aalen_pl},
we have
\begin{equation*}
\E[e^{D_i\theta_0 t}Y_i(t)|D_i,\bZ_i]= \E\{Y_i(t)|\bZ_i,D_i=0\}.
\end{equation*}
Under model \eqref{model:D},
\begin{equation}\label{eq:lemeY}
  \E[\{D_i-\expit(\bgr_0^\top\bZ_i)\}e^{D_i\theta_0 t}Y_i(t)] = 0
  \text{ and }
  \E[\{D_i-\expit(\bgr_0^\top\bZ_i)\}e^{D_i\theta_0 t}Y_i(t)\bZ_i] = \mathbf{0}.
\end{equation}
Moreover, we have for index set $\fold$ defined as in Definition \ref{def:fold}
under Assumption \ref{assume:Z},
\begin{align}
 & \sup_{t\in[0,\tau]}\left|\frac{1}{|\fold|}\sum_{i\in\fold}\{D_i-\expit(\bgr_0^\top\bZ_i)\}e^{D_i\theta_0 t}Y_i(t)\right| = O_p\left(n^{-\frac{1}{2}}\right) \text{ and} \notag \\
  &  \sup_{t\in[0,\tau]}\left\|\frac{1}{|\fold|}\sum_{i\in\fold}\{D_i-\expit(\bgr_0^\top\bZ_i)\}e^{D_i\theta_0 t}Y_i(t)\bZ_i\right\| = O_p\left(\sqrt{\frac{\log(p)}{n}}\right). \label{eq:rate-eY}
\end{align}
\end{lemma}

\begin{lemma}\label{lemA:wY}
Suppose model \eqref{model:D} is correct, and $\cgr$ is consistent for $\bgr_0$, i.e. $\mseg(\cgr,\bgr_0) = o_p(1)$.
For an index set $\fold$ defined as in Definition \ref{def:fold},
we have under Assumption \ref{assume:varD} (see also \eqref{eq:assume-varD})
\begin{equation}\label{eq:wY1}
  \lim_{n\to\infty}\P\left( \inf_{t\in[0,\tau]}\frac{1}{|\fold|}\sum_{i\in\fold}w^1_i(\cgr)Y_i(t)> e^{-\Kth\tau}\varepsilon_Y/2  \right) = 1
\end{equation}
\begin{equation}\label{eq:wY0}
\text{and }
  \lim_{n\to\infty}\P\left(
 \inf_{t\in[0,\tau]}\frac{1}{|\fold|}\sum_{i\in\fold}(1-D_i)\expit(\cgr^\top\bZI[i])Y_i(t)> \varepsilon_Y/2 \right) = 1.
\end{equation}
\end{lemma}

\subsection{Proofs of the Auxiliary Results}\label{aux:proof}
\begin{definition}\label{def:MVT-et}
By the mean value theorem for $e^{\theta t}-e^{\theta_0 t}$,
we have
\begin{equation}\label{eq:MVT-et}
  e^{\theta t}-e^{\theta_0 t} = (\theta-\theta_0) t e^{\theta_t t},
  \text{ for } \theta_t = \xi_t \theta_0 + (1-\xi_t)\theta \text{ with } \xi_t \in (0,1).
\end{equation}
\end{definition}
In a bounded set of $\theta$ such that $|\theta|<K_\theta$,
we have the bound $\sup_{t\in[0,\tau]}e^{\theta_t t} \le e^{K_{\theta} \tau}$.
Since $\theta_t$ depends on $\theta$, potentially estimated with all information from the data,
the process $e^{\theta_t t}$ is not necessarily $\Fjt$-adapted,
causing extra complication in our proof.

\begin{proof}[Proof of Lemma \ref{propA:thm-aalen-taylor}]
  We define the filtration as
$$\Ftn = \sigma\left(\{N_i(u),Y_i(u+),D_i,Z_i:u\le t, i = 1,\dots,n \}\right),$$ using $\fold = \{1,\dots,n\}$ in Definition \ref{def:FIt}.

We prove the statement \eqref{def:psin_taylor} by investigating each terms in the following
expansion,
\begin{eqnarray*}
&& n^{-1/2}\sum_{i=1}^{n}\scorei_i(\theta;\hbeta,\hlam(\cdot,\theta),\hgr) \notag \\
 &=& n^{-1/2}\sum_{i=1}^{n}\scorei_i(\theta;\bbeta_0,\Lambda_0,\bgr_0) \notag\\ 
 &&- n^{-\frac{1}{2}}\sum_{i=1}^n \{D_i-\expit(\bgr_0^\top\bZI[i])\}
   \int_0^\tau e^{D_i\theta t}Y_i(t)(\hbeta-\bbeta_0)^\top\bZ_idt \notag\\
 &&- n^{-\frac{1}{2}}\sum_{i=1}^n \{D_i-\expit(\bgr_0^\top\bZI[i])\}
   \int_0^\tau e^{D_i\theta t}Y_i(t)\left\{d\hlam(t,\theta)-d\Lambda_0(t)\right\} \notag\\  
&& - n^{-\frac{1}{2}}\sum_{i=1}^n \{\expit(\hgr^\top\bZI[i])-\expit(\bgr_0^\top\bZI[i])\}
 \int_0^\tau e^{D_i\theta t}dM_i(t;\theta,\bbeta_0,\Lambda_0)\notag\\  
&& + n^{-\frac{1}{2}}\sum_{i=1}^n \{\expit(\hgr^\top\bZI[i])-\expit(\bgr_0^\top\bZI[i])\}
 \int_0^\tau  e^{D_i\theta t}Y_i(t)(\hbeta-\bbeta_0)^\top\bZ_idt\notag\\
 && + n^{-\frac{1}{2}}\sum_{i=1}^n \{\expit(\hgr^\top\bZI[i])-\expit(\bgr_0^\top\bZI[i])\}
 \int_0^\tau  e^{D_i\theta t}Y_i(t)\left\{d\hlam(t,\theta)-d\Lambda_0(t)\right\}\notag\\
&=& Q_1+Q_2+Q_3+Q_4 +Q_5+Q_6.
\end{eqnarray*}
The first term $Q_1$ contains the leading terms. The rest $Q_2-Q_6$ are the remainders.

We expand $Q_1$ with respect to $\theta$ at $\theta_0$,
\begin{eqnarray*}
  Q_1&=& n^{-1/2}\sum_{i=1}^{n}\scorei_i(\theta_0;\bbeta_0,\Lambda_0,\bgr_0) \notag\\
  &&- n^{-\frac{1}{2}}(\theta-\theta_0) \sum_{i=1}^n \{D_i-\expit(\bgr_0^\top\bZI[i])\}
  \int_0^\tau e^{\theta_0t}D_i Y_i(t)dt \notag\\
    &&+ \frac{1}{\sqrt{n}}(\theta-\theta_0) \sum_{i=1}^n \{D_i-\expit(\bgr_0^\top\bZI[i])\}
  \int_0^\tau e^{\theta_0t}D_i t dM_i(t) \notag\\
  &&+ \frac{1}{\sqrt{n}}(\theta-\theta_0)^2 \sum_{i=1}^n \{D_i-\expit(\bgr_0^\top\bZI[i])\}
  \int_0^\tau e^{\theta_\xi t}D_i \{t^2 dM_i(t) + t Y_i(t)dt\} \notag\\
  &=& Q_{1,1} + Q_{1,2}+ Q_{1,3}+ Q_{1,4},
\end{eqnarray*}
where $Q_{1,4}$ comes from the mean value theorem for $e^{\theta t}-e^{\theta_0 t}$ \eqref{eq:MVT-et}.
 $Q_{1,1}$ is the leading term.
 Each summands in $Q_{1,2}$ is bounded
 by $e^{\theta_0\tau}$, so $Q_{1,2}$ is of order $O_p(\sqrt{n}|\theta-\theta_0|)$.
 Through an integral calculation, we have
 \begin{equation*}
   \int_0^\tau e^{D_i\theta_0t}D_iY_i(t)dt
   = D_i \int_0^{X_i} e^{\theta_0t}dt
   = D_i (e^{\theta_0X_i}-1)/\theta_0,
 \end{equation*}
 so we can write  $Q_{1,2}$ as
 \begin{equation*}
  - \frac{1}{\sqrt{n}}(\theta-\theta_0) \sum_{i=1}^n D_i\{1-\expit(\bgr_0^\top\bZI[i])\}
  (e^{\theta_0X_i}-1)/\theta_0.
 \end{equation*}
 In $Q_{1,3}$, we have a $\Ftn$-martingale
 \begin{equation}
   \frac{1}{n}\sum_{i=1}^n
  \int_0^\tau \{D_i-\expit(\bgr_0^\top\bZI[i])\} e^{D_i\theta_0t}D_i t dM_i(t), \label{term:psin_1_2-mart}
 \end{equation}
 whose integrand is bounded by $e^{\theta_0\tau}$.
 By Lemma \ref{lemA:mart}, \eqref{term:psin_1_2-mart} is
 of order $O_p(n^{-1/2})$.
 Hence, $Q_{1,3}$ is of order $O_p(|\theta-\theta_0|)=o_p(\sqrt{n}|\theta-\theta_0|)$.
 Note that we need to prove our statement uniformly in $\theta$, so
 we cannot directly utilize the martingale structure in $Q_{1,4}$
 \begin{equation}\label{term:psin_1_3-mart}
  \int_0^\tau e^{\theta_t t}\frac{1}{n}\sum_{i=1}^n \{D_i-\expit(\bgr_0^\top\bZI[i])\}D_i t^2 dM_i(t).
 \end{equation}
 Alternatively, we use Lemma \ref{lemA:ximart} to establish the rate of \eqref{term:psin_1_3-mart} as $O_p(n^{-1/2})$.
 The other term in $Q_{1,4}$
 \begin{equation*}
    \frac{1}{n}\sum_{i=1}^n \{D_i-\expit(\bgr_0^\top\bZI[i])\}
  \int_0^\tau e^{\theta_\xi t}D_i  t Y_i(t)dt
 \end{equation*}
 is bounded by $e^{K_\theta\tau} \tau$.
 Then, $Q_{1,4}$ is of order $O_p(\sqrt{n}|\theta-\theta_0|^2)$.
 Therefore, we have term $Q_1$ equals
 \begin{equation*}
   \frac{1}{\sqrt{n}}\sum_{i=1}^n \scorei_i(\theta_0;\bbeta_0,\Lambda_0,\bgr_0)
   - \frac{1}{\sqrt{n}}(\theta-\theta_0) \sum_{i=1}^n D_i\{1-\expit(\bgr_0^\top\bZI[i])\}
  (e^{\theta_0X_i}-1)/\theta_0
 \end{equation*}
 plus an $o_p(\sqrt{n}|\theta-\theta_0|)+O_p(\sqrt{n}|\theta-\theta_0|^2)$ error.

We expand $Q_2$ with respect to $\theta$,
\begin{eqnarray*}
Q_2&=&- n^{-\frac{1}{2}}\sum_{i=1}^n \{D_i-\expit(\bgr_0^\top\bZI[i])\}
   \int_0^\tau e^{D_i\theta_0 t}Y_i(t)(\hbeta-\bbeta_0)^\top\bZ_idt \notag\\
&&    - n^{-\frac{1}{2}}(\theta-\theta_0)\sum_{i=1}^n D_i\{1-\expit(\bgr_0^\top\bZI[i])\}
   \int_0^\tau e^{\theta_\xi t}Y_i(t)(\hbeta-\bbeta_0)^\top\bZ_idt \notag\\
&=& Q_{2,1}+Q_{2,2},
\end{eqnarray*}
where $Q_{2,2}$ comes from  the mean value theorem for $e^{\theta t}-e^{\theta_0 t}$ as in Definition \ref{def:MVT-et}.
By the H\"{o}lder's inequality, we have an bound for $Q_{2,1}$,
\begin{equation*}
  |Q_{2,1}| \le \sqrt{n} \tau\|\hbeta-\bbeta\|_1 \sup_{t\in[0,\tau]}\left\|\frac{1}{n}\sum_{i=1}^n \{D_i-\expit(\bgr_0^\top\bZI[i])\}e^{D_i\theta_0 t}Y_i(t)\bZ_i\right\|_\infty.
\end{equation*}
From Lemma \ref{lem:eY}, we have
\begin{equation*}
 \sup_{t\in[0,\tau]}\left\|\frac{1}{n}\sum_{i=1}^n \{D_i-\expit(\bgr_0^\top\bZI[i])\}e^{D_i\theta_0 t}Y_i(t)\bZ_i\right\|_\infty = O_p\left(\sqrt{\frac{\log(p)}{n}}\right).
\end{equation*}
Under Assumption \ref{assume:rate-inf}, we have $Q_{2,1} = O_p\left(\sqrt{\log(p)} \|\hbeta-\bbeta\|_1\right)= o_p(1)$.
We again apply the H\"{o}lder's inequality to find the upper bound for $Q_{2,2}$,
\begin{equation*}
  |Q_{2,2}| \le \sqrt{n}|\theta-\theta_0| \tau\|\hbeta-\bbeta\|_1 \sup_{t\in[0,\tau]}\left\|\frac{1}{n}\sum_{i=1}^n D_i \{1-\expit(\bgr_0^\top\bZI[i])\}e^{\theta_0 t}Y_i(t)\bZ_i\right\|_\infty.
\end{equation*}
Under Assumptions \ref{assume:Z} and \ref{assume:rate-inf},
we have $Q_{2,2} = O_p\left( \sqrt{n}|\theta-\theta_0| \|\hbeta-\bbeta\|_1 \right) = o\left( \sqrt{n}|\theta-\theta_0| \right)$.
Hence, term $Q_2=Q_{2,1}+Q_{2,2}$ is of order $o_p(\sqrt{n}|\theta-\theta_0|+1)$.

Very similar to our treatment of $Q_2$, we expand $Q_3$ with respect to $\theta$,
\begin{eqnarray*}
Q_3 &=& - \sqrt{n}
   \int_0^\tau \left[\frac{1}{n}\sum_{i=1}^n \{D_i-\expit(\bgr_0^\top\bZI[i])\}e^{D_i\theta_0 t}Y_i(t)\right]\left\{d\hlam(t,\theta)-d\hlam(t,\theta_0)\right\} \notag \\
   &&- \sqrt{n}
   \int_0^\tau  \left[\frac{1}{n}\sum_{i=1}^n \{D_i-\expit(\bgr_0^\top\bZI[i])\}e^{D_i\theta_0 t}Y_i(t)\right]\left\{d\hlam(t,\theta_0)-d\Lambda_0(t)\right\} \notag \\
   && - n^{-\frac{1}{2}}(\theta-\theta_0)\sum_{i=1}^n D_i\{1-\expit(\bgr_0^\top\bZI[i])\}
   \int_0^\tau t e^{\theta_\xi t}Y_i(t)\left\{d\hlam(t,\theta)-d\Lambda_0(t)\right\} \notag \\
   &=& Q_{3,1} + Q_{3,2}+ Q_{3,3},
\end{eqnarray*}
where $Q_{3,3}$ comes from  the mean value theorem for $e^{\theta t}-e^{\theta_0 t}$ as in Definition \ref{def:MVT-et}.
From Lemma \ref{lem:eY}, we know that,
\begin{equation*}
  \sup_{t\in[0,\tau]}\left|\frac{1}{n}\sum_{i=1}^n \{D_i-\expit(\bgr_0^\top\bZI[i])\}e^{D_i\theta_0 t}Y_i(t)\right|
  = O_p\left(n^{-\frac{1}{2}}\right).
\end{equation*}
Together with Assumption \ref{assume:hlamlim}, the integral $Q_{3,1}$ as an upper bound
\begin{equation*}
\sqrt{n}\sup_{t\in[0,\tau]}\left|\frac{1}{n}\sum_{i=1}^n \{D_i-\expit(\bgr_0^\top\bZI[i])\}e^{D_i\theta_0 t}Y_i(t)\right| \bigvee_{t=0}^\tau
  \left\{\hlam(t,\theta)-\hlam(t,\theta_0)\right\}
  = o_p(\sqrt{n}|\theta-\theta_0|).
\end{equation*}
We apply \eqref{eq:rate-hlam} in Assumption \ref{assume:rate-inf} to $Q_{3,2}$
and get $Q_{3,2} =  o_p(1)$.
By Helly-Bray argument \citep{Murphy94}, we have a bound for $Q_{3,3}$
\begin{equation*}
  |Q_{3,3}| \le  \sqrt{n}|\theta-\theta_0| \left\{\left|\hlam(\tau,\theta)-\Lambda_0(\tau)\right|\tau e^{\Kth \tau}
  + \int_0^\tau  \left|\hlam(t,\theta)-\Lambda_0(t)\right|d t e^{\theta_\xi t}\right\}.
\end{equation*}
Under Assumptions
\ref{assume:hlamlim} and \ref{assume:rate-inf},
our bound gives the rate $Q_{3,3} = o_p(\sqrt{n}|\theta-\theta_0|) + O_p(\sqrt{n}|\theta-\theta_0|^2)$.
Therefore, $Q_3 = Q_{3,1}+Q_{3,2}+Q_{3,3} = o_p(\sqrt{n}|\theta-\theta_0|+1)+ O_p(\sqrt{n}|\theta-\theta_0|^2)$.

In terms $Q_4-Q_6$, we have the model estimation error for the logistic regression.
By a mean value theorem argument, we have a uniform bound for the error
\begin{equation}\label{eq:error-pi}
  \left|\expit(\hgr^\top\bZI[i])-\expit(\bgr_0^\top\bZI[i])\right|
  \le \|\hgr-\bgr\|_1\sup_{i=1,\dots,n}\|\bZ_i\|_\infty
  = \le \|\hgr-\bgr\|_1 K_Z
\end{equation}
because the derivative of function $\expit(\cdot)$ is uniformly bounded by one.

We expand $Q_4$ with respect to $\theta$,
\begin{align*}
Q_4 =
& -n^{-\frac{1}{2}}\sum_{i=1}^n\int_0^\tau \{\expit(\hgr^\top\bZI[i])-\expit(\bgr_0^\top\bZI[i])\}
 e^{D_i\theta_0 t}dM_i(t)\notag\\
& -n^{-\frac{1}{2}}(\theta-\theta_0) \sum_{i=1}^n \int_0^\tau e^{\theta_\xi t}D_i\{\expit(\hgr^\top\bZI[i])-\expit(\bgr_0^\top\bZI[i])\}t dM_i(t)\notag \\
& +n^{-\frac{1}{2}}(\theta-\theta_0) \sum_{i=1}^n \{\expit(\hgr^\top\bZI[i])-\expit(\bgr_0^\top\bZI[i])\} \int_0^\tau Y_i(t)D_i e^{\theta_\xi t}(t\theta_\xi-t\theta_0+1)dt \notag \\
= & Q_{4,1}+Q_{4,2}+Q_{4,3},
\end{align*}
where $Q_{4,3}$ comes from  the mean value theorem for $e^{\theta t}-e^{\theta_0 t}$ as in Definition \ref{def:MVT-et}.
$\hgr$ is $\Ftn$-measurable,
so we can apply Lemma \ref{lemA:mart-rate} to $Q_{4,1}$.
According to \eqref{eq:error-pi} and Assumptions \ref{assume:Z} and \ref{assume:rate-inf}, $Q_{4,1} = O_p( \|\hgr-\bgr\|_1) = o_p(1)$.
For $Q_{4,2}$, we apply Lemma \ref{lemA:ximart} with $\mathcal{H}$ be the set of $\{e^{\theta_\xi t}:|\theta_\xi|\le \Kth\}$
to get $Q_{4,2} = O_p(|\theta-\theta_0|)$.
For $Q_{4,3}$, we use the uniform bound from \eqref{eq:error-pi}
\begin{equation*}
|Q_{4,3}| \le \sqrt{n}|\theta-\theta_0| \|\hgr-\bgr\|_1 K_Z
e^{\Kth \tau} (2\Kth \tau +1).
\end{equation*}
Under Assumption \ref{assume:Z} and \ref{assume:rate-inf}, $Q_{4,3} = o_p(\sqrt{n}|\theta-\theta_0|)$.
Therefore, we obtain $Q_4 = o_p(\sqrt{n}|\theta-\theta_0| + 1)$.

We apply the Cauchy-Schwartz inequality to $Q_5$,
$$
  |Q_5| \le  n^{-\frac{1}{2}}e^{\Kth \tau}\sqrt{\sum_{i=1}^n \{\expit(\hgr^\top\bZI[i])-\expit(\bgr_0^\top\bZI[i])\}^2}\sqrt{\sum_{i=1}^n \left\{(\hbeta-\bbeta_0)^\top\bZ_i\right\}^2X_i^2 }
  $$
Hence,
we have $$Q_5= O_p\left(\sqrt{n}\mseng\left(\hgr,\bgr_0\right)\msenb\left(\hbeta,\bbeta_0\right)\right),$$
which is $o_p(1)$ under Assumption \ref{assume:rate-inf}.

Similarly, we apply the Cauchy-Schwartz inequality to $Q_6$,
\begin{align*}
  |Q_6| \le  & n^{-\frac{1}{2}}e^{\Kth \tau}\sqrt{\sum_{i=1}^n \{\expit(\hgr^\top\bZI[i])-\expit(\bgr_0^\top\bZI[i])\}^2} \\
& \times \sqrt{\sum_{i=1}^n\left[\int_0^\tau  e^{D_i\theta t}Y_i(t)\left\{d\hlam(t,\theta)-d\Lambda_0(t)\right\}\right]^2}.
\end{align*}
Under Assumption \ref{assume:hlamTV}, we can apply
the Helly-Bray argument \citep{Murphy94} to find the bound,
\begin{align*}
  \left|\int_0^\tau  e^{D_i\theta t}Y_i(t)\left\{d\hlam(t,\theta)-d\Lambda_0(t)\right\}\right|
  \le & \left|e^{D_i\theta X_i}\left\{\hlam(X_i,\theta)-\Lambda_0(X_i)\right\}\right| \\
&  +
   \left|\int_0^{X_i} D_i \theta e^{\theta t}\left\{\hlam(t,\theta)-\Lambda_0(t)\right\}dt \right|.
\end{align*}
Hence, $Q_6 = O_p\left(\sqrt{n}\mseng\left(\hgr,\bgr_0\right)\sup_{t\in[0,\tau]}\left|\hlam(t,\theta)-\Lambda_0\right|\right)$,
which is $o_p(1+\sqrt{n}|\theta-\theta_0|)$ under Assumptions \ref{assume:hlamlim} and \ref{assume:rate-inf}.

Combining the results for $Q_1$-$Q_6$,  the proof is finalized.
\end{proof}

\begin{proof}[Proof of Lemma \ref{propA:unique-decomp}]

For the concentration of the score, we consider the following decomposition,
\begin{align}
  & n^{-1/2}\sum_{i=1}^{n}\scorei_i(\theta; \hbeta,\hHaz,\hgr)-\E\{\scorei_i(\theta; \bbeta_0,\overline{\Haz},\bgr_0)\} \notag \\
 = & \frac{1}{n}\sum_{i=1}^{n} \{D_i-\expit(\hgr^\top\bZI[i])\}
   \int_0^\tau e^{D_i\theta t}Y_i(t)(\bbeta_0-\hbeta)^\top\bZ_idt \notag\\
   & + \frac{1}{n}\sum_{i=1}^{n} \{D_i-\expit(\hgr^\top\bZI[i])\}
   \int_0^\tau e^{D_i\theta t}Y_i(t)\{d_0(t)-\hat{d}(t))D_i(\theta-\theta_0)dt \notag\\
 &+ \frac{1}{n}\sum_{i=1}^n\{ D_i -\expit(\hgr^\top\bZI[i])\}
   \int_0^\tau e^{\theta t}Y_i(t)\left\{d\Lambda_0(t)-d\hlam(t,\theta_0)\right\} \notag\\
 & - \frac{1}{n}\sum_{i=1}^n \{\expit(\hgr^\top\bZI[i])-\expit(\bgr_0^\top\bZI[i])\}
 \int_0^\tau e^{D_i\theta t} Y_i(t)\{D_i - d_0(t)\} (\theta-\theta_0) dt \notag \\
& - \frac{1}{n}\sum_{i=1}^n(1-D_i)\expit(\hgr^\top\bZI[i])
 \int_0^\tau dM_i(t;\theta_0,\bbeta_0,\Lambda_0)\notag\\  
& + \frac{1}{n}\sum_{i=1}^n D_i\{1-\expit(\hgr^\top\bZI[i])\}
 \int_0^\tau e^{\theta t}dM_i(t;\theta_0,\bbeta_0,\Lambda_0)\notag\\  
& + \left[-\frac{(\theta-\theta_0)}{n}\sum_{i=1}^n  \{D_i-\expit(\bgr_0^\top\bZI[i])\}
 \int_0^\tau e^{D_i\theta t} Y_i(t)\{D_i - d_0(t)\} dt
 - \E\{\scorei_i(\theta; \bbeta_0,\overline{\Haz},\bgr_0)\}\right] \notag \\
=& Q_1+Q_2+Q_3+Q_4 +Q_5+Q_6+Q_7. \label{decomp:propA-unique-decomp}
\end{align}
We prove the lemma by showing that $Q_1, \dots, Q_6$ are all asymptotically negligible
uniformly for $\theta \in [-K, K]$.

Under Assumptions \ref{assume:Z} and \ref{assume:rate-inf}, we have for $|\theta|<K$,
$$
  \sup_{|\theta|<K}|Q_1|
  \le \int_0^\tau e^{K \tau} \|\bbeta_0 - \hbeta\|_1 \sup_{i=1,\dots,n}\|Z_i\|_\infty dt
  = o_p(1).
$$

Under the assumption $\sup_{t\in [0,\tau]} |\hat{d}(t) - d_0(t)| = o_p(1)$ of the Lemma,
we have
$$
  \sup_{|\theta|<K}|Q_2| \le \int_0^\tau e^{K \tau} \sup_{t\in [0,\tau]} |\hat{d}(t) - d_0(t)| 2K dt = o_p(1).
$$



Since each $e^{D_i\theta t}$ is continuous and $\hHaz(t;\theta) - \Haz_0(t)$ is of bounded variation
under Assumption \ref{assume:hlamTV}, we may rewrite $Q_3$
\begin{align*}
Q_3 = & \frac{1}{n}\sum_{i=1}^n\{ D_i -\expit(\hgr^\top\bZI[i])\}
   e^{D_i\theta X_i}\left\{\Lambda_0(X_i)-\hlam(X_i,\theta_0)\right\} \\
&   -  \frac{1}{n}\sum_{i=1}^n\{ D_i -\expit(\hgr^\top\bZI[i])\} \int_{0}^{X_i}
\left\{\Lambda_0(X_i)-\hlam(X_i,\theta_0)\right\} D_i\theta e^{D_i\theta t}dt.
\end{align*}
Under Assumption \ref{assume:rate-inf}, we have
$$
  \sup_{|\theta|<K}|Q_3|
  \le e^{K \tau} \sup_{t\in[0,\tau]}\left|\Lambda_0(t)-\hlam(t,\theta_0)\right|
  +  \tau \sup_{t\in[0,\tau]}\left|\Lambda_0(t)-\hlam(t,\theta_0)\right| K e^{K \tau}
  = o_p(1).
$$

Since the expit-sigmoid function is Lipshitz,
$$
|\expit(x)-\expit(x')| \le \frac{1}{4}|x-x'|,
$$
we have under Assumptions \ref{assume:Z} and \ref{assume:rate-inf},
$$
  \sup_{|\theta|<K}|Q_4| \le \frac{1}{4}\|\hgr-\bgr_0\|_1 \sup_{i=1,\dots,n}\|Z_i\|_\infty
  \tau e^{K \tau} 2K  = o_p(1).
$$

Notice that $Q_5$ does not depend on $\theta$.
By Lemma \ref{lemA:mart}, we have
$$
\sup_{|\theta|<K}|Q_5| = |Q_5|  = O_p(n^{-1/2}) = o_p(1).
$$

Denote
$$
H(t) = \int_0^t\frac{1}{n}\sum_{i=1}^n D_i\{1-\expit(\hgr^\top\bZI[i])\}
 dM_i(u;\theta_0,\bbeta_0,\Lambda_0)
$$
Also by Lemma \ref{lemA:mart}, we have the uniform concentration of
$H(t)$ at zero
$$
\sup_{t\in[0,\tau]} H(t) = O_p(n^{-1/2}) = o_p(1),
$$
and $H(t)$
is of bounded variation with large probability.
On the event that $H(t)$ is of bounded variation, we may rewrite $Q_6$ as
$$
Q_6 = e^{K \tau}H(\tau) - \int_{0}^{\tau} H(t)e^{\theta t}\theta dt.
$$
Thus, we have the following bound for $Q_6$,
$$
\sup_{|\theta|<K}|Q_6| \le e^{K \tau} \sup_{t\in[0,\tau]} H(t)
+ \tau K e^{K \tau} \sup_{t\in[0,\tau]} H(t) = o_p(1).
$$

The order of the last term $Q_7$ is determined by the rate of concentration for the empirical process
$$
G_n(\theta) = \frac{1}{n}\sum_{i=1}^n  \{D_i-\expit(\bgr_0^\top\bZI[i])\}
 \int_0^\tau e^{D_i\theta t} Y_i(t)\{D_i - d_0(t)\} dt.
$$
Each summand in $G_n(\theta)$ is bounded by $\tau e^{K \tau}$.
For a fixed $\theta$, we can use the Hoeffding's inequality (Lemma \ref{lemma:Hoeffding})
to establish
$$
G_n(\theta) = \E\{G_n(\theta)\} + O_p(n^{-1/2}) = -\E\{\scorei_i(\theta; \bbeta_0,\overline{\Haz},\bgr_0)\}/(\theta-\theta_0) + o_p(1).
$$
To extend to the uniform concentration, we use the bracketing number argument.
Since $G_n(\theta)$ is Lipschitz in $\theta$,
$$
|G_n(\theta)-G_n(\theta')|
\le \int_0^\tau |e^{\theta t} - e^{\theta't}|dt
\le \tau e^{K \tau}|\theta -\theta'|,
$$
the bracketing number for $G_n(\theta)$ as a class index by $\theta$
is of order $O(1/\varepsilon)$.
Using bracketing number argument with the point-wise Hoeffding's bound,
we have
$$
\sup_{|\theta|<K}|G_n(\theta)-\E\{G_n(\theta)\})| = O_p(n^{-1/2}) = o_p(1).
$$
This leads to
$$
\sup_{|\theta|<K}|Q_7| \le K \sup_{|\theta|<K}|G_n(\theta)-\E\{G_n(\theta)\})| = o_p(1).
$$

We have shown that $Q_1,\dots, Q_7$ all converge to zero uniformly in $\theta$ with large probability.
Through the decomposition \eqref{decomp:propA-unique-decomp},
we have proved the lemma.

\end{proof}

\begin{proof}[Proof of Lemma \ref{propA:thm-aalen-taylor-cf}]
The proof of the lemma follows fundamentally the same strategy as that of Lemma \ref{propA:thm-aalen-taylor}.
The main difference is that we use the Cauchy Schwartz inequality instead of the H\"{o}lder's inequality to derive MSE type of bounds. We provide the details for completeness.

We define the filtration for the $j$-th fold as
$$\Fjt = \sigma\left(\{N_i(u),Y_i(u+),D_i,Z_i:u\le t, i\in\fold[j] \}\cup \{\delta_i,X_i,D_i,Z_i: i \in \fold[-j]\}\right),$$ using $\fold = \fold[j]$ in Definition \ref{def:FIt}.

We prove the statement \eqref{def:psin_taylor-cf} by investigating each terms in the following
expansion,
\begin{eqnarray*}
&& \frac{\sqrt{n}}{|\fold[j]|}\sum_{i\in\fold[j]}\scorei_i(\theta;\hbj,\hlamj(\cdot,\theta),\hgrj) \notag \\
 &=& \frac{\sqrt{n}}{|\fold[j]|}\sum_{i\in\fold[j]}\scorei_i(\theta;\bbeta_0,\Lambda_0,\bgr_0) \notag\\ 
 &&- \frac{\sqrt{n}}{|\fold[j]|}\sum_{i\in\fold[j]} \{D_i-\expit(\bgr_0^\top\bZI[i])\}
   \int_0^\tau e^{D_i\theta t}Y_i(t)(\hbj-\bbeta_0)^\top\bZ_idt \notag\\
 &&- \frac{\sqrt{n}}{|\fold[j]|}\sum_{i\in\fold[j]} \{D_i-\expit(\bgr_0^\top\bZI[i])\}
   \int_0^\tau e^{D_i\theta t}Y_i(t)\left\{d\hlamj(t,\theta)-d\Lambda_0(t)\right\} \notag\\  
&& - \frac{\sqrt{n}}{|\fold[j]|}\sum_{i\in\fold[j]} \{\expit(\hgrjt\bZI[i])-\expit(\bgr_0^\top\bZI[i])\}
 \int_0^\tau e^{D_i\theta t}dM_i(t;\theta,\bbeta_0,\Lambda_0)\notag\\  
&& + \frac{\sqrt{n}}{|\fold[j]|}\sum_{i\in\fold[j]} \{\expit(\hgrjt\bZI[i])-\expit(\bgr_0^\top\bZI[i])\}
 \int_0^\tau  e^{D_i\theta t}Y_i(t)(\hbj-\bbeta_0)^\top\bZ_idt\notag\\
 && + \frac{\sqrt{n}}{|\fold[j]|}\sum_{i\in\fold[j]} \{\expit(\hgrjt\bZI[i])-\expit(\bgr_0^\top\bZI[i])\}
 \int_0^\tau  e^{D_i\theta t}Y_i(t)\left\{d\hlamj(t,\theta)-d\Lambda_0(t)\right\}\notag\\
&=& Q_1+Q_2+Q_3+Q_4 +Q_5+Q_6. \label{terms:thm-aalen-taylor-cf}
\end{eqnarray*}
The first term $Q_1$ contains the leading terms. The rest $Q_2-Q_6$ are the remainders.

 Following exactly the same derivations in the proof of Lemma \ref{propA:thm-aalen-taylor}, we have term $Q_1$ equals
 \begin{equation*}
   \frac{\sqrt{n}}{|\fold[j]|}\sum_{i\in\fold[j]} \scorei_i(\theta_0;\bbeta_0,\Lambda_0,\bgr_0)
   - \frac{\sqrt{n}}{|\fold[j]|}(\theta-\theta_0) \sum_{i\in\fold[j]} D_i\{1-\expit(\bgr_0^\top\bZI[i])\}
  (e^{\theta_0X_i}-1)/\theta_0
 \end{equation*}
 plus an $o_p(\sqrt{n}|\theta-\theta_0|)+O_p(\sqrt{n}|\theta-\theta_0|^2)$ error.

We expand $Q_2$ with respect to $\theta$,
\begin{eqnarray*}
Q_2&=&- \frac{\sqrt{n}}{|\fold[j]|}\sum_{i\in\fold[j]} \{D_i-\expit(\bgr_0^\top\bZI[i])\}
   \int_0^\tau e^{D_i\theta_0 t}Y_i(t)(\hbj-\bbeta_0)^\top\bZ_idt \notag\\
&&    - (\theta-\theta_0)\frac{\sqrt{n}}{|\fold[j]|}\sum_{i\in\fold[j]} D_i\{1-\expit(\bgr_0^\top\bZI[i])\}
   \int_0^\tau e^{\theta_\xi t}Y_i(t)(\hbj-\bbeta_0)^\top\bZ_idt \notag\\
&=& Q_{2,1}+Q_{2,2},
\end{eqnarray*}
where $Q_{2,2}$ comes from  the mean value theorem for $e^{\theta t}-e^{\theta_0 t}$ as in Definition \ref{def:MVT-et}.
Denote
\begin{equation*}
  Q_{2,1,i} = \{D_i-\expit(\bgr_0^\top\bZI[i])\}
   \int_0^\tau e^{D_i\theta_0 t}Y_i(t)(\hbj-\bbeta_0)^\top\bZ_idt.
\end{equation*}
Using the independence across folds,
we can calculate the expectation for $i \in \fold[j]$
\begin{align*}
  &\E(Q_{2,1,i}) \notag\\
  = & \int_0^\tau\E(\hbj-\bbeta_0)^\top \E[\{D_i-\expit(\bgr_0^\top\bZI[i])\}e^{D_i\theta_0 t}Y_i(t)\bZ_i]dt
  \notag \\
  =& \int_0^\tau \E(\hbj-\bbeta_0)^\top \E[\E\{D_i-\expit(\bgr_0^\top\bZI[i])|\bZ_i\}\E\{e^{D_i\theta_0 t}Y_i(t)|D_i,\bZ_i\}\bZ_i]dt,
  \end{align*}
which equals zero by Lemma \ref{lem:eY}.
Hence, $\E(Q_{2,1}) = 0$.
We calculate the variance of $Q_{2,1}$
\begin{equation*}
\Var(Q_{2,1})
 = \frac{n}{|\fold[j]|^2}  \sum_{i \in \fold[j]} \E(Q_{2,1,i}^2) + \frac{2n}{|\fold[j]|^2}  \sum_{i< j, \{i,j\} \subset \fold[j]} \E(Q_{2,1,i}Q_{2,1,j}).
\end{equation*}
Note that we have
\begin{equation}\label{bound:bZ}
\left|\int_0^\tau e^{D_i\theta t}Y_i(t)(\hbj-\bbeta_0)^\top\bZ_idt\right|
\le e^{\Kth \tau} X_i \left|(\hbj-\bbeta_0)^\top\bZ_i\right|.
\end{equation}
Under Assumption \ref{assume:hat-cf},
\begin{equation*}
\frac{n}{|\fold[j]|^2} \sum_{i \in \fold[j]} \E(Q_{2,1,i}^2) \le
\frac{n}{|\fold[j]|} e^{2\Kth \tau}\tau\left\{\mseb\left(\hbj,\bbeta_0\right)\right\}^2
= O_p(r^{*2}_n) = o_p(1).
\end{equation*}
Using the independence across folds again, we
have
\begin{equation*}
\E(Q_{2,1,i}Q_{2,1,j}) = \E\{\E(Q_{2,1,i}|\hbj)\E(Q_{2,1,j}|\hbj)\} = 0.
\end{equation*}
Thus, we establish the rate $\Var(Q_{2,1}) = o_p(1)$.
By the Tchebychev's inequality, we have $Q_{2,1} = o_p(1)$.
 For $Q_{2,2}$, we denote
 \begin{equation*}
   Q_{2,2,i} = D_i\{1-\expit(\bgr_0^\top\bZI[i])\}
   \int_0^\tau e^{\theta_\xi t}Y_i(t)(\hbj-\bbeta_0)^\top\bZ_idt
 \end{equation*}
  apply Cauchy-Schwartz inequality to give an upper bound
 \begin{equation*}
|Q_{2,2}|
  \le \frac{\sqrt{n}}{|\fold[j]|}(\theta-\theta_0) \sqrt{\sum_{i\in\fold[j]} Q_{2,2,i}^2}.
 \end{equation*}
 Under Assumption \ref{assume:hat-cf},
 we have from bound \eqref{bound:bZ}
 \begin{equation*}
   \E\{Q^2_{2,2,i}\} \le e^{2\Kth \tau}\tau\left\{\mseb\left(\hbj,\bbeta_0\right)\right\}^2
=o_p(1).
 \end{equation*}
 Applying the Markov's inequality to $\sum_{i\in\fold[j]} Q_{2,2,i}^2$,
 we have  $Q_{2,2} = o_p(\sqrt{n}|\theta-\theta_0|)$.
Hence, term $Q_2$ is of order $o_p(\sqrt{n}|\theta-\theta_0|+1)$.

Very similar to our treatment of $Q_2$, we expand $Q_3$ with respect to $\theta$,
\begin{eqnarray*}
Q_3 &=& - \sqrt{n}
   \int_0^\tau \left[\frac{1}{|\fold[j]|}\sum_{i\in\fold[j]} \{D_i-\expit(\bgr_0^\top\bZI[i])\}e^{D_i\theta_0 t}Y_i(t)\right]\left\{d\hlamj(t,\theta)-d\hlamj(t,\theta_0)\right\} \notag \\
   &&- \frac{\sqrt{n}}{|\fold[j]|}\sum_{i\in\fold[j]} \{D_i-\expit(\bgr_0^\top\bZI[i])\}
   \int_0^\tau e^{D_i\theta_0 t}Y_i(t)\left\{d\hlamj(t,\theta_0)-d\Lambda_0(t)\right\} \notag \\
   && -  \frac{\sqrt{n}}{|\fold[j]|}(\theta-\theta_0)\sum_{i\in\fold[j]} D_i\{1-\expit(\bgr_0^\top\bZI[i])\}
   \int_0^\tau t e^{D_i\theta_\xi t}Y_i(t)\left\{d\hlamj(t,\theta)-d\Lambda_0(t)\right\} \notag \\
   &=& Q_{3,1} + Q_{3,2}+ Q_{3,3},
\end{eqnarray*}
where $Q_{3,3}$ comes from  the mean value theorem for $e^{\theta t}-e^{\theta_0 t}$ as in Definition \ref{def:MVT-et}.
From Lemma \ref{lem:eY}, we  have,
\begin{equation*}
  \sup_{t\in[0,\tau]}\left|\frac{1}{|\fold[j]|}\sum_{i\in\fold[j]} \{D_i-\expit(\bgr_0^\top\bZI[i])\}e^{D_i\theta_0 t}Y_i(t)\right|
  = O_p\left(n^{-\frac{1}{2}}\right).
\end{equation*}
Together with Assumption \ref{assume:hlamlim}, the integral $Q_{3,1}$ as an upper bound
\begin{equation*}
\sqrt{n}\sup_{t\in[0,\tau]}\left|\frac{1}{|\fold[j]|}\sum_{i\in\fold[j]} \{D_i-\expit(\bgr_0^\top\bZI[i])\}e^{D_i\theta_0 t}Y_i(t)\right| \bigvee_{t=0}^\tau
  \left\{\hlamj(t,\theta)-\hlamj(t,\theta_0)\right\}
  = o_p(\sqrt{n}|\theta-\theta_0|).
\end{equation*}
Denote
\begin{equation*}
  Q_{3,2,i} = \{D_i-\expit(\bgr_0^\top\bZI[i])\}
   \int_0^\tau e^{D_i\theta_0 t}Y_i(t)\left\{d\hlamj(t,\theta_0)-d\Lambda_0(t)\right\}.
\end{equation*}
Using the independence across folds,
we can calculate the expectation for $i \in \fold[j]$ according to Lemma \ref{lem:eY}
\begin{equation*}
  \E(Q_{3,2,i})
  =\int_0^\tau \E \left(\E[\{D_i-\expit(\bgr_0^\top\bZI[i])\}e^{D_i\theta_0 t}Y_i(t)|\bZ_i]\right)
 \left[d\E\left\{\hlamj(t,\theta_0)\right\}-d\Lambda_0(t)\right],
  \end{equation*}
which equals zero by Lemma \ref{lem:eY}.
Hence, $\E(Q_{3,2}) = 0$.
Moreover, we have a diminishing bound for $Q_{3,2,i}$ by Helly-Bray argument \citep{Murphy94}
under Assumption \ref{assume:hat-cf}
\begin{equation*}
  \max_{i\in\fold[j]}|Q_{3,2,i}| \le  \left|\hlamj(\tau,\theta_0)-\Lambda_0(\tau)\right|e^{\Kth \tau}
  + \int_0^\tau  \left|\hlamj(t,\theta_0)-\Lambda_0(t)\right|d e^{\theta_0 t}
 =o_p(1).
\end{equation*}
We denote $M_{3,2,m}= \frac{\sqrt{n}}{|\fold[j]|}\sum_{i \in\fold[j]^{1:m} }Q_{3,2,i}$ as the partial sum of the first $m$ terms in $Q_{3,2}$ whose indices are in $\fold[j]^{1:m}$.
It is a martingale with respect to filtration $\mathcal{F}_{3,2,m} = \sigma\left(\{W_i: i \in \fold[j]^{1:m}\cup\fold[-j]\}\right)$.
By the Azuma's inequality (as in Lemma \ref{lemma:Azuma}),
we have $Q_{3,2} = M_{3,2,|\fold[j]|} = o_p(1)$.
Similarly, we apply Helly-Bray argument \citep{Murphy94} to show that
\begin{equation*}
  |Q_{3,3}| \le  \sqrt{n}|\theta-\theta_0| \left\{\left|\hlamj(\tau,\theta)-\Lambda_0(\tau)\right|\tau e^{\Kth \tau}
  + \int_0^\tau  \left|\hlamj(t,\theta)-\Lambda_0(t)\right|d e^{\theta_\xi t}\right\}.
\end{equation*}
Under Assumptions
\ref{assume:hlamlim} and \ref{assume:hat-cf},
we have $Q_{3,3} = o_p(\sqrt{n}|\theta-\theta_0|) + O_p(\sqrt{n}|\theta-\theta_0|^2)$.
Therefore, $Q_3 = Q_{3,1}+Q_{3,2}+Q_{3,3} = o_p(\sqrt{n}|\theta-\theta_0|+1)+ O_p(\sqrt{n}|\theta-\theta_0|^2)$.

We expand $Q_4$ with respect to $\theta$,
\begin{align*}
Q_4 =
& -\frac{\sqrt{n}}{|\fold[j]|}\sum_{i\in\fold[j]} \{\expit(\hgrjt\bZI[i])-\expit(\bgr_0^\top\bZI[i])\}
 \int_0^\tau e^{D_i\theta_0 t}dM_i(t)\notag\\
& -\frac{\sqrt{n}}{|\fold[j]|}(\theta-\theta_0) \sum_{i\in\fold[j]} \int_0^\tau e^{\theta_\xi t}D_i\{\expit(\hgrjt\bZI[i])-\expit(\bgr_0^\top\bZI[i])\}t dM_i(t)\notag \\
& +\frac{\sqrt{n}}{|\fold[j]|}(\theta-\theta_0) \sum_{i\in\fold[j]} \{\expit(\hgrjt\bZI[i])-\expit(\bgr_0^\top\bZI[i])\} \int_0^\tau Y_i(t)D_i e^{\theta_\xi t}(t\theta_\xi-t\theta_0+1)dt \notag \\
= & Q_{4,1}+Q_{4,2}+Q_{4,3},
\end{align*}
where $Q_{4,3}$ comes from  the mean value theorem for $e^{\theta t}-e^{\theta_0 t}$ as in Definition \ref{def:MVT-et}.
Denote
\begin{equation*}
  Q_{4,1,i}(t) = \{\expit(\hgrjt\bZI[i])-\expit(\bgr_0^\top\bZI[i])\}
 \int_0^t e^{D_i\theta_0 t}dM_i(t).
\end{equation*}
Since $\{\expit(\hgrjt\bZI[i])-\expit(\bgr_0^\top\bZI[i])\}e^{D_i\theta_0 t}$
is $\Fjt$-adapted, each $Q_{4,1,i}(t)$ is $\Fjt$-martingales.
Then, $\E\{Q_{4,1}\} = 0$.
The optional quadratic variation of $\sum_{i\in \fold[j]}Q_{4,1,i}$ is
\begin{align*}
  \left[ \sum_{i\in \fold[j]}Q_{4,1,i}\right]_t = &
 \sum_{i\in \fold[j]} \{\expit(\hgrjt\bZI[i])-\expit(\bgr_0^\top\bZI[i])\}^2
 \int_0^t e^{2D_i\theta_0 t}dN_i(t) \\
 \le& \sum_{i\in \fold[j]} \{\expit(\hgrjt\bZI[i])-\expit(\bgr_0^\top\bZI[i])\}^2
 e^{2\Kth\tau}.
\end{align*}
Under Assumption \ref{assume:hat-cf},
we have $\E\{\expit(\hgrjt\bZI[i])-\expit(\bgr_0^\top\bZI[i])\}^2 =\left\{\mseg\left(\hgrj,\bgr_0\right)\right\}^2 = o_p(1)$.
Hence,
\begin{equation*}
  \Var(Q_{4,1}) = \frac{n}{|\fold[j]|^2} \sum_{i \in \fold[j]} \E\left\{\left[ \sum_{i\in \fold[j]}Q_{4,1,i}\right]_\tau\right\}
   = o_p(1).
\end{equation*}
We obtain $Q_{4,1} = o_p(1)$ by the Tchebychev's inequality.
For $Q_{4,2}$, we apply Lemma \ref{lemA:ximart} with $\mathcal{H}$ be the set of $\{e^{\theta_\xi t}:|\theta_\xi|\le \Kth\}$
to get $Q_{4,2} = O_p(|\theta-\theta_0|)$.
For $Q_{4,3}$, we apply the Cauchy-Schwartz inequality
\begin{equation*}
|Q_{4,3}| \le \frac{\sqrt{n}}{|\fold[j]|}|\theta-\theta_0| \sqrt{\sum_{i\in\fold[j]}\{\expit(\hgrjt\bZI[i])-\expit(\bgr_0^\top\bZI[i])\}^2 }
\sqrt{ne^{2\Kth \tau}(\Kth\tau+\tau)^2}.
\end{equation*}
Again with $\E\{\expit(\hgrjt\bZI[i])-\expit(\bgr_0^\top\bZI[i])\}^2 = O_p(q^*_n) = o_p(1)$,
we obtain from the Markov's inequality that
$\sum_{i\in\fold[j]}\{\expit(\hgrjt\bZI[i])-\expit(\bgr_0^\top\bZI[i])\}^2 = o_p(1)$.
Hence, $Q_{4,3} = o_p(\sqrt{n}|\theta-\theta_0|)$.
Therefore, we obtain $Q_4 = o_p(\sqrt{n}|\theta-\theta_0| + 1)$.

We apply the Cauchy-Schwartz inequality to $Q_5$,
\begin{equation*}
  |Q_5| \le \frac{\sqrt{n}}{|\fold[j]|}e^{\Kth \tau}\sqrt{\sum_{i\in\fold[j]} \{\expit(\hgrjt\bZI[i])-\expit(\bgr_0^\top\bZI[i])\}^2}\sqrt{\sum_{i\in\fold[j]} \left\{(\hbj-\bbeta_0)^\top\bZ_i\right\}^2X_i^2 }.
\end{equation*}
Using the independence across folds, we apply the Markov's inequality
to get $$Q_5= O_p\left(\sqrt{n}\mseg\left(\hgrj,\bgr_0\right)\mseb\left(\hbj,\bbeta_0\right)\right),$$
which is $o_p(1)$
under Assumption \ref{assume:hat-cf}.

Similarly, we apply the Cauchy-Schwartz inequality to $Q_6$,
\begin{align*}
  |Q_6| \le  & \frac{\sqrt{n}}{|\fold[j]|}e^{\Kth \tau}\sqrt{\sum_{i\in\fold[j]} \{\expit(\hgrjt\bZI[i])-\expit(\bgr_0^\top\bZI[i])\}^2} \\
& \times \sqrt{\sum_{i\in\fold[j]}\left[\int_0^\tau  e^{D_i\theta t}Y_i(t)\left\{d\hlamj(t,\theta)-d\Lambda_0(t)\right\}\right]^2}.
\end{align*}
Under Assumption \ref{assume:hlamTV}, we can apply
the Helly-Bray argument \citep{Murphy94} to find the bound,
\begin{align*}
  \left|\int_0^\tau  e^{D_i\theta t}Y_i(t)\left\{d\hlamj(t,\theta)-d\Lambda_0(t)\right\}\right|
  \le & \left|e^{D_i\theta X_i}\left\{\hlamj(X_i,\theta)-\Lambda_0(X_i)\right\}\right| \\
&  +
   \left|\int_0^{X_i} D_i \theta e^{\theta t}\left\{\hlamj(t,\theta)-\Lambda_0(t)\right\}dt \right|.
\end{align*}
Hence, $Q_6 = O_p\left(\sqrt{n}\mseg\left(\hgrj,\bgr_0\right)\sup_{t\in[0,\tau]}\left|\hlamj(t,\theta)-\Lambda_0\right|\right)$,
which is $o_p(1+\sqrt{n}|\theta-\theta_0|)$ under Assumptions \ref{assume:hlamlim} and \ref{assume:hat-cf}.

Plugging the rates for $Q_1$-$Q_6$ into the decomposition \eqref{terms:thm-aalen-taylor-cf}, we have
proven \eqref{def:psin_taylor-cf}.

\end{proof}

\begin{proof}[Proof of Lemma \ref{lemA:HbZ}]
We prove the result for nonnegative $H_i(t)$.
The general result can be obtained through decomposing $H_i(t)$
into the difference of two nonnegative processes
\begin{equation*}
H_i(t)=H_i(t)\vee 0 -[-\{H_i(t) \wedge 0\}]
\end{equation*}
and use the union bound with the result for the nonnegative processes.

  Under the model \eqref{model:aalen}, $\mu$ satisfies
  $\P(D_i\theta_0 + \bbeta_0^\top \bZ_i \ge -d\Lambda_0(t)) = 1$.
  By \eqref{def:Klam}, we have a lower bound
  $\P(\bbeta_0^\top \bZ_i > -K_\Lambda-\theta_0\vee 0) = 1 $.
  The $\bbeta_0^\top\bZ_i$  is potentially unbounded from above, so we have
  to study the bound for the upper tail.
  For $x > K_H(K_\Lambda+\theta_0\vee 0)\tau$,
  \begin{eqnarray*}
  &&\P\left(\int_0^\tau H_i(t)Y_i(t)\bbeta_0^\top\bZ_i dt>x\right) \notag \\
  &\le& \P\left(K_H X_i\bbeta_0^\top\bZ_i >x \right) \notag  \\
  &\le& \E\left[I(\bbeta_0^\top\bZ_i > K_\Lambda+\theta_0\vee 0 )I(C_i>x/K_H)\exp\left\{-\frac{x}{K_H}\frac{D_i\theta_0+\bbeta_0^\top\bZ_i}{\bbeta_0^\top\bZ_i} -\Lambda_0\left(\frac{x/K_H}{\bbeta_0^\top\bZ_i}\right)\right\}\right] \notag \\
  &\le& e^{-x/(2K_H)}.
  \end{eqnarray*}
  Denote $A_i = \int_0^\tau H_i(t)Y_i(t)\bbeta_0^\top\bZ_i dt$,
  $\mu = \int_0^\tau\E\{H_i(t)Y_i(t)\bbeta_0^\top\bZ_i\}dt$ and $K_A = K_H(K_\Lambda+\theta_0\vee 0)\tau$.
    First, we can find a bound for the expectation
  \begin{eqnarray*}
    |\mu| &=& \left|\int_0^\tau\E\{H_i(t)Y_i(t)\bbeta_0^\top\bZ_i\}dt\right| \notag \\
  &\le & \left|\int_0^\tau\E\{H_i(t)Y_i(t)\bbeta_0^\top\bZ_i I(|\bbeta_0^\top\bZ_i|<K_A)\}dt\right|
  + \left|\E\left\{\int_0^\tau H_i(t)Y_i(t)\bbeta_0^\top\bZ_i I(\bbeta_0^\top\bZ_i\ge K_A)dt\right\}\right| \notag \\
  &\le & K_H K_A \tau + \int_0^\infty \P\left(\int_0^\tau H_i(t)Y_i(t)\bbeta_0^\top\bZ_i I(\bbeta_0^\top\bZ_i\ge K_A)dt>x\right) dx \notag \\
  &\le& K_H K_A +2K_H.
  \end{eqnarray*}
  Then, we bound the centered moments for $k\ge 2$
  \begin{eqnarray*}
    \E(A_i-\mu)^k &=& \E\{(A_i-\mu)^kI(A_i < K_A+\mu\vee 0)\} + \E\{(A_i-\mu)^kI(A_i \ge K_A+\mu\vee 0)\} \notag \\
    &\le&  (K_A+|\mu|)^k + \int_0^\infty \P\{(A_i-\mu)^kI(A_i \ge K_A+\mu\vee 0)>x\}dx \notag \\
    &\le& (K_A+|\mu|)^k + \int_0^{(K_A-\mu\wedge 0)^k} \P(A_i \ge K_A+\mu\vee 0) dx \notag \\
    &&+ \int_{(K_A-\mu\wedge 0)^k}^\infty \P(A_i > x^{1/k}+\mu) dx \notag \\
    &\le& 2(K_A+|\mu|)^k + k!(2K_H)^{k} \notag \\
    &\le& k! (K_A+|\mu|+2K_H)^{k}
  \end{eqnarray*}
  Thus, $A_i$ is sub-exponential.
  By Bernstein inequality for sub-exponential random variables (as in Lemma \ref{lemma:Bernstein}), we have for
  any $\varepsilon \in [0, \sqrt{2}]$
  \begin{equation*}
    \P\left(\left|\frac{1}{|\fold|}\sum_{i\in\fold} A_i-\mu\right| > \varepsilon(K_A+|\mu|+2K_H)\right)
    < 2e^{-|\fold|\varepsilon^2/2}.
  \end{equation*}
  We thus complete the proof.
\end{proof}

\begin{proof}[Proof of Lemma \ref{lemA:martdiff}]
Let $X_{(1)},\dots, X_{(|\fold|)}$ be the order statistics of observed times.
Under filtration $\FIt$, they  are ordered stopping times (see Definition \ref{def:FIt} and Remark \ref{remark:FIt}).
By optional stopping theorem \citep{Durrett13},
we construct a discrete stopped martingale
\begin{equation*}
  M^H_k = \frac{1}{|\fold|}\sum_{i\in\fold} \int_0^{X_{(k)}}H_i(t)dM_i(t)
\end{equation*}
under filtration $\mathcal{F}^H_k = \sigma\{N_i(u),Y_i(u+),D_i,\bZ_i,X_{(k)}: u \in [0,X_{(k)}], i\in \fold\}$.
The increment of the discrete martingale has two components,
\begin{eqnarray}
  M^H_k - M^H_{k-1} &=& \frac{1}{|\fold|}\sum_{i\in\fold} H_i(X_{(k)}) dN_i(X_{(k)})  \notag \\
  &&- \frac{1}{|\fold|}\sum_{i\in\fold} Y_i(X_{(k-1)}) \int_{X_{(k-1)}}^{X_{(k)}} H_i(t)[ \{D_i\theta_0+\bbeta_0^\top\bZ_i\}dt + d\Lambda_0(t)], \label{decomp:lemA-mart}
\end{eqnarray}
one from the jumps of $N_i(t)$ and the other from the compensator.
Under model \eqref{model:aalen}, there is almost surely no ties in the observed event times, so
we have a bound
\begin{equation*}
  \P\left(\left|\frac{1}{|\fold|}\sum_{i\in\fold} H_i(X_{(k)}) dN_i(X_{(k)})\right| \le K_H/|\fold|\right)
  =\P\left(\frac{1}{|\fold|}\max_{i\in \fold} H_i(X_{(k)}) \le K_H/|\fold|\right) = 1.
\end{equation*}
The compensator term in \eqref{decomp:lemA-mart}, second on the right hand side, is potentially unbounded.
We have to study its tail distribution.
Conditioning on $\mathcal{F}^H_{k-1}$, we calculate the distribution of $X_{(k)}$ as
\begin{eqnarray*}
 && \P(X_{(k)} \ge X_{(k-1)}+x|\mathcal{F}^H_{k-1}) \notag \\
  &=&
  \prod_{i=1}^{|\fold|} \P(C_i\wedge T_i \ge X_{(k-1)}+x| C_i\wedge T_i \ge X_{(k-1)})^{Y_i(X_{(k-1)})} \notag \\
  &\le&\exp\left[- \sum_{i\in\fold}Y_i(X_{(k-1)})\{(D_i\theta_0+\bbeta_0^\top\bZ_i)x + \Lambda_0(X_{(k-1)}+x)-
  \Lambda_0(X_{(k-1)})\}\right].
\end{eqnarray*}
We denote the function in the exponential index as
\begin{equation*}
  A(x) = \sum_{i\in\fold}Y_i(X_{(k-1)})\{(D_i\theta_0+\bbeta_0^\top\bZ_i)x + \Lambda_0(X_{(k-1)}+x)-
  \Lambda_0(X_{(k-1)})\}.
\end{equation*}
Note that $A(x)$ is nondecreasing, so its inverse $A^{-1}(x)$ is well defined.
Next, we evaluate the tail distribution of the compensator term
\begin{eqnarray*}
 &&\P\left(\frac{1}{|\fold|}\sum_{i\in\fold} Y_i(X_{(k-1)}) \int_{X_{(k-1)}}^{X_{(k)}} H_i(t)[ \{D_i\theta_0+\bbeta_0^\top\bZ_i\}dt + d\Lambda_0(t)] \ge x \bigg|\mathcal{F}^H_{k-1}\right) \notag \\
 &\le& \P(K_H A(X_{(k)}-X_{(k-1)})/|\fold| \ge x) \notag \\
 & = & \P\{X_{(k)} \ge X_{(k-1)}+ A^{-1}(nx/K_H)\} \notag \\
 &\le & e^{-|\fold|x/K_H}.
\end{eqnarray*}
For $j\ge 2$, we calculate the moments
\begin{eqnarray*}
  &&  \left|\E\left\{(M^H_k - M^H_{k-1})^j|\mathcal{F}^H_{k-1}\right\}\right|  \notag \\
    &\le&  \Bigg[\E\left\{\left|\frac{1}{|\fold|}\sum_{i\in\fold} H_i(X_{(k)}) dN_i(X_{(k)})\right|^j\bigg|\mathcal{F}^H_{k-1}\right\}^{\frac{1}{j}} \notag \\
   && +\E\left\{\left|\frac{1}{|\fold|}\sum_{i\in\fold} Y_i(X_{(k-1)}) \int_{X_{(k-1)}}^{X_{(k)}} H_i(t)[ \{D_i\theta_0+\bbeta_0^\top\bZ_i\}dt + d\Lambda_0(t)] \right|^j \bigg|\mathcal{F}^H_{k-1}\right\}^{\frac{1}{j}}\Bigg]^j \notag \\
   &\le& \left[ \frac{K_H}{|\fold|}+ \left\{\int_0^\infty e^{-|\fold|x^{\frac{1}{j}}/K_H}dx\right\}^{\frac{1}{j}}\right]^j \notag \\
   &=& \left[ \frac{K_H}{|\fold|}+ \frac{K_H}{|\fold|}(j!)^{\frac{1}{j}}\right]^j  \notag \\
   &\le& j! (2K_H/|\fold|)^j.
\end{eqnarray*}
This statement above proves \eqref{eq:mart_moment}, the first conclusion of the lemma.

For $\varepsilon>K_H/\sqrt{|\fold|}$, event
\begin{equation*}
  \sqrt{|\fold|}|M^H_k - M^H_{k-1}|>\varepsilon
\end{equation*}
occurs only if the following event occurs,
\begin{eqnarray*}
&& \frac{1}{|\fold|}\sum_{i\in\fold} H_i(X_{(k)}) dN_i(X_{(k)}) + \varepsilon/\sqrt{|\fold|} \notag \\
  &<& \frac{1}{|\fold|}\sum_{i\in\fold} Y_i(X_{(k-1)}) \int_{X_{(k-1)}}^{X_{(k)}} H_i(t)[ \{D_i\theta_0+\bbeta_0^\top\bZ_i\}dt + d\Lambda_0(t)].
\end{eqnarray*}
We can bound
\begin{eqnarray*}
&&\E\left\{(M^H_k - M^H_{k-1})^2; \sqrt{|\fold|}|M^H_k - M^H_{k-1}|>\varepsilon\right\} \notag \\
&\le& \E\Bigg\{\left(\frac{1}{|\fold|}\sum_{i\in\fold} Y_i(X_{(k-1)}) \int_{X_{(k-1)}}^{X_{(k)}} H_i(t)[ \{D_i\theta_0+\bbeta_0^\top\bZ_i\}dt + d\Lambda_0(t)]\right)^2 \notag \\
&& \qquad \times I\left(\frac{1}{|\fold|}\sum_{i\in\fold} Y_i(X_{(k-1)}) \int_{X_{(k-1)}}^{X_{(k)}} H_i(t)[ \{D_i\theta_0+\bbeta_0^\top\bZ_i\}dt + d\Lambda_0(t)] > \varepsilon/\sqrt{|\fold|}\right)\bigg\} \notag \\
&\le& \frac{\varepsilon^2}{|\fold|}e^{-\varepsilon\sqrt{|\fold|}/K_H} + \int_{\varepsilon^2/|\fold|}^\infty e^{-|\fold|\sqrt{x}/K_H}dx \notag \\
&=& \frac{\varepsilon^2|\fold|+2K_H\sqrt{|\fold|}+2K_H^2}{|\fold|^2} e^{-\varepsilon\sqrt{|\fold|}/K_H}.
\end{eqnarray*}
This proves \eqref{eq:mart_clt}, the other conclusion of the lemma.
\end{proof}

\begin{proof}[Proof of Lemma \ref{lemA:mart}]
Without loss of generality, we again prove the result for the nonnegative $H_i(t)$.

Let $X_{(1)},\dots, X_{(|\fold|)}$ be the order statistics of observed times.
We define the sequence $M^H_k$, $k=1,\dots,n$, along
with filtration $\Fk = \mathcal{F}_{\fold,X_{(k)}}$,
as in Lemma \ref{lemA:martdiff}.
By Lemma \ref{lemA:martdiff}, $M^H_k$ is a $\Fk$-martingale
satisfying \eqref{eq:mart_moment},
so we can apply the Bernstein's inequality for martingale differences
(as in Lemma \ref{lemma:Bernstein}).
For $\varepsilon <1$, we have
\begin{eqnarray}
    \P\left(\sup_{k=1,\dots,|\fold|}\left|\frac{1}{|\fold|}\sum_{i\in\fold} \int_0^{X_{(i)}} H_i(t)dM_i(t)\right| > 4K_H \varepsilon\right)
    &=&
    \P\left(\sup_{k=1,\dots,|\fold|}|M^H_k| > 4K_H \varepsilon\right) \notag \\
  &<& 2 e^{-|\fold|\varepsilon^2/2}. \label{eq:lemA-mart-1+}
\end{eqnarray}
This proves \eqref{eq:lemA-mart-1}, the first result of the lemma.

The total variation of the integral with nonnegative $H_i(t)$'s can be written as
\begin{eqnarray*}
\bigvee_{t=0}^\tau \left\{\frac{1}{|\fold|}\sum_{i\in\fold} \int_0^t H_i(u)dM_i(u)\right\}
&=& \frac{1}{|\fold|}\sum_{i\in\fold} \bigvee_{t=0}^\tau \int_0^t H_i(u)dM_i(u) \notag \\
  &=&\frac{2}{|\fold|}\sum_{i\in\fold} \int_0^\tau H_i(u)dN_i(u) - \frac{1}{|\fold|}\sum_{i\in\fold} \int_0^\tau H_i(u)dM_i(u).
\end{eqnarray*}
Hence, \eqref{eq:lemA-mart-2} the second result of the lemma follows directly from the first result \eqref{eq:lemA-mart-1+}.

To find the bound of variation between $X_{(k-1)}$ and $X_{(k)}$,
simply consider that $H_i(t)$ is nonnegative while $dN_i(t)$ and $Y_i(t)\{(D_i\theta_0+\bbeta_0^\top\bZ_i)dt+d\Lambda_0(t)\}$
are nonnegative measures.
Hence, the extremal values in the intervals can be explicitly expressed as
\begin{equation*}
  \sup_{t\in [X_{(k-1)},X_{(k)})}
  \frac{1}{|\fold|}\sum_{i\in\fold} \int_0^t H_i(u)dM_i(u)
  = \frac{1}{|\fold|}\sum_{i\in\fold} \int_0^{X_{(k-1)}} H_i(u)dM_i(u)
  = M^H_{k-1},
\end{equation*}
and
\begin{equation*}
  \inf_{t\in [X_{(k-1)},X_{(k)})}
  \frac{1}{|\fold|}\sum_{i\in\fold} \int_0^t H_i(u)dM_i(u)
  = \frac{1}{|\fold|}\sum_{i\in\fold} \int_0^{X_{(k)}-} H_i(u)dM_i(u)
  = M^H_{k} - \frac{H_{i_k}(X_{(k)})}{|\fold|}.
\end{equation*}
Therefore,
\begin{equation*}
  \sup_{t\in[0,\tau]}\left|\frac{1}{|\fold|}\sum_{i\in\fold} \int_0^t H_i(u)dM_i(u)\right| \le \sup_{k=1,\dots,n}|M^H_k| + K_H/|\fold|.
\end{equation*}

For general $H_i(t)$, we simply decompose $H_i(t)=H_i^+(t)-H_i^-(t)$ and use the union bound.
\end{proof}

\begin{proof}[Proof of Lemma \ref{lemA:mart-rate}]
The proof uses the conclusion of Lemma \ref{lemA:mart}.
For any $\varepsilon > 0$, we can find $K_\varepsilon$ according to the tightness of $H_i(t)$
such that $\P\left(\max_{i=1,\dots,n}\sup_{t\in[0,\tau]}|H_i(t)|>K_\varepsilon\right) < \varepsilon/2$.
Define the truncated processes
$H_{i,\varepsilon} (t) = (-K_\varepsilon)\vee \{H_i(t)\wedge K_\varepsilon\}$,
which is still $\FIt$-adapted, as well as bounded by $K_\varepsilon$.
By Lemma \ref{lemA:mart}, we have
\begin{equation*}
  \P\left(\left|\frac{1}{|\fold|}\sum_{i\in\fold} \int_0^\tau H_{i,\varepsilon} (t)dM_i(t)\right| < 8K_\varepsilon \frac{\log(8/\varepsilon)}{\sqrt{|\fold|/2}}\right)
  > 1- \varepsilon/2.
\end{equation*}
Since $H_{i,\varepsilon} (t) = H_i(t)$ for all $i=1,\dots,n$ and $t\in[0,\tau]$
with probability at least $1-\varepsilon/2$,
we have
\begin{equation*}
  \P\left(\left|\frac{1}{|\fold|}\sum_{i\in\fold} \int_0^\tau H_i (t)dM_i(t)\right| < 8K_\varepsilon \frac{\log(8/\varepsilon)}{\sqrt{|\fold|/2}}\right)
  > 1- \varepsilon.
\end{equation*}
The last equation defines the rate in \eqref{eq:lemA-mart-rate}.
\end{proof}

\begin{proof}[Proof of Lemma \eqref{lemA:HY}]
Let $B_i$, $i\in\fold$, be independent Bernoulli random variables with rate $(H_i+K_H)/(2K_H)$.
By a simple calculation, we have the following empirical distribution for $B_iX_i$
\begin{equation*}
  \frac{1}{|\fold|}\sum_{i\in\fold} B_iY_i(t) = \frac{1}{|\fold|}\sum_{i\in\fold} I(B_iX_i \ge t)
  \text{ and }
  \E\{B_iY_i(t)\} = \frac{1}{2K_H}\E\{H_iY_i(t)\} + \frac{1}{2}\E\{Y(t)\}.
\end{equation*}
We decompose
\begin{eqnarray}
  \frac{1}{|\fold|}\sum_{i\in\fold} H_iY_i(t)-\E\{H_iY_i(t)\} &=&
  \frac{2K_H}{|\fold|}\sum_{i\in\fold} B_iY_i(t) - \E\{H_iY_i(t)\} -K_H\E\{Y(t)\} \notag \\
  && - \frac{K_H}{|\fold|}\sum_{i\in\fold}Y_i(t) + K_H\E\{Y(t)\} \notag \\
  && - \frac{2K_H}{|\fold|}\sum_{i\in\fold} \left(B_i-\frac{H_i+K_H}{2K_H}\right)Y_i(t). \label{decomp:lemA-HY}
\end{eqnarray}
Applying the Dvoretzky-Kiefer-Wolfowitz (DKW) inequality (as in Lemma \ref{lemma:DKW}) to the first
two terms in \eqref{decomp:lemA-HY}, we have
\begin{eqnarray*}
 && \P\left(\sup_{t\in[0,\tau]} \left|\frac{2K_H}{|\fold|}\sum_{i\in\fold} B_iY_i(t)- \E\{H_iY_i(t)\}
  - K_H\E\{Y(t)\}\right| > K_H\varepsilon\right) \le 2e^{-|\fold|\varepsilon^2/2} \\
&&\text{and } \P\left(\sup_{t\in[0,\tau]} \left|\frac{K_H}{|\fold|}\sum_{i\in\fold}Y_i(t)-K_H\E\{Y(t)\} \right| > K_H\varepsilon\right) \le 2e^{-|\fold|\varepsilon^2/2}.
\end{eqnarray*}
Denote $X_{(i)}$, $i=1,\dots,n$, as the order statistics of $X_i$'s.
We further decompose the third term in \eqref{decomp:lemA-HY} as
\begin{eqnarray}
  \frac{2K_H}{|\fold|}\sum_{i\in\fold} \left(B_i-\frac{H_i+K_H}{2K_H}\right)Y_i(X_{(k)})
  &=& \frac{2K_H}{|\fold|}\sum_{i\in\fold} \left(B_i-\frac{H_i+K_H}{2K_H}\right) \notag \\
  && - \frac{2K_H}{|\fold|}\sum_{i=1}^k \left(B_{(i)}-\frac{H_{(i)}+K_H}{2K_H}\right). \label{term:lemA-HY}
\end{eqnarray}
By the Hoeffding's inequality (as in Lemma \ref{lemma:Hoeffding}), we bound the first term in \eqref{term:lemA-HY}
\begin{equation*}
  \P\left(\left|\frac{2K_H}{|\fold|}\sum_{i\in\fold} \left(B_i-\frac{H_i-K_H}{2K_H}\right)\right|>K_H\varepsilon\right) < 2e^{-|\fold|\varepsilon^2/2}.
\end{equation*}
Let $(i)$ be the $i$-th element in fold $\fold$.
We define a filtration $\mathcal{F}^H_m = \sigma(\{(H_i,X_i): i\in \fold\}\cup \{B_{(i)}:i=1,\dots,m\})$
under which we have the following martingale
\begin{equation*}
M^H_m = \frac{2K_H}{|\fold|}\sum_{i=1}^m \left(B_{(i)}-\frac{H_{(i)}+K_H}{2K_H}\right).
\end{equation*}
By the Azuma's inequality (as in Lemma \ref{lemma:Azuma}), we have
\begin{equation*}
  \P\left(\left|M^H_{|\fold|}\right| > 2K_H\varepsilon \right) < 2 e^{-|\fold|\varepsilon^2/2}.
\end{equation*}
We finish the proof by putting the concentration inequalities together.
\end{proof}

\begin{proof}[Proof of Lemma \ref{lemA:ximart}]
By Lemma \ref{lemA:mart}, the probability that the event
\begin{equation*}
  \frac{1}{n}\sum_{i=1}^{n}\int_0^\tau H_i(u)dM_i(u) < 8K_H\varepsilon
\end{equation*}
is no less than $1-4e^{-n\varepsilon^2/2}$.
We shall show that
\begin{equation}
  \left|\frac{1}{n}\sum_{i=1}^{n}\int_0^\tau \tilde{H}(t)H_i(t)dM_i(t)\right|
  < 16K_HK_V\varepsilon +2K_HK_V/n \label{event:lemA-ximart}
\end{equation}
on such event.

By Lemma \ref{lemA:mart}, the following function
\begin{equation*}
  \frac{1}{n}\sum_{i=1}^{n}\int_0^t H_i(u)dM_i(u)
\end{equation*}
 has total variation bounded by $4K_H+8K_H\varepsilon$ on event \eqref{event:lemA-ximart}.
As a result, we can apply the Helly-Bray integration by parts \citep{Murphy94}
\begin{equation}\label{decomp:lemA-ximart}
\frac{1}{n}\sum_{i=1}^{n}\int_0^\tau \tilde{H}(t)H_i(t)dM_i(t)
 = \frac{\tilde{H}(\tau)}{n}\sum_{i=1}^{n}\int_0^\tau H_i(t)dM_i(t)
 -\int_0^\tau \left\{\frac{1}{n}\sum_{i=1}^{n}\int_0^t H_i(u)dM_i(u)\right\} d\tilde{H}(t).
\end{equation}
By Lemma \ref{lemA:mart}, both terms have bound on event \eqref{event:lemA-ximart}
\begin{align}
 & \left|\frac{\tilde{H}(\tau)}{n}\sum_{i=1}^{n}\int_0^\tau H_i(t)dM_i(t)\right| \le K_V \times 8K_H\varepsilon,  \\
 & \left|\int_0^\tau \left\{\frac{1}{n}\sum_{i=1}^{n}\int_0^t H_i(u)dM_i(u)\right\} d\tilde{H}(t)\right| \le K_V \times (8K_H\varepsilon + 2K_H/n).
\end{align}
Plugging in the upper bounds to \eqref{decomp:lemA-ximart} finish the proof.
\end{proof}

\begin{proof}[Proof of Lemma \ref{lem:eY}]
Since we assume that $T_i$ and $C_i$ are independent given $D_i$ and $\bZ_i$, we have
\begin{equation*}
  \E[Y_i(t)|D_i,\bZ_i] =\P(T_i\wedge C_i \ge t|D_i,\bZ_i)
  = \P(C_i \ge t | D_i,\bZ_i) \P(T_i \ge t | D_i,\bZ_i).
\end{equation*}
Under the assumption \eqref{eq:CindD}, the censoring time is independent of treatment given covariates, so
\begin{equation*}
  \P(C_i \ge t | D_i,\bZ_i) =\P(C_i \ge t |\bZ_i)
\end{equation*}
is $\sigma\{\bZ_i\}-$measurable.
Under model \eqref{model:aalen_pl},
\begin{equation*}
  \P(T_i \ge t | D_i,\bZ_i) = e^{\int_0^t \haz(t;D_i,\bZ_i)dt}=e^{-D_i\theta_0t-\int_0^t g_0(t;\bZ_i)dt}
  = e^{-D_i\theta_0t}\P(T_i \ge t | D_i=0,\bZ_i).
\end{equation*}
Therefore, we have the following representation
\begin{equation*}
  \E[e^{D_i\theta_0 t}Y_i(t)|D_i,\bZ_i] = \P(C_i \ge t |\bZ_i)e^{-\int_0^t g_0(t;\bZ_i)dt}
  = \E\{Y_i(t)|\bZ_i,D_i=0\},
\end{equation*}
which is obviously $\sigma\{\bZ_i\}-$measurable.
By the tower property of conditional expectation, we can calculate the expectations for
any $\sigma\{\bZ_i\}$-measurable random variable $U_i$ through
\begin{eqnarray*}
  &&\E[\{D_i-\expit(\bgr_0^\top\bZ_i)\}e^{D_i\theta_0 t}Y_i(t)U_i] \notag \\
  &=& \E[\{D_i-\expit(\bgr_0^\top\bZ_i)\}\E\{e^{D_i\theta_0 t}Y_i(t)|D_i,\bZ_i\}U_i] \notag \\
  &=&\E[\E\{D_i-\expit(\bgr_0^\top\bZ_i)|\bZ_i\}\E\{Y_i(t)|\bZ_i,D_i=0\}U_i] \notag \\
  &=& 0.
\end{eqnarray*}
We obtain the two equations in \eqref{eq:lemeY} by setting $U_i$ above as $1$ and $\bZ_i$, respectively.

To deliver the concentration result \eqref{eq:rate-eY}, we decompose
\begin{align*}
  \frac{1}{|\fold|}\sum_{i\in\fold} \{D_i-\expit(\bgr_0^\top\bZI[i])\} e^{D_i\theta_0 t}Y_i(t)
  = &\frac{1}{|\fold|}\sum_{i\in\fold} D_i\{1-\expit(\bgr_0^\top\bZI[i])\}Y_i(t) \\
  & -\frac{1}{|\fold|}\sum_{i\in\fold} (1-D_i)\expit(\bgr_0^\top\bZI[i])\}Y_i(t).
\end{align*}
Each coordinate of
\begin{equation*}
\frac{1}{|\fold|}\sum_{i\in\fold}  D_i\{1-\expit(\bgr_0^\top\bZI[i])\}Y_i(t)
\text{ and } \frac{1}{|\fold|}\sum_{i\in\fold} \expit(\bgr_0^\top\bZI[i])Y_i(t),
\end{equation*}
is bounded, so we can apply Lemma \ref{lemA:HY} to get
\begin{align}
   & \sup_{t\in[0,\tau]} \left|e^{\theta_0 t} \frac{1}{|\fold|}\sum_{i\in\fold}  D_i\{1-\expit(\bgr_0^\top\bZI[i])\}Y_i(t)-e^{\theta_0 t}\E\left[ D_i\{1-\expit(\bgr_0^\top\bZI[i])\}Y_i(t)\right] \right| = O_p\left(n^{-\frac{1}{2}}\right), \notag \\
   & \sup_{t\in[0,\tau]} \left|\frac{1}{|\fold|}\sum_{i\in\fold}  (1-D_i)\expit(\bgr_0^\top\bZI[i])\}Y_i(t)-\E\left[ (1-D_i)\expit(\bgr_0^\top\bZI[i])\}Y_i(t)\right] \right| = O_p\left(n^{-\frac{1}{2}}\right).\label{eq:DeY}
\end{align}
From \eqref{eq:lemeY}, we know that
\begin{equation}\label{eq:EDeY}
 e^{\theta_0 t} \E\left[D_i\{1-\expit(\bgr_0^\top\bZI[i])\}Y_i(t)\right] =
  \E\left[(1-D_i)\expit(\bgr_0^\top\bZI[i])Y_i(t)\right].
\end{equation}
Therefore, we have proved the first rate in \eqref{eq:rate-eY} by combining \eqref{eq:DeY} and \eqref{eq:EDeY}.
In the same way under Assumption \ref{assume:Z},
we have a concentration result from Lemma \ref{lemA:HY} for each coordinate of
$ \frac{1}{|\fold|}\sum_{i\in\fold} \{D_i-\expit(\bgr_0^\top\bZI[i])\} e^{D_i\theta_0 t}Y_i(t)\bZ_i$.
We take the union bound to obtain
 the second rate in \eqref{eq:rate-eY}.
\end{proof}

\begin{proof}[Proof of Lemma \ref{lemA:wY}]
We provide the proof for the first result \eqref{eq:wY1}.
The proof for the second result  \eqref{eq:wY0}  is identical.
Since the weights $w^1_i(\cgr)$ are nonnegative and $Y_i(t)$'s are non-increasing,
we have lower bound
\begin{equation*}
  \frac{1}{|\fold|}\sum_{i\in\fold}w^1_i(\cgr)Y_i(t)
  \ge \frac{1}{|\fold|}\sum_{i\in\fold}D_i\{1-\expit(\cgr^\top\bZI[i])\}Y_i(\tau).
\end{equation*}
it is sufficient to show
\begin{equation}\label{eq:lemA-wY}
  \lim_{n\to\infty}\P\left( \frac{1}{|\fold|}\sum_{i\in\fold}w^1_i(\cgr)Y_i(\tau)> \varepsilon_Y/2 \right) = 1.
\end{equation}
We decompose
\begin{eqnarray}
  \frac{1}{|\fold|}\sum_{i\in\fold}w^1_i(\cgr)Y_i(\tau)
  &=& \frac{1}{|\fold|}\sum_{i\in\fold}w^1_i(\bgr_0)Y_i(\tau) \notag \\
  &&- \frac{1}{|\fold|}\sum_{i\in\fold}D_i\{\expit(\cgr^\top\bZI[i])-\expit(\bgr_0^\top\bZI[i])\}Y_i(\tau).
\label{decomp:lemA-wY}
\end{eqnarray}
The first term in \eqref{decomp:lemA-wY} has expectation
bounded away from zero by Assumption \ref{assume:varD} (see also \eqref{eq:assume-varD})
\begin{equation*}
  \E\{w^1_i(\bgr_0)Y_i(\tau)\}
  = \E\{\Var(D_i|\bZI[i])e^{\theta_0 t}\E\{Y_i(\tau)|\bZI[i],D_i=0\}\} \ge e^{-\Kth \tau }\varepsilon_Y.
\end{equation*}
Since $w^1_i(\bgr_0)Y_i(\tau)$ are i.i.d. random variables in $[0,1]$, we have by Hoeffding's inequality (as in Lemma \ref{lemma:Hoeffding}),
\begin{equation*}
  \frac{1}{|\fold|}\sum_{i\in\fold}w^1_i(\bgr_0)Y_i(\tau)
  = \E\{\Var(D_i|\bZI[i])e^{\theta_0 t}\E\{Y_i(\tau)|\bZI[i],D_i=0\}\} + O_p(n^{-1/2})
  \ge e^{-\Kth \tau }\varepsilon_Y + o_p(1).
\end{equation*}
By the Cauchy-Schwartz inequality, we have the bound for the second term
in \eqref{decomp:lemA-wY},
\begin{eqnarray*}
 && \left|\frac{1}{|\fold|}\sum_{i\in\fold}D_i\{\expit(\cgr^\top\bZI[i])-\expit(\bgr_0^\top\bZI[i])\}Y_i(\tau)
  \right| \notag \\
  &\le&\sqrt{\frac{1}{|\fold|}\sum_{i\in\fold}\{\expit(\cgr^\top\bZI[i])-\expit(\bgr_0^\top\bZI[i])\}^2}.
\end{eqnarray*}
By the Markov's inequality, the bound above is of order
$O_p\left(\mseg(\cgr,\bgr_0)\right) = o_p(1)$.
Therefore, we have
\begin{equation*}
\frac{1}{|\fold|}\sum_{i\in\fold}w^1_i(\cgr)Y_i(\tau)
+o_p(1)
\ge \varepsilon_Y.
\end{equation*}
Hence, we obtain \eqref{eq:lemA-wY}, a sufficient condition for \eqref{eq:wY1}.
\end{proof}

\end{document}